\documentclass[11pt,a4paper]{article}
\usepackage{amsmath}
\usepackage{amssymb, amscd}
\usepackage{amsfonts}
\usepackage{graphicx}
\usepackage{a4wide}
\usepackage{microtype}
\usepackage{bbm}
\usepackage{slashed}
\usepackage{times}
\usepackage{amsthm}
\usepackage{mathrsfs}
\usepackage{hyperref}
\usepackage{color}

\makeatletter
\newcommand*\rel@kern[1]{\kern#1\dimexpr\macc@kerna}
\newcommand*\widebar[1]{%
  \begingroup
  \def\mathaccent##1##2{%
    \rel@kern{0.8}%
    \overline{\rel@kern{-0.8}\macc@nucleus\rel@kern{0.2}}%
    \rel@kern{-0.2}%
  }%
  \macc@depth\@ne
  \let\math@bgroup\@empty \let\math@egroup\macc@set@skewchar
  \mathsurround\z@ \frozen@everymath{\mathgroup\macc@group\relax}%
  \macc@set@skewchar\relax
  \let\mathaccentV\macc@nested@a
  \macc@nested@a\relax111{#1}%
  \endgroup
}
\makeatother

\usepackage{comment}
\usepackage[margin=2.7cm]{geometry}
\usepackage[all,cmtip]{xy}

\usepackage{marvosym}

\theoremstyle{plain}
\newtheorem{theo}{Theorem}[section]
\newtheorem{lem}[theo]{Lemma}
\newtheorem{propo}[theo]{Proposition}
\newtheorem{cor}[theo]{Corollary}

\theoremstyle{definition}

\numberwithin{equation}{section}

\def\bbR{\mathbb{R}}
\def\bbC{\mathbb{C}}

\def\M{\mathcal{M}}
\def\A{\mathcal{A}}

\def\E{\mathcal{E}}
\def\H{\mathcal{H}}

\def\L{\mathcal{L}}
\def\V{\mathcal{V}}
\def\W{\mathcal{W}}
\def\G{\mathcal{G}}
\def\PP{\mathcal{P}}
\def\OO{\mathcal{O}}

\def\Sol{\mathrm{Sol}}
\def\Ker{\mathrm{Ker}}

\def\id{\mathrm{id}}
\def\supp{\mathrm{supp}}

\def\vol{\mathrm{vol}_M^{}}
\def\vols{\mathrm{vol}_\Sigma^{}}
\def\sc{\mathrm{sc}}
\def\tc{\mathrm{tc}}

\def\1{\mathbbm{1}}

\def\GG{\mathfrak{G}}

\def\beq{\begin{equation}}
\def\eeq{\end{equation}}

\title{Quantization of the linearised Einstein-Klein-Gordon system on arbitrary backgrounds and the special case of perturbations in Inflation}

\author{Thomas-Paul Hack\vspace{4mm}\\{\small Dipartimento di Matematica}\\{\small Universit{\`a} degli Studi di Genova –- Via Dodecaneso 35, 16146 Genova, Italy.}\vspace{2mm}\\ {\footnotesize hack@dima.unige.it} }
\date{\today}

\begin{document}

\maketitle

\begin{abstract} We quantize the linearised Einstein-Klein-Gordon system on arbitrary on-shell backgrounds in a manifestly covariant and gauge-invariant manner. For the special case of perturbations in Inflation, i.e. on-shell backgrounds of Friedmann-Lema\^itre-Robertson-Walker type, we compare our general quantization construction with the standard approach to the quantum theory of perturbations in Inflation. We find that not all local quantum observables of the linearised Einstein-Klein-Gordon system can be split into local observables of scalar and tensor type as in the standard approach. However, we argue that this subclass of observables is sufficient for measuring perturbations which vanish at spatial infinity, which is in line with standard assumptions. Finally, we comment on a recent observation that, upon standard quantization, the quantum Bardeen potentials display a non-local behaviour and argue that a similar phenomenon occurs in any local quantum field theory. It is the hope of the author that the present work may constitute a bridge between the generally applicable and thus powerful framework of algebraic quantum field theory in curved spacetimes and the standard treatment of perturbations in Inflation.
\end{abstract}


\section{\label{sec_intro}Introduction}

The inflationary paradigm, see e.g. the monographs \cite{Ellis, Mukhanov:2005sc, Straumann:2005mz}, is by now an important cornerstone of modern cosmology. In the simplest models for Inflation, one assumes that a classical real Klein-Gordon field $\phi$ with a suitable potential $V(\phi)$, coupled to spacetime metric via the Einstein equations, drives a phase of exponential expansion in the early universe. After this phase, the universe respectively its matter-energy content is thought to be almost completely homogenised, whereby the quantized perturbations of the scalar field and the metric are believed to constitute the seeds for the small-scale inhomogeneities in the universe that we observe today.

Mathematically, this idea is usually implemented by considering the coupled Einstein-Klein-Gordon system on a Friedmann-Lema\^itre-Robertson-Walker (FLRW) spacetime. Given a suitable potential $V(\phi)$, this coupled system will have solutions which display the wanted exponential behaviour. In order to analyse the perturbations in Inflation, the Einstein-Klein-Gordon system is linearised and the resulting linear field theory is quantized on the background solution in the framework of quantum field theory in curved spacetimes. The theory of perturbations in Inflation thus constitutes one of the major applications of this framework; a general treatment of quantum field theory in curved spacetimes may be found e.g. in the monographs and reviews \cite{Benini:2013fia, Birrell:1982ix, Hollands:2014eia, Mukhanov:2007zz, Parker:2009uva, Wald:1995yp}.

However, a direct quantization of the linearised Einstein-Klein-Gordon system is potentially obstructed by the fact that this system has gauge symmetries. Thus the usual approach to the quantization of perturbations in Inflation, see e.g. \cite{Bardeen:1980kt, Ellis,  Mukhanov:1990me, Mukhanov:2005sc, Straumann:2005mz} and the recent work \cite{Eltzner:2013soa}, consists of first splitting the degrees of freedom of the perturbed metric into components which transform as scalars, vectors and tensors under the isometry group of the FLRW background, the Euclidean group. Subsequently, gauge-invariant linear combinations of these components as well as the perturbed scalar field are identified, which are then quantized in the standard manner. Thereby it turns out that the tensor components of the perturbed metric are manifestly gauge-invariant, whereas the vector components are essentially pure gauge and thus unphysical. The scalar perturbations instead are usually quantized in terms of the gauge-invariant Mukhanov-Sasaki variable, which is essentially a Klein-Gordon field with a time-dependent mass. In the recent work \cite{Eltzner:2013soa}, this choice of dynamical variable has been shown to be uniquely fixed by certain natural requirements. The relation between the quantized perturbations of the Einstein-Klein-Gordon system and the small-scale inhomogeneities in the present universe is usually established by relating the power-spectrum of the latter to the power spectrum of the former in several non-trivial steps, cf. e.g. \cite{Ellis, Mukhanov:1990me, Mukhanov:2005sc, Straumann:2005mz}. An approach which differs in the way this relation is made, and is closer to the spirit of stochastic gravity, may be found in the recent work \cite{Pinamonti:2013zba}.
 
The conceptual drawback of the standard approach to quantizing perturbations in Inflation is that this approach makes heavy use of the isometry group and the related preferred coordinate system of FLRW spacetimes and is thus inherently non-covariant. In that sense, it is a bottom-up approach, which is of course well-motivated by the fact that it allows one to make explicit computations. Notwithstanding, it seems advisable to check whether the same results can be obtained in a rather top-down approach, as this would provide a firm conceptual underpinning of the standard approach. Motivated by this, we develop the quantum theory of the linearised Einstein-Klein-Gordon system on arbitrary on-shell backgrounds, and with arbitrary potential $V(\phi)$ and non-minimal coupling to the scalar curvature $\xi$, in the first part of this work. In order to deal with the gauge symmetries of this system, we follow the ideas of \cite{DimockVector}, which deals with the gauge-invariant quantization of the vector potential on curved spacetimes. This approach was later used in \cite{Fewster:2012bj} for quantizing linearised pure gravity on cosmological vacuum spacetimes and axiomatised in \cite{HS} in order to encompass arbitrary (Bosonic and Fermionic) linear gauge theories on curved spacetimes. In contrast to the BRST/BV approach to quantum gauge theories, see e.g. \cite{Hollands:2007zg, Fredenhagen:2011mq}, and \cite{Brunetti:2013maa} for an application to perturbative pure quantum gravity on curved spacetimes, the formalism used here works without the introduction of auxiliary fields, at the expense of being applicable only to linear field theories. As this  method to quantize linear gauge field theories is very much in the spirit of algebraic quantum field theory, a framework usually not encountered in the discussion of perturbations of Inflation, we have outlined the essential idea by means of a toy model in Section \ref{sec_toymodel}, where also the relation to the basic idea behind the standard approach to perturbations in Inflation is drawn.

In the second part of this work, we consider the special case of FLRW backgrounds and compare the quantum theory of the linearised Einstein-Klein-Gordon system constructed in the first part to the standard quantization of perturbations in Inflation by studying the set of quantum observables in both constructions. Indeed, we find that set of quantum observables in the standard approach, which is spanned by local observables of scalar and tensor type, is contained in the set of observables obtained in the general construction, but strictly smaller. However, we also find that this discrepancy seems to be alleviated if one restricts to configurations of the linearised Einstein-Klein-Gordon system which vanish at spatial infinity, which apparently is a general assumption in the standard approach, see e.g. \cite{Makino:1991sg}. Namely, we argue that local observables of scalar and tensor type are sufficient for measuring this subset of configurations. As a by-product of our analysis, we comment on a recent observation in \cite{Eltzner:2013soa} that, upon standard quantization, the quantum Bardeen potentials do not commute at space-like separations in contrast to the quantum Mukhanov-Sasaki variable. We argue that this occurs due to the manifestly non-local character of the Bardeen potentials and that the occurrence of non-local observables in a local quantum field theory and their failure to satisfy local commutation relations is generic and not at variance with physical principles.
 
This paper is organised at follows. After introducing our notations and conventions in Section \ref{sec_notations}, we discuss the equations of motion and gauge symmetries of the linearised Einstein-Klein-Gordon system in Section \ref{sec_eometc}. Using a suitable field redefinition, we cast these equations into a form which satisfies the axioms of linear gauge theories in \cite{HS}, in order to be able to directly apply the quantization construction devised in this work. Before doing so, we briefly sketch the analysis of the Cauchy problem for the linearised Einstein-Klein-Gordon system in Section \ref{sec_Cauchy}. In Section \ref{sec_quantization}, we then apply the results of \cite{HS} in order to quantize the linearised Einstein-Klein-Gordon system on arbitrary on-shell backgrounds. To this avail, we first construct a suitable presymplectic space and then quantize it in a canonical manner. Afterwards, we  demonstrate that this presymplectic space can be constructed in an equivalent way, which is better suited for computations. Finally, we argue that the field redefinition introduced in the beginning is physically irrelevant. In the second part of the paper, we first review the classification of perturbations in Inflation in Section \ref{sec_classpert}. We then discuss the splitting of the equations of motions and gauge symmetries as well as the standard gauge-invariant potentials in Section \ref{sec_spliteom}, where we already obtain a few intermediate results necessary for the comparison of the the general and standard quantization approaches. The main results of this comparative analysis can be found in Section \ref{sec_inflationresults}, where we demonstrate that the set of observables in the standard approach is contained in, but strictly smaller than the set of observables in the general construction. Subsequently, we argue that this smaller set of observables is still separating on configurations which vanish at spatial infinity in Section \ref{sec_separability}, which also contains a proof that the presymplectic space in the general quantization construction is actually symplectic on FLRW backgrounds. We close the paper with a few comments on the non-locality of the Bardeen potentials in Section \ref{sec_noncomm}. The Appendix contains several technical details and proofs as well as a list of symbols.


\section{\label{sec_general}Quantization of the linearised Einstein-Klein-Gordon system on arbitrary backgrounds}
\label{sec_quantgen}

\subsection{Notation and conventions}
\label{sec_notations}
 In the following, $M$ is a four-dimensional smooth (infinitely often differentiable) manifold on which we consider only smooth Lorentzian metrics $g$ which render $(M,g)$ globally hyperbolic, see e.g. \cite{Bar:2007zz} for a definition and properties. In particular, $M$ is diffeomorphic to $\bbR \times \Sigma$, i.e. to a Cartesian product of ''time'' and ''space'', and contains smooth Cauchy surfaces diffeomorphic to $\Sigma$. \cite{Bernal:2004gm, Bernal:2005qf}. We denote the causal future (past) of a subset $\OO\subset M$ by $J^+(\OO)$, ($J^-(\OO)$) and call $\OO\subset M$ time-like compact if there exist two Cauchy surfaces $\Sigma^\pm$ of $(M,g)$ with $\Sigma^-\subset J^-(\Sigma^+)$ and $\Sigma^+\cap\Sigma^-=\emptyset$, such that $\OO\subset J^-(\Sigma^+)\cap J^+(\Sigma^-)$. Moreover, we say that $\OO\subset M$ is space-like compact if its intersection with any Cauchy surface of $(M,g)$ is compact. See e.g. \cite{Sanders:2012ac, Baernew} for a detailed discussion of such sets.
 
 We introduce the vector bundles over $M$ $\V:= \bigvee^2 T^*M \oplus \left(M\times \bbR\right)$, where $\bigvee$ denotes the symmetric tensor product, and $\W := TM$. The space of smooth sections of a vector bundle such as $\V$ will be denoted by $\Gamma(\V)$. Important subspaces of $\Gamma(\V)$ are $\Gamma_0(\V)$, $\Gamma_\tc(\V)$ and $\Gamma_\sc(\V)$ the space of smooth sections of compact, time-like compact and space-like compact support, respectively. 
 
 Our metric sign convention is mostly plus and the curvature tensor conventions used are $(\nabla_a\nabla_b-\nabla_b\nabla_a)v_c = R_{abc}^{\phantom{abc}d}v_d$, $R_{ab}=R_{acb}^{\phantom{abc}c}$, $R=R_{a}^{\phantom{a}a}$, where $\nabla$ denotes the Levi-Civita covariant derivative and raising and raising and lowering indices is always defined with respect to the background metric $g$. Finally, we use units in which $8\pi G=1$, $G$ being Newton's gravitational constant; this renders the Klein-Gordon field dimensionless.  

A list of the symbols which are used across different sections of the text may be found in the Appendix \ref{sec_listsymbols}. 

\subsection{Equations of motion and gauge invariance}
\label{sec_eometc}

We consider the following coupled system of a Lorentzian metric $g$ and a scalar field $\phi$ on a four-dimensional smooth manifold $M$.
\begin{align}\label{eq_fullcoupled}
G_{ab}(g)=T_{ab}(g,\phi)\\
-g^{ab}\left(\nabla_g\right)_a\left(\nabla_g\right)_b \phi + \xi R(g) \phi +  \partial_\phi V(\phi)=0.\notag
\end{align}
Here $G_{ab}=R_{ab}-\frac12 R g_{ab}$ is the Einstein tensor, 
\begin{align*}
T_{ab}=&\;(1-2\xi)\left(\nabla_a \phi\right) \nabla_{b}\phi-2\xi\phi\nabla_{a}\nabla_b\phi+\xi G_{ab}\phi^2\\&\;+g_{ab}\left\{2\xi \phi\nabla_c\nabla^c\phi+\left(2\xi-\frac12\right)\left(\nabla^c\phi\right)\nabla_{c}\phi-V\right\}
\end{align*}
 is the stress-energy tensor of the scalar field $\phi$ and $V(\phi)$ is an arbitrary potential, whereas $\xi$ denotes a possible non-minimal coupling of the scalar field to the scalar curvature. We have emphasised the dependence of the Levi-Civita covariant derivative as well as the dependence of the curvature tensors on the metric in \eqref{eq_fullcoupled} but will omit this dependence, as well as the dependence of $V$ on $\phi$ in the following. Note that the system \eqref{eq_fullcoupled} is not conformally invariant if $\xi=\frac16$ and $V=\lambda \phi^4$, i.e. not invariant under the transformation $g\mapsto g \Omega^2$, $\phi\mapsto \phi \Omega^{-1}$ for $\Omega: M\to (0,\infty)$\footnote{In the context of inflation driven by the scalar field $\phi$, conformal coupling $\xi=\frac16$ is disfavoured because the number of ``e-foldings'' $N_e$ is proportional to $\xi-\frac{1}{6}$, cf. \cite{Makino:1991sg}.}. As is well-known, one can absorb the non-minimal coupling $\xi\neq0$ by a re-definition of $g$ and $\phi$, see e.g. \cite{Futamase:1987ua, Salopek:1988qh, Makino:1991sg} and a short review in Section \ref{sec_EinsteinJordan}. However, we chose to perform our computations without this detour in order to keep them more transparent. Yet, we will use this idea as an inspiration when considering the special case of perturbations in inflation.
 
\subsubsection{The original equations of motion and gauge transformations}
 
 We are interested in the field theory defined by the linearisation of the Einstein-Klein-Gordon system \eqref{eq_fullcoupled}. To this avail, we consider a smooth one-parameter family $\lambda \mapsto\GG(\lambda):= (g(\lambda), \phi(\lambda))^T$ of smooth solutions to \eqref{eq_fullcoupled}, and define 
 $\Gamma:=(\gamma,\varphi)^T:=\frac{d}{d\lambda}\GG(\lambda)|_{\lambda=0}$, $\GG:=\GG(0)=(g(0), \phi(0))^T=:(g, \phi)^T$. In this work, we restrict to $g$ with the property that $(M,g)$ is globally hyperbolic. The perturbation $\Gamma$ is a smooth section of the vector bundle $\V$, $\Gamma\in \Gamma(\V)$, cf. Section \ref{sec_notations}. We introduce on such sections a symmetric and non-degenerate bilinear form by
 \beq\label{eq_bilinearV}\langle \Gamma_1, \Gamma_2\rangle_{\V} := \int\limits_M \vol \left(g^{ab} g^{cd}\gamma_{1,ac}\gamma_{2,bd}+\varphi_1\varphi_2\right)\eeq
 which is well-defined for all sections $\Gamma_1, \Gamma_2\in \Gamma(\V)$ with compact overlapping support. Analogously we introduce a symmetric and non-degenerate bilinear form on smooth sections $\varsigma_1, \varsigma_2\in \Gamma(\W)$ with compact overlapping support by 
\begin{equation}\label{eq_formW}\langle\varsigma_1, \varsigma_2\rangle_\W  := \int\limits_M \vol \;g^{ab}\varsigma_{1,a}\varsigma_{2,b}\,.\end{equation}

 For the ensuing discussion, we write the Einstein-Klein-Gordon equations in the form 
$$\Gamma(\V)\ni E=\begin{pmatrix}E^2_{ab}\\E^0\end{pmatrix}=\begin{pmatrix}\frac{1}{2}(G_{ab}-T_{ab})\\-\nabla_a\nabla^a \phi + \xi R \phi +  \partial_\phi V\end{pmatrix}=0\,.$$
 The factor of $\frac12$ in $E^2_{ab}$ follows from $E=0$ being the Euler-Lagrange equations obtained by varying the Einstein-Hilbert-Klein-Gordon action w.r.t to the field tuple $\GG=(g,\phi)^T$. The linearised equation of motion is obtained from 
 \begin{equation}\label{eq_deflinearisedeom}\frac{d}{d\lambda}E(\GG(\lambda))|_{\lambda=0}=:P \Gamma=0\,.\end{equation}
 This defines the partial differential operator (see \eqref{eq_originaleom} in the Appendix for the complete expression)
 
\begin{gather*}P:\Gamma(\V)\to\Gamma(\V)\qquad P=\begin{pmatrix}P_0 & P_2\\P_3 & P_1\end{pmatrix}\notag\\
(P_0 \gamma)_{ab}=\frac{1-\xi\phi^2}{4}\left(-\nabla^c\nabla_c \gamma_{ab} + 2 \nabla^c\nabla^{\phantom{c}}_{(a}\gamma_{b)c}-g_{ab}\nabla^c\nabla^d\gamma_{cd}-\nabla_a\nabla_b{\gamma}_{c}^{\phantom{c}c}+\right.\\
\left.+g_{ab}\nabla_c\nabla^c{\gamma}_{d}^{\phantom{c}d}\right)+\text{ lower derivative orders}\\
(P_2\varphi)_{ab}=\left(-\xi g_{ab}\phi \nabla_c\nabla^c+\xi \phi \nabla_a\nabla_b\right)\varphi+\text{ lower derivative orders}\notag\\
P_3 \gamma = \left(\xi\phi \nabla^c\nabla^d-\xi\phi\nabla_a\nabla^a g^{cd}\right) \gamma_{cd}+\text{ lower derivative orders}\notag\\
P_1\varphi=\left(-\nabla_c\nabla^c+\xi R + \partial^2_\phi V\right)\varphi\notag\,.
\end{gather*} 
 
 Expanding the Einstein-Hilbert-Klein-Gordon action 
$$S(\GG):=\int \limits_M \vol\left(\frac{R}{2}-\frac{(\nabla\phi)^2}{2}-\frac{\xi\phi^2 R}{2}-V\right)$$
with respect to the perturbation $\Gamma$ up to second order, one obtains
\begin{equation}\label{eq_quadraticaction}S(\GG+\Gamma) = S(\GG)-\langle E,\Gamma\rangle_V- \frac12 \langle \Gamma, (P+A) \Gamma\rangle_\V+{\cal O}(\Gamma^3)=:S^{(2)}(\Gamma)+{\cal O}(\Gamma^3)\,,\end{equation}
with $$A:=\begin{pmatrix}
-2 {E}_{(a}^{2\phantom{d}c} {g}_{b)}^d+\frac12 E^2_{ab}g^{cd}& 0\\
\frac12 E_0 g^{cd}& 0
\end{pmatrix}.$$
 The Euler-Lagrange equations of $S^{(2)}(\Gamma)$ (w.r.t. compactly supported variations of $\Gamma$) are
 $$\frac{1}{2}\left(P^{\dagger}+P+A^\dagger+A\right)\Gamma=(P+A)\Gamma=-E$$
 where the formal adjoint $P^{\dagger}$ of $P$ is defined by 
 $$\langle P^{\dagger}\Gamma_1, \Gamma_2\rangle_\V:=\langle \Gamma_1, P^{}\Gamma_2\rangle_\V$$
 for all sections $\Gamma_1, \Gamma_2\in \Gamma(\V)$ with compact overlapping support and $A^\dagger$ is defined analogously. From the formal selfadjointness of $P+A$ we can infer
 $$\left(P^{\dagger}-P^{}\right)=\left(A^{}-A^{\dagger}\right)=\begin{pmatrix}
 \frac12\left(E^2_{ab}{g}^{cd}-g^{\phantom{2}}_{ab}{E^2}^{cd}\right)& -\frac12 E^0 g_{ab}\\
 \frac12 E^0 g^{cd} & 0
 \end{pmatrix}.$$
Thus, $P$ is formally selfadjoint and the equation \eqref{eq_deflinearisedeom} is the Euler-Lagrange equation of $S^{(2)}(\Gamma)$ if and only if the background fields $(g,\phi)$ are on-shell. The phenomenon that the linearised equation of motion operator $P$ is neither formally selfadjoint nor the Euler-Lagrange operator of the quadratic term in $S^{(2)}(\Gamma)$ is attributable to the metric dependence of $\langle \cdot,\cdot\rangle_\V$ \eqref{eq_bilinearV} via the volume element $\vol=\sqrt{|\det g|} \,d^4x$ and the inverse metric. Indeed, linearising $\sqrt{|\det g|}(g^{ac}g^{bd}E^2_{cd},E^0)^T$ rather than $E$ leads to the ``correct'' equation $(P+A)\Gamma=0$ also on off-shell backgrounds.

The Einstein-Hilbert-Klein-Gordon action is invariant under diffeomorphisms of $M$, thus we expect that the linearised theory is invariant under ''linearised diffeomorphisms''\footnote{In fact, \eqref{eq_PK} follows at first order in $\varsigma$ from the diffeomorphism invariance of $\langle E,\Gamma\rangle_V$, cf. \cite{Stewart:1974uz}, whereas the diffeomorphism invariance of $S(\GG)$ implies $K^\dagger E=0$.}, i.e. under the transformation
\beq\label{eq_gaugetrafosorig}\Gamma\mapsto \Gamma + \L_\varsigma \GG
= \Gamma + \begin{pmatrix}2 \nabla_{(a}\varsigma_{b)}\\\varsigma^a\nabla_a\phi\end{pmatrix}=:\Gamma + K\varsigma\,,\eeq
 where $\L$ denotes the Lie derivative, $_{(a\;b)}$ denotes idempotent symmetrisation in the indices $a, b$ and $\varsigma\in \Gamma(\W)$ (recall $\W=TM$). This defines a partial differential operator $K:\Gamma(\W)\to \Gamma(\V)$. Indeed, a necessary and sufficient condition for invariance of the quadratic action \eqref{eq_quadraticaction} w.r.t. this transformation for arbitrary compactly supported $\varsigma$ is 
 \begin{equation}\label{eq_gaugeinvariance}(P+A) \circ K = 0\qquad\text{and}\qquad K^{\dagger} \circ (P+A)\circ K = 0,\end{equation} with the adjoint $K^{\dagger}:\Gamma(\V)\to\Gamma(\W)$ of $K^{}$ being defined by $\langle K^{\dagger}\Gamma, \varsigma\rangle_\W:=\langle \Gamma, K^{}\varsigma\rangle_\V$ for arbitrary $\Gamma\in\Gamma(\V)$, $\varsigma\in\Gamma(\W)$ with compact overlapping support. One can either compute or argue via the chain rule for the Lie derivative that for arbitrary $\varsigma\in\Gamma(\W)$
 \beq \label{eq_PK}P K\varsigma = P\L_\varsigma \GG= \L_\varsigma E=\begin{pmatrix} \varsigma^c\nabla_c E^2_{ab}+2E^2_{c(a}\nabla^{\phantom{2}}_{b)}\varsigma^c\\\varsigma^c\nabla_c E^0\end{pmatrix}\eeq
and thus $P \circ K = 0$ if and only if $E^2_{ab}=0$ and $E^0$ is constant, in particular if the background fields $(g,\phi)$ are on-shell. In this case, \eqref{eq_gaugeinvariance} follows from the vanishing of $A$. Thus, in the following, we shall always assume that the background metric and scalar field are on-shell. While this is in line with the common approach to perturbation theory, it also assures that the linearised Einstein-Klein-Gordon equations are the Euler-Lagrange-equations of the action expanded to second order and that this expanded action is gauge-invariant w.r.t. compactly supported gauge transformations. This restriction to on-shell backgrounds is an artefact of truncating the expansion of the action $S(\GG+\Gamma)$ at second order in $\Gamma$ respectively the diffeomorphism $e^{\displaystyle\varsigma}$ at second order in $\varsigma$, in analogy to the pure gravity case, cf. \cite{Brunetti:2013maa}.
 
Hyperbolic properties of linear equations of motion in field theory are essential for two reasons. On the one hand, they guarantee causal propagation of initial data, such that value of a solution $\Gamma(x)$ of a hyperbolic equation at a point $x\in M$ depends only on the values of $\Gamma$ in the causal past $J^-(x)$ of $x$. On the other hand they guarantee that the propagation of initial data is deterministic, i.e. unique solutions exist for given initial data on a Cauchy surface, in other words, the Cauchy problem for the equations of motion is well-posed. From $P \circ K = 0$ we can infer that $P$ is not hyperbolic as the equation $P \Gamma=0$ has solutions with compact support in time and thus can not have a well-posed Cauchy problem. Notwithstanding, on account of the gauge freedom of the theory one should rather look at gauge-equivalence classes of solutions; then it is sufficient to check whether each such class contains an element which solves a hyperbolic equation, i.e. whether a gauge-fixing exists which turns  $P \Gamma=0$ into a hyperbolic equation. 

\subsubsection{The redefined equations of motion and gauge transformations and their properties}

The original form of the operator $P$ is quite difficult to handle and to analyse in that respect because the principal symbol -- the coefficient of the second derivative -- is quite complicated. To cope with this, we introduce a field redefinition, i.e. a fibre-wise map on $\Gamma(\V)$, which is a generalisation of the usual ''trace-reversal'' map used in linearised pure gravity. 
\begin{equation}\label{eq_tracereversal}\widebar {\,\cdot\,}:\Gamma(\V)\mapsto \Gamma(\V)\qquad \Gamma=\begin{pmatrix}\gamma_{ab}\\\varphi\end{pmatrix}\mapsto\widebar {\Gamma}=\begin{pmatrix}\displaystyle\frac{\alpha}{4}\left(\gamma_{ab}-\frac12 g_{ab}{\gamma}_{c}^{\phantom{c}c}\right)+\frac{\xi\phi}{2} g_{ab}\varphi\\\displaystyle\left(1+\frac{2\xi^2\phi^2}{\alpha}\right)\varphi+\frac{\xi\phi}{2}{\gamma}_{c}^{\phantom{c}c}\end{pmatrix},\end{equation}
where
$$\alpha := 1-\xi\phi^2\qquad\beta := 1+\kappa\phi^2\qquad \kappa := \left(6\xi-1\right)\xi\,.$$
The inverse transformation is given by
\begin{equation}\label{eq_tracereversalinverse}\widebar {\,\cdot\,}^{-1}:\Gamma(\V)\mapsto \Gamma(\V)\end{equation}
$$\Gamma=\begin{pmatrix}\gamma_{ab}\\\varphi\end{pmatrix}\mapsto\widebar {\Gamma}^{-1}=\begin{pmatrix}\displaystyle
\frac{4}{\alpha}\left(\gamma_{ab}-\frac12 g_{ab}{\gamma}_{c}^{\phantom{c}c}+\xi \phi g_{ab}\frac{\alpha\varphi+2\xi\phi{\gamma}_{c}^{\phantom{c}c} }{2\beta}\right)\\\displaystyle
\frac{\alpha\varphi+2\xi\phi{\gamma}_{c}^{\phantom{c}c} }{\beta}
\end{pmatrix}\,.$$

This field redefinition and its inverse are always well-defined if $\xi=0$. For $\xi\neq0$ singularities occur if $\alpha=1-\xi\phi^2=0$ or $\beta=1+\left(6\xi-1\right)\xi\phi^2= 0$ somewhere on $M$. Without going into details as this would go beyond the scope of this work, we briefly argue why we exclude these cases. At points where $1-\xi\phi^2=0$ the Einstein tensor $G_{ab}$ cancels from the background equations \eqref{eq_fullcoupled} and thus these equations completely change their character. A somewhat weaker degeneracy occurs in the case where $1+\left(6\xi-1\right)\xi\phi^2= 0$, which implies that the Ricci scalar cancels from the Einstein equation, i.e. the first equation in \eqref{eq_fullcoupled}, and thus only the trace-free part of the Ricci tensor remains. One might expect that solutions of the Einstein-Klein-Gordon equations which display one of the two above-mentioned degeneracies are singular at the degenerate points, but we are not aware of any results in this direction. Thus, for simplicity, we choose to discard these presumably pathological cases and restrict to backgrounds where $1-\xi\phi^2\neq0$ and $1+\left(6\xi-1\right)\xi\phi^2\neq0$ on all $M$. Note that in the conformally coupled case, $\xi=\frac16$, the condition $1-\xi\phi^2\neq0$ is sufficient.

Using this field redefinition, we now define 
\beq \label{eq_def_redefinedq}{\widebar {P}}:=P\circ \widebar {\,\cdot\,}^{-1}\qquad \widebar {K}:=\widebar {\,\cdot\,}\circ K\qquad \langle \cdot ,\cdot\rangle_{\widebar {\V}}  := \langle \widebar {\,\cdot\,}^{-1} ,\cdot\rangle_\V\qquad \Theta:=\begin{pmatrix}\theta_{ab}\\\zeta\end{pmatrix}:=\widebar {\Gamma}\,.\eeq
These definitions are tailored in such a way that the second order action for $\Gamma$ on on-shell backgrounds can now be re-written as
$$S^{(2)}(\Gamma) = S(\GG) -\frac12 \langle \Gamma, P \Gamma\rangle_\V=S(\GG)-\frac12\langle \Theta,{\widebar {P}}\Theta\rangle_{\widebar {\V}} =:\widebar {S}^{(2)}(\Theta)\,.$$
Moreover, the bilinear form $\langle \cdot ,\cdot\rangle_{\widebar {\V}} $ is non-degenerate as $\langle \cdot ,\cdot\rangle_\V$ is non-degenerate and $\widebar {\,\cdot\,}^{-1}$ is injective due to our standing assumptions; one can also check that the re-defined bilinear form is symmetric\footnote{The symmetry of the re-defined bilinear form is not automatic, but only holds if $\widebar {\,\cdot\,}$ is given by a symmetric automorphism on the fibres of $\V$. To achieve this also for $\xi\neq0$, our normalisation of the tensor part of $\widebar {\,\cdot\,}$ deviates from the one induced by a direct generalisation of standard trace reversal by a factor of $\frac14$.}, cf. \eqref{eq_redefinedForm}. Furthermore, provided the background fields $(g,\phi)$ are on-shell, ${\widebar {P}}$ is formally selfadjoint with respect to  $\langle \cdot ,\cdot\rangle_{\widebar {\V}} $ and ${\widebar {P}}\circ \widebar {K}=0$. Thus, in this case, ${\widebar {P}}\Theta=0$ is the Euler-Lagrange equation for the action $\widebar {S}^{(2)}(\Theta)$, which is invariant with respect to $\Theta\mapsto \Theta+\widebar {K}\varsigma$ for arbitrary compactly supported $\varsigma\in \Gamma(\W)$. We stress that the field re-definition used here is merely a computational trick in order to cast the linearised Einstein-Klein-Gordon system in a more manageable form. In particular, all physical properties of the quantized field theory we shall discuss in the following, such as e.g. gauge-invariance and the properties of the commutation relations, do not depend on this field re-definition. After all constructions are performed, one may re-write the results in terms of the old field variables without further effort. We shall sketch this at the end of the next section.

We now provide the expressions for ${\widebar {P}}$, $\widebar {K}$ and $\langle \cdot, \cdot\rangle_{\widebar {\V}} $ (see \eqref{eq_redefinedP} in the Appendix for the complete form of ${\widebar {P}}$). As we always assume that the background metric and scalar field satisfy the full Einstein-Klein-Gordon equations, we have used these to simplify ${\widebar {P}}$.

\begin{equation*}
{\widebar {P}}:\Gamma(\V)\to\Gamma(\V)\qquad {\widebar {P}}=\begin{pmatrix}{\widebar {P}}_0 & {\widebar {P}}_2\\{\widebar {P}}_3 & {\widebar {P}}_1\end{pmatrix}\end{equation*}

$$({\widebar {P}}_0 \theta)_{ab}=-\nabla_c\nabla^c \theta_{ab}+2\nabla^c \nabla_{(a} \theta_{b)c} - g_{ab} \nabla^c\nabla^d \theta_{cd}+\text{ lower derivative orders}$$
 
$${\widebar {P}}_2 \zeta = \text{``no second derivatives''}$$ 
 
$${\widebar {P}}_3 \theta = 4\frac{\xi\phi}{\alpha}\nabla^a\nabla^b\theta_{ab}+\text{ lower derivative orders}$$
 
$${\widebar {P}}_1 \zeta = -\nabla_a\nabla^a\zeta+\text{ lower derivative orders}$$

\begin{equation}\label{eq_redefinedK}\widebar {K}:\Gamma(\W)\to\Gamma(\V)\qquad \widebar {K}\varsigma=\begin{pmatrix}\displaystyle\frac{\alpha}{2}\left( \nabla_{(a}\varsigma_{b)}-\frac12 g_{ab} \nabla_c\varsigma^c\right)+\frac{\xi\phi}{2} g_{ab}\varsigma^c\nabla_c\phi\\\displaystyle\left(1+\frac{2\xi^2\phi^2}{\alpha}\right)\varsigma^c\nabla_c\phi+\xi\phi\nabla_c\varsigma^c\end{pmatrix}\end{equation}

\begin{equation}\label{eq_redefinedForm}\langle \Theta_1,\Theta_2\rangle_{\widebar {\V}} =
\int\limits_M \vol \left(\frac{4}{\alpha}{\theta_1}_{ab}\theta_2^{ab}+\frac{4\xi^2\phi^2-2\beta}{\beta\alpha}{\theta_1}_{c}^{\phantom{c}c}{\theta_2}_{d}^{\phantom{c}d}+\frac{2\xi\phi}{\beta}\left({\theta_1}_{c}^{\phantom{c}c}\zeta_2+{\theta_2}_{c}^{\phantom{c}c}\zeta_1 \right)+\frac{\alpha}{\beta}\zeta_1\zeta_2\right)\end{equation}

Coming back to our aim of using gauge-invariance to put the equation of motion ${\widebar {P}}\Theta=0$ into a hyperbolic form, we note that ${\widebar {P}}$ has the form of a  normally hyperbolic differential operator, i.e. $-\nabla_c\nabla^c +"\text{lower orders}"$, up to terms which contain derivatives of $\theta_{ab}$ of the form $\nabla^b\theta_{ba}$. These terms are not present for any field configuration $\Theta\in\Gamma(\V)$ which satisfies 
$$\widebar {K}^\dagger \Theta = -2\nabla^b \theta_{ba} + (\nabla_a \phi)\zeta=0$$
in addition to ${\widebar {P}}\Theta=0$. Here, $\widebar {K}^\dagger:\Gamma(\V)\to \Gamma(\W)$ is the adjoint of $\widebar {K}$ defined by $\langle \widebar {K}^\dagger \Theta,\varsigma\rangle_\W := \langle  \Theta,\widebar {K}\varsigma\rangle_{\widebar {\V}} $ for all $\Theta\in\Gamma(\V)$, $\varsigma\in\Gamma(\W)$ with compact overlapping support. Given any solution $\Theta$ of ${\widebar {P}}\Theta=0$, we can find a gauge-equivalent solution $\Theta^\prime = \Theta + \widebar {K}\varsigma$ which satisfies $\widebar {K}^\dagger \Theta^\prime=0$ by solving $\widebar {K}^\dagger \widebar {K} \varsigma = - \widebar {K}^\dagger \Theta$ for $\varsigma$. This is possible because 
$$\widebar {K}^\dagger\circ  \widebar {K}: \Gamma(\W)\to \Gamma(\W)$$
\begin{equation}\label{eq_KDaggerK}( \widebar {K}^\dagger \widebar {K} \varsigma)_a = -\frac{\alpha}{2} \nabla^b\nabla_b \varsigma_a +\xi \phi (\nabla^b \phi)\nabla_b \varsigma_a-\frac{\alpha}{2}R_{ab}\varsigma^b-\end{equation}$$-\frac{2\xi+\alpha-3\xi\alpha}{\alpha}(\nabla_a \phi)(\nabla^b \phi)\varsigma_b-\frac{\xi \phi}{4}(\nabla_a \nabla_b\phi)\varsigma^b$$
is a multiple of a normally hyperbolic operator and thus has a well-posed Cauchy problem for arbitrary sources, see \cite[Thm 3.2.11]{Bar:2007zz} and \cite[Corollary 5]{Ginoux}. One may call the gauge defined by $\widebar {K}^\dagger \Theta=0$ ``generalised de Donder gauge''. Introducing a ``gauge-fixing operator'' 
\begin{equation}\label{eq_T}
T:\Gamma(\W)\to\Gamma(\V)\qquad T:= \frac{2}{\alpha} \widebar {K}\,,
\end{equation}
we can define a ``gauge-fixed equation of motion operator'' by (the full expression is displayed in \eqref{eq_PTilde} in the Appendix)
$$
\widetilde{P}:\Gamma(\V)\to\Gamma(\V)\qquad \widetilde{P}:= {\widebar {P}} + T\circ \widebar {K}^\dagger$$

$$\widetilde{P}=\begin{pmatrix}-\nabla_c\nabla^c & 0\\0 & -\nabla_c\nabla^c\end{pmatrix}+\text{ lower derivative orders}.$$

$\widetilde{P}$ is indeed normally hyperbolic and thus unique solutions to $\widetilde{P}\Theta=0$ exist for arbitrary initial data, i.e. arbitrary prescriptions for $\Theta$ and its normal derivative, on any Cauchy surface of $(M,g)$ \cite[Thm 3.2.11]{Bar:2007zz}. Note that this property does not fix $T$ uniquely, in fact every $T^\prime$ which differs from $T$ by terms without derivatives leads to a normally hyperbolic gauge-fixed equation of motion. 

A further important property of $T$ \eqref{eq_T} is that $Q:= \widebar {K}^\dagger \circ T$ is normally hyperbolic and has a well-defined Cauchy problem as well. Indeed, we can compute
$$Q:\Gamma(\W)\to\Gamma(\W)\qquad   Q:= \widebar {K}^\dagger \circ T $$
\begin{equation}\label{eq_KDaggerT}\left( Q \varsigma\right)_a = -\nabla^b\nabla_b \varsigma_a  + \frac{4\xi \phi}{\alpha} (\nabla_a \phi)\nabla^b \varsigma_b-R_{ab}\varsigma^b+\end{equation}$$+\frac{\xi^2\left(\beta+3\alpha-2\xi\right)}{\alpha^2}(\nabla_a \phi)(\nabla^b \phi)\varsigma_b-\frac{2\xi \phi}{\alpha}(\nabla_a \nabla_b\phi)\varsigma^b-\frac{2\xi \phi}{\alpha} (\nabla^b \phi)\nabla_a \varsigma_b\,.$$
This property implies that the constraint $\widebar {K}^\dagger \Theta=0$ is compatible with the time evolution induced by $\widetilde{P}$ and thus one can obtain solutions of $\widebar {P}\Theta=0$ by solving $\widetilde{P}\Theta=0$ subject to the constraint $\widebar {K}^\dagger \Theta=0$, cf. Section \ref{sec_Cauchy}.

We now collect the already discussed properties of the linearised Einstein-Klein-Gordon theory in the following theorem. At this point we can view this model as a field theory defined by the data $(\M,\V,\W,{\widebar {P}},\widebar {K})$, where $\M:=(M,g,\phi)$ is shorthand for the smooth manifold $M$ with the background fields $(g,\phi)$. Moreover, the spaces of smooth sections $\Gamma(\V)$ and $\Gamma(\W)$ of the real vector bundles $\V$ and $\W$ over $M$ are endowed with the bilinear forms $\langle\cdot,\cdot\rangle_{\widebar {\V}} $ and $\langle\cdot,\cdot\rangle_\W$ respectively.

\begin{theo}
\label{prop_propLEKGS}The linearised Einstein-Klein-Gordon system defined by $(\M,\V,\W,{\widebar {P}},\widebar {K})$, where 
\begin{itemize}
\item $\M:=(M,g,\phi)$ with $(M,g)$ a smooth (connected, Hausdorff, orientable and time-orientable) four-dimensional globally hyperbolic spacetime and $(g,\phi)$ a solution of the Einstein-Klein-Gordon equations \eqref{eq_fullcoupled} with $1-\xi\phi^2\neq 0$ and $1+\left(6\xi-1\right)\xi\phi^2\neq 0$ on all $M$
\item $\V:= \bigvee^2 T^*M \oplus \left(M\times \bbR\right)$ and $\W:=TM$ real vector bundles over $M$
\item the spaces of smooth sections $\Gamma(\V)$ and $\Gamma(\W)$ of $\V$ and $\W$ over $M$ are endowed with the bilinear forms $\langle\cdot,\cdot\rangle_{\widebar {\V}} $ \eqref{eq_redefinedForm} and $\langle\cdot,\cdot\rangle_\W$ \eqref{eq_formW}
\item ${\widebar {P}}$ is the differential operator ${\widebar {P}}:\Gamma(\V)\to\Gamma(\V)$ defined in \eqref{eq_redefinedP} and $\widebar {K}$ is the differential operator $\widebar {K}:\Gamma(\W)\to\Gamma(\V)$ defined in \eqref{eq_redefinedK}
\end{itemize}
has the following properties.
\begin{enumerate}
\item $\langle\cdot,\cdot\rangle_{\widebar {\V}} $ and $\langle\cdot,\cdot\rangle_\W$ are symmetric and non-degenerate.
\item ${\widebar {P}}$ is formally selfadjoint with respect to $\langle\cdot,\cdot\rangle_{\widebar {\V}} $ and satisfies ${\widebar {P}}\circ \widebar {K}=0$.
\item The differential operator $R:\Gamma(\W)\to\Gamma(\W)$, $R:=\widebar {K}^\dagger \circ \widebar {K}$ \eqref{eq_KDaggerK}, with $\widebar {K}^\dagger:\Gamma(\V)\to\Gamma(\W)$, $\langle \widebar {K}^\dagger\cdot,\cdot \rangle_\W:=\langle \cdot,\widebar {K}\cdot \rangle_{\widebar {\V}} $ is a multiple of a normally hyperbolic operator and thus has a well-posed Cauchy problem.
\item There exists a differential operator $T:\Gamma(\W)\to\Gamma(\V)$, e.g. $T=\frac{2}{1-\xi\phi^2}\widebar {K}$, such that $\widetilde{P}:\Gamma(\V)\to\Gamma(\V)$, $\widetilde{P}:= {\widebar {P}} + T\circ \widebar {K}^\dagger$ \eqref{eq_PTilde} and $Q:\Gamma(\W)\to\Gamma(\W)$, $Q:= \widebar {K}^\dagger \circ T$ \eqref{eq_KDaggerT} are normally hyperbolic and have a well-posed Cauchy problem.
\end{enumerate}
\end{theo}

\subsection{Existence and uniqueness of solutions to the redefined equations of motion}
\label{sec_Cauchy}

Using the normal hyperbolicity of $\widetilde{P}$ as well as the properties of $T$ and $\widebar {K}$, we can analyse the Cauchy problem for ${\widebar {P}}$, i.e. existence of uniqueness of solutions to ${\widebar {P}}\Theta=0$ with suitable initial conditions. We only sketch this analysis as it is the straightforward generalisation of the linearised pure gravity case which has been discussed in Chapters 3.1 and 3.2 of \cite{Fewster:2012bj}. Moreover, we shall not need the results on the Cauchy problem for ${\widebar {P}}$ for the construction of the classical and quantum theory of the linearised Einstein-Klein-Gordon system, for this only the immediate properties of  $(\M,\V,\W,{\widebar {P}},\widebar {K})$ listed in Theorem \ref{prop_propLEKGS} are needed. Notwithstanding, a good understanding of the Cauchy problem for ${\widebar {P}}$ is e.g. necessary in order to convince oneself that the space of solutions to ${\widebar {P}}\Theta=0$ is ''sufficiently large'' and in particular non-empty.

To this avail, we pick an arbitrary but fixed Cauchy surface $\Sigma$ of $(M,g)$ with future-pointing unit normal vector field $n$. We define a map 
$$N:\Gamma(\V)|_\Sigma\to\Gamma(\W)|_\Sigma\qquad \Theta|_\Sigma=\left.\begin{pmatrix}\theta_{ab}\\\zeta\end{pmatrix}\right|_\Sigma\mapsto N(\Theta|_\Sigma)_a:=n^b (\theta_{ab})|_\Sigma.$$ One can check that $N({\widebar {P}}\Theta|_\Sigma)$ does not contain second derivatives with respect to $n$, thus $N({\widebar {P}}\Theta|_\Sigma)=0$ is a necessary constraint for initial data for ${\widebar {P}}\Theta=0$. We express this constraint in terms of a linear map\footnote{\eqref{eq_ConstraintP} defines a map on $\Gamma(\V)|_\Sigma\oplus \Gamma(\V)|_\Sigma$ because $\Gamma(\V)\ni\Theta\mapsto(\Theta|_\Sigma,\nabla_n\Theta|_\Sigma)\in\Gamma(\V)|_\Sigma\oplus \Gamma(\V)|_\Sigma$ is surjective and $N({\widebar {P}}\Theta|_\Sigma)$ does not contain second derivatives with respect to $n$.}
\begin{equation}\label{eq_ConstraintP}
C:\Gamma(\V)|_\Sigma\oplus \Gamma(\V)|_\Sigma\to \Gamma(\W)|_\Sigma\qquad C\left((\Theta|_\Sigma,\nabla_n\Theta|_\Sigma)\right):=N({\widebar {P}}\Theta|_\Sigma)
\end{equation}
and consider arbitrary initial data $(\Theta|_\Sigma, \nabla_n \Theta|_\Sigma)=(\Theta^0,\Theta^1)\in \Gamma(\V)|_\Sigma\oplus \Gamma(\V)|_\Sigma$ subject to $C((\Theta^0,\Theta^1))=0$. Using the method outlined in the proof of \cite[Theorem 3.1]{Fewster:2012bj}, this initial data can be smoothly extended to an auxiliary $\Theta^\prime\in \Gamma(\V)$. Owing to the normal hyperbolicity of $\widebar {K}^\dagger\circ \widebar {K}$, there exists a $\varsigma\in\Gamma(\W)$ s.t. $\Theta^{\prime\prime}:= \Theta^\prime + \widebar {K}\varsigma$ satisfies $\widebar {K}^\dagger \Theta^{\prime\prime}=0$. We now solve $\widetilde{P}\Theta^{\prime\prime\prime}=0$ with the Cauchy data $(\Theta^{\prime\prime}|_\Sigma, \nabla_n \Theta^{\prime\prime}|_\Sigma)$. Note that this Cauchy data still satisfies the necessary constraint $N({\widebar {P}}\Theta^{\prime\prime}|_\Sigma)=0$ because ${\widebar {P}}\circ \widebar {K}=0$, and that $\widebar {K}^\dagger \Theta^{\prime\prime}|_\Sigma=0$ is a further constraint on this Cauchy data because $\widebar {K}^\dagger$ contains derivatives of at most first order. We thus obtain $\Theta^{\prime\prime\prime}\in\Gamma(\V)$ which satisfies $\widetilde{P}\Theta^{\prime\prime\prime}=0$, $N({\widebar {P}}\Theta^{\prime\prime\prime}|_\Sigma)=0$ and $\widebar {K}^\dagger \Theta^{\prime\prime\prime}|_\Sigma=0$. One can now compute $0=N({\widebar {P}}\Theta^{\prime\prime\prime}|_\Sigma)=N((\widetilde{P}-T\circ \widebar {K})\Theta^{\prime\prime\prime}|_\Sigma)=-N(T\widebar {K}\Theta^{\prime\prime\prime})=-\nabla_n \widebar {K}^\dagger \Theta^{\prime\prime\prime}|_\Sigma$, where in the last steps one needs $\widebar {K}^\dagger \Theta^{\prime\prime\prime}|_\Sigma=0$ and the fact that $n_{a}n^{b}+\delta_{a}^{b}$ projects vectors fields on $M$ to their components tangential to $\Sigma$. To show that this implies $\widebar {K}^\dagger \Theta^{\prime\prime\prime}$ on the full spacetime, we 
compute  
$$\widebar {K}^\dagger \circ \widetilde{P} = \widebar {K}^\dagger \circ {\widebar {P}} +\widebar {K}^\dagger \circ T \circ \widebar {K}^\dagger  = \left({\widebar {P}}\circ \widebar {K}\right)^\dagger + Q \circ \widebar {K}^\dagger = Q \circ \widebar {K}^\dagger $$
with $Q$ as in \eqref{eq_KDaggerT}.
Thus, each solution of $\widetilde{P}\Theta^{\prime\prime\prime}=0$ satisfies $Q\widebar {K}^\dagger \Theta^{\prime\prime\prime}=0$ with $Q$ normally hyperbolic, and the unique solution of the latter equation with initial data $\nabla_n \widebar {K}^\dagger \Theta^{\prime\prime\prime}|_\Sigma=0$ and $\widebar {K}^\dagger \Theta^{\prime\prime\prime}|_\Sigma=0$ is $\widebar {K}^\dagger \Theta^{\prime\prime\prime}=0$. Hence, the previously constructed solution of $\widetilde{P}\Theta^{\prime\prime\prime}=0$ satisfies $\widebar {K}^\dagger \Theta^{\prime\prime\prime}=0$, and consequently ${\widebar {P}}\Theta^{\prime\prime\prime}=0$. Finally, setting $\Theta:= \Theta^{\prime\prime\prime}-\widebar {K}\varsigma$ we obtain a solution of ${\widebar {P}}\Theta=0$ with the original Cauchy data $(\Theta|_\Sigma, \nabla_n \Theta|_\Sigma)=(\Theta^0,\Theta^1)$, which was arbitrary barring the constraint $C((\Theta^0,\Theta^1))=N({\widebar {P}}\Theta|_\Sigma)=0$.

To analyse uniqueness, we assume that ${\widebar {P}}\Theta={\widebar {P}}\Theta^\prime=0$ with $\Theta$ and $\Theta^\prime$ having the same Cauchy data on an arbitrary but fixed Cauchy surface $\Sigma$ with future-pointing unit normal vector field $n$. Thus $\Theta^{\prime\prime}:=\Theta^{\prime}-\Theta^{}$ solves ${\widebar {P}}\Theta^{\prime\prime}=0$ with vanishing Cauchy data. Using once more the hyperbolicity of $\widebar {K}^\dagger \circ \widebar {K}$, we can write $\Theta^{\prime\prime}=\Theta^{\prime\prime\prime}+\widebar {K}\varsigma$, where $\varsigma$ solves $\widebar {K}^\dagger \widebar {K}\varsigma=\widebar {K}^\dagger \Theta^{\prime\prime}$ with vanishing Cauchy data on $\Sigma$, and $\widebar {K}^\dagger \Theta^{\prime\prime\prime}=0$. Using the properties of $\varsigma$ and its Cauchy data, as well as the vanishing of the Cauchy data of $\Theta^{\prime\prime}$, one can compute that $\Theta^{\prime\prime\prime}$ has vanishing Cauchy data as well. As $\Theta^{\prime\prime\prime}$ solves the normally hyperbolic equation $\widetilde{P}\Theta^{\prime\prime\prime}=0$, $\Theta^{\prime\prime\prime}$ must vanish identically because zero is the only solution of a normally hyperbolic equation with vanishing Cauchy data. This in turn implies that a solution to ${\widebar {P}}\Theta=0$ with given Cauchy data is unique up to gauge transformations. Summing up, we have found the following.

\begin{theo}
The properties of $(\M,\V,\W,{\widebar {P}},\widebar {K})$ proved in Theorem \ref{prop_propLEKGS} imply the following for the Cauchy problem of ${\widebar {P}}$. Let $\Sigma$ be an arbitrary but fixed Cauchy surface of $(M,g)$ with future pointing unit normal vector field $n$.
\begin{enumerate}
\item For every $(\Theta^0,\Theta^1)\in \Gamma(\V)|_\Sigma\oplus \Gamma(\V)|_\Sigma$ subject to the constraint $C((\Theta^0,\Theta^1))=0$ with $C$ as in \eqref{eq_ConstraintP} there exists $\Theta\in \Gamma(\V)$ with ${\widebar {P}}\Theta=0$, $(\Theta|_\Sigma,\nabla_n\Theta|_\Sigma)=(\Theta^0,\Theta^1)$.
\item If ${\widebar {P}}\Theta={\widebar {P}}\Theta^\prime=0$ and $(\Theta|_\Sigma,\nabla_n\Theta|_\Sigma)=(\Theta^\prime|_\Sigma,\nabla_n\Theta^\prime|_\Sigma)$ then there exists $\varsigma\in\Gamma(\W)$ s.t. $\Theta^\prime-\Theta=\widebar {K}\varsigma$.
\end{enumerate}
\end{theo}

 \subsection{Quantization}
\label{sec_quantization}

In \cite{HS} a general procedure how to quantize arbitrary linear gauge theories on curved spacetimes has been developed by generalising the quantization of the Maxwell field in \cite{DimockVector} and the quantization of linearised gravity in \cite{Fewster:2012bj}. As a matter of fact it is demonstrated in  \cite{HS} how to quantize any system $(\M,\V,\W,{\widebar {P}},\widebar {K})$ which satisfies the conditions $1-4$ in Theorem \ref{prop_propLEKGS} in a consistent and gauge-invariant manner. Thus, we can just apply the results of \cite{HS} in order to obtain a quantum theory for the linearised Einstein-Klein-Gordon system on arbitrary on-shell backgrounds (satisfying the conditions $\alpha=1-\xi\phi^2\neq 0$ and $\beta=1+\left(6\xi-1\right)\xi\phi^2\neq 0$ on all $M$). In this section we review this construction and add a few more details as a preparation for looking at the special case of perturbations in Inflation.

\subsubsection{Motivation: a toy model}
\label{sec_toymodel}

The quantization construction performed in \cite{HS} works without introducing auxiliary fields as in the BRST/BV formalism, see e.g. \cite{Hollands:2007zg,Fredenhagen:2011mq}, but is restricted to linear field theories, for which it gives the same results as the BRST/BV formalism. As it is very much in the spirit of the algebraic approach to QFT, and some of the readers might not be familiar with this framework, we would like to briefly sketch and motivate the construction by means of a very simple gauge theory model.

We consider as a gauge field $\Phi=(\phi_1,\phi_2)^T$ a tuple of two scalar fields on a spacetime $(M,g)$ satisfying the equation of motion
$$P\Phi =\begin{pmatrix}-\nabla_a\nabla^a & \phantom{-}\nabla_a\nabla^a\\\phantom{-}\nabla_a\nabla^a & -\nabla_a\nabla^a\end{pmatrix}\begin{pmatrix}\phi_1\\\phi_2\end{pmatrix}=0\,.$$
The gauge transformations are $\Phi\mapsto \Phi + K \varsigma:=\Phi+(\varsigma,\varsigma)^T$ and indeed $P\circ K=0$ holds which is equivalent to the gauge-invariance of the action $S(\Phi):=\frac12\langle \Phi,P\Phi\rangle$ with $\langle \Phi,\Phi^\prime\rangle:=\int_M \vol( \phi^{\phantom{\prime}}_1 \phi^\prime_1+\phi^{\phantom{\prime}}_2 \phi^\prime_2)$. Clearly, the linear combination $\psi:=\phi_1-\phi_2$ is gauge-invariant and satisfies $-\nabla_a\nabla^a\psi=0$, and it would be rather natural to quantize $\Phi$ by directly quantizing $\psi$ as a massless scalar field. This would be much in the spirit of the usual quantization of perturbations in Inflation, where gauge-invariant linear combinations of the gauge field components, e.g. the Bardeen-Potentials or the Mukhanov-Sasaki variable, are taken as the fundamental fields for quantization. However, in more general cases such as the linearised Einstein-Klein-Gordon system on arbitrary backgrounds it may be rather difficult to directly identify a gauge-invariant fundamental field like $\psi$ whose quantum theory is equivalent to the quantum theory of the original gauge field. Notwithstanding, an indirect characterisation of such a gauge-invariant linear combination of gauge-field components, which can serve as a fundamental field for quantization, is still possible. In the toy model under consideration we consider a tuple  $f=(f_1,f_2)$ of test functions $f_i$, i.e. infinitely often differentiable functions which vanish outside of a compact set in $M$. We ask that $K^\dagger f:=f_1+f_2=0$, where $K^\dagger$ is the adjoint of the gauge transformation operator $K$ i.e. $\int_M\vol \varsigma K^\dagger f=\langle K\varsigma, f\rangle$. Clearly, any $f$ satisfying this condition is of the form $f=(h,-h)^T$ for a test function $h$. We now observe that the pairing between a gauge field configuration $\Phi$ and such an $f$ is gauge-invariant, i.e. $\langle \Phi + K\sigma, f\rangle = \langle \Phi, f\rangle + \int_M\vol \varsigma K^\dagger f=\langle \Phi, f\rangle$. Thus we can consider the ``smeared field'' $\langle \Phi, f\rangle$, with $f=(h,-h)^T$ and arbitrary $h$, as a gauge-invariant linear combination of gauge-field components which is suitable for playing the role of a fundamental field for quantization. We can compute $\langle \Phi, f\rangle=\int_M\vol \psi h$, and observe that, up to the ``smearing'' with $h$, this indirect choice of gauge-invariant fundamental field is exactly the one discussed in the beginning. If one chooses $h$ to be the delta distribution $\delta(x,y)$ rather than a test function, one even finds $\langle \Phi, f\rangle=\psi(x)$, but using the more regular test functions is advantageous for mathematical consistency and also more realistic in physical terms, since the extended support of $h$ in spacetime reflects the finite spatial and temporal dimensions of a measurement and thus $\langle \Phi, f\rangle$ can be interpreted as a weighted, gauge-invariant measurement of the field $\Phi$. Moreover, as already anticipated, in general gauge theories with more complicated gauge transformation operators $K$ it might be extremely difficult to classify all solutions of $K^\dagger f=0$, which would be equivalent to a direct characterisation of one or several fundamental gauge-invariant fields such as $\psi$, whereas working implicitly with the condition $K^\dagger f=0$ is always possible.

\subsubsection{The presymplectic space of classical linear observables}

After discussing the toy model gauge field, we turn back to the quantization of the linearised Einstein-Klein-Gordon system specified by $(\M,\V,\W,{\widebar {P}},\widebar {K})$ as in Theorem \ref{prop_propLEKGS}. In analogy to the toy model example, we consider test sections, i.e. smooth and compactly supported sections $h=(k_{ab},f)^T\in \Gamma_0(\V)$ which satisfy the condition
\beq \widebar {K}^\dagger h=-2\nabla^b k_{ba} + (\nabla_a \phi)f=0\,.\eeq
For any configuration $\Theta\in\Gamma(\V)$ of the gauge field and any such $h$, the pairing $\langle \Theta,h\rangle_{\widebar {\V}} $ is gauge invariant, i.e.
$$\langle \Theta+\widebar {K}\varsigma,h\rangle_{\widebar {\V}} =\langle \Theta,h\rangle_{\widebar {\V}} +\langle \varsigma,\widebar {K}^\dagger h\rangle_\W=\langle \Theta,h\rangle_{\widebar {\V}} \qquad \forall \varsigma\in\Gamma(\W)\,.$$
Thus, we consider the ``smeared field'' $\langle \Theta,h\rangle_{\widebar {\V}} $ to be the fundamental gauge-invariant dynamical field.
As we are interested in configurations $\Theta$ which satisfy ${\widebar {P}}\Theta=0$, we observe that all $h$ which are of the form $h={\widebar {P}}h^\prime$, $h^\prime \in\Gamma_0(\V)$, will give zero once paired with a solution $\Theta$ and thus correspond to the trivial observable. This motivates the following definitions of configuration and test section spaces, cf. Section \ref{sec_notations} for the definition of $\Gamma_\sc(\W)$ and $\Gamma_\sc(\V)$.
\beq\label{eq_defSol}\widebar {\Sol}:= \{\Theta\in \Gamma(\V)\,|\, {\widebar {P}}\Theta=0\}\qquad \widebar {\Sol}_\text{sc}:= \widebar {\Sol}\cap \Gamma_\sc(\V)\eeq
\beq\label{eq_defcalG}\widebar {\G}:=\widebar {K}\left[\Gamma(\W)\right]\qquad {\widebar {\G}}_\text{sc}:= {\widebar {\G}}\cap \Gamma_\sc(\V)\qquad \widebar {\G}_{\sc,0}:=\widebar {K}[\Gamma_\sc(\W)] \eeq
\beq\label{eq_defcalE} \Ker_0(\widebar {K}^\dagger):=\{h\in\Gamma_0(\V)\,|\,\widebar {K}^\dagger h=0\}\qquad \widebar {\E}:= \Ker_0(\widebar {K}^\dagger)/{\widebar {P}}[\Gamma_0(\V)]\eeq
We have $\widebar {\G}_{\sc,0}\subset{\widebar {\G}}_\sc\subset\widebar {\Sol}_\sc$ and similarly ${\widebar {\G}}\subset\widebar {\Sol}$, because ${\widebar {P}}\circ \widebar {K}=0$. Thus we can consider the quotients $\widebar {\Sol}/{\widebar {\G}}$, $\widebar {\Sol}_\sc/{\widebar {\G}}_\sc$ and $\widebar {\Sol}_\sc/\widebar {\G}_{\sc,0}$. Elements of $\widebar {\Sol}/{\widebar {\G}}$ are gauge-equivalence classes of solutions and $\widebar {\Sol}/{\widebar {\G}}$ can be viewed as the space of (pure) states in the classical field theory of the linearised Einstein-Klein-Gordon system. The other two quotients are related to observables as we will discuss later. Observe that $\widebar {\E}$ is well defined due to $\widebar {K}^\dagger \circ {\widebar {P}}=0$ (which follows from the formal selfadjointness of ${\widebar {P}}$ and ${\widebar {P}}\circ \widebar {K}=0$). 

Our previous discussion can now be re-phrased by saying that the non-degenerate bi-linear form $\langle \cdot, \cdot\rangle_{\widebar {\V}} $ induces a pairing between $\widebar {\Sol}/{\widebar {\G}}$ and $\widebar {\E}$, because for any $[\Theta]\in \widebar {\Sol}/{\widebar {\G}}$ and any $[h]\in\widebar {\E}$, $\langle \Theta,h\rangle_{\widebar {\V}} $ is independent of the representatives of the two equivalence classes. We will denote this pairing by the same symbol $\langle \cdot, \cdot\rangle_{\widebar {\V}} $. Via this pairing, elements of $\widebar {\E}$ can be considered as labels for the most simple gauge-invariant observables $$\widebar {\Sol}/{\widebar {\G}}\ni [\Theta]\mapsto\langle \Theta,h\rangle_{\widebar {\V}} \qquad [h]\in\widebar {\E}$$  on the space of classical states. After quantization, these observables can be viewed as ``smeared quantum fields''.

 A natural question is whether $\widebar {\E}$ is separating on $\widebar {\Sol}/{\widebar {\G}}$, i.e. whether $$\langle \Theta,h\rangle_{\widebar {\V}} =0\quad\forall\; [h]\in\widebar {\E}\qquad\Rightarrow\qquad[\Theta]=[0]\,.$$
Physically this would imply that the set of classical observables labelled by $\widebar {\E}$ is large enough for distinguishing all classical pure states. While this question can be answered positively in pure electromagnetism on arbitrary curved spacetimes due to the rather rich algebraic structure of the relevant differential operators, cf.\ \cite{Benini:2013tra,Benini:2013ita}, we expect that an analysis of the separability for the gauge theory considered here would be difficult without any restrictions on the background spacetime $(M,g)$ as the known strategies to tackle this issue involve the analysis of certain elliptic operators on Cauchy surfaces of $(M,g)$ -- see e.g. \cite{Fewster:2012bj} and \cite{HS} for the case of linearised pure gravity and linearised pure supergravity -- whose properties in turn strongly depend on $(M,g)$. While we do not address the question of separability in full generality, we will demonstrate in the next section that, at least on cosmological backgrounds $(M,g,\phi)$, a subset of $\widebar {\E}$ is indeed separating on a subset of $\widebar {\Sol}/{\widebar {\G}}$.

The next step in the quantization procedure is to endow the linear observables labelled by $\widebar {\E}$ with a Poisson bracket, which upon quantization defines the canonical commutation relations of the smeared quantum fields. To this avail we consider the causal propagator \beq\label{eq_def_G}G^{\widetilde{P}}:\Gamma_0(\V)\to\Gamma(\V)\,,\qquad G^{\widetilde{P}}:=G_+^{\widetilde{P}}-G_-^{\widetilde{P}}\eeq
 of the gauge-fixed equation of motion operator $\widetilde{P}={\widebar {P}}+T\circ \widebar {K}^\dagger$ (cf. Theorem \ref{prop_propLEKGS}) built from the retarded/advanced Green's operators 
$$G_\pm^{\widetilde{P}}:\Gamma_0(\V)\to\Gamma(\V),\qquad \widetilde{P}G_\pm^{\widetilde{P}}=\id_{\Gamma_0(\V)}\,,\qquad \supp\, G_\pm^{\widetilde{P}} h\subset J^\pm(\supp \,h)\quad\forall \,h\in\Gamma_0(\V)$$ 
  of $\widetilde{P}$ which exist and are unique because $\widetilde{P}$ is normally hyperbolic \cite{Bar:2007zz}. The normal hyperbolicity of $\widetilde{P}$ implies that the kernel of $G^{\widetilde{P}}$ is precisely the image of $\widetilde{P}$ restricted to $\Gamma_0(\V)$ and that the image of $G^{\widetilde{P}}$ are precisely the space-like compact solutions of $\widetilde{P}\Theta=0$. Consequently, $G^{\widetilde{P}}$ descends to a bijective map between $\Gamma_0(\V)/\widetilde{P}[\Gamma_0(\V)]$ and the space-like compact solutions of $\widetilde{P}\Theta=0$. 
In fact, one can show that similar relations hold between $G^{\widetilde{P}}$ and ${\widebar {P}}$ rather than $\widetilde{P}$, such that, in a certain sense, $G^{\widetilde{P}}$ can be considered as an effective causal propagator for ${\widebar {P}}$ itself.

\begin{theo}\label{prop_prop_G} The causal propagator $G^{\widetilde{P}}$ \eqref{eq_def_G} of $\widetilde{P}={\widebar {P}}+T\circ \widebar {K}^\dagger$ \eqref{eq_PTilde} satisfies the following relations.
\begin{enumerate}
\item $h\in\Ker_0(\widebar {K}^\dagger)$ and $G^{\widetilde{P}}h\in\widebar {\G}_{\sc,0}$ if and only if $h\in{\widebar {P}}[\Gamma_0(\V)]$, with $\widebar {\G}_{\sc,0}$ defined in \eqref{eq_defcalG}.
\item Every $\Theta\in \widebar {\Sol}_\sc$ 
can be split as $h=h_1+h_2$ with $h_1\in G^{\widetilde{P}}\left[\Ker_0(\widebar {K}^\dagger)\right]$ and $h_2\in \widebar {\G}_{\sc,0}$, with $\widebar {\Sol}_\sc$ defined in \eqref{eq_defSol}.
\item $G^{\widetilde{P}}$ descends to a bijective map $G^{\widetilde{P}}:\Ker_0(\widebar {K}^\dagger)/{\widebar {P}}[\Gamma_0(\V)]=\widebar {\E}\to \widebar {\Sol}_\sc/\widebar {\G}_{\sc,0}$.
\item $G^{\widetilde{P}}$ is formally skew-adjoint w.r.t. $\langle\cdot,\cdot\rangle_{\widebar {\V}} $ on $\Ker_0(\widebar {K}^\dagger)$, i.e. $\langle h_1,G^{\widetilde{P}}h_2\rangle_{\widebar {\V}} =-\langle G^{\widetilde{P}}h_1,h_2\rangle_{\widebar {\V}} $ for all $h_1,h_2\in\Ker_0(\widebar {K}^\dagger)$.
\item Let $T^\prime:\Gamma(\W)\to\Gamma(\V)$ be any differential operator satisfying property 4 in Theorem \ref{prop_propLEKGS} and let $G^{\widetilde{P}^\prime}$ be the causal propagator of $\widetilde{P}^\prime:={\widebar {P}}+T^\prime\circ \widebar {K}^\dagger$. Then $G^{\widetilde{P}^\prime}$ satisfies the properties 1-4 above.
\end{enumerate}
\end{theo}
These results have been proved in \cite[Theorem 3.12+Theorem 5.2]{HS} using the abstract properties 1-4 of $(\M,\V,\W,{\widebar {P}},\widebar {K})$ proven in Theorem \ref{prop_propLEKGS}. Thereby it is essential that for every $h\in\Ker_0(\widebar {K}^\dagger)$, $\widetilde{P}h=\widetilde{P}^\prime h={\widebar {P}}h$. Property 5 above indicates that $G^{\widetilde{P}}$ restricted to $\Ker_0(\widebar {K}^\dagger)$ is independent of the particular form of the gauge fixing operator $T$ and enforces the point of view that $G^{\widetilde{P}}$ restricted to $\Ker_0(\widebar {K}^\dagger)$ is effectively a causal propagator for ${\widebar {P}}$.

By means of $G^{\widetilde{P}}$ we define the bilinear map
\beq\label{eq_def_tau}\widebar {\sigma}:\widebar {\E}\times\widebar {\E}\mapsto \bbR,\qquad([h_1],[h_2])\mapsto\widebar {\sigma}([h_1],[h_2]):=\langle h_1,G^{\widetilde{P}}h_2\rangle_{\widebar {\V}} \,.\eeq
Theorem \ref{prop_prop_G} now immediately implies the following properties of $\widebar {\sigma}$.
\begin{cor}\label{prop_proptau}
The bilinear form $\widebar {\sigma}$ on $\widebar {\E}$, defined respectively in \eqref{eq_def_tau} and \eqref{eq_defcalE}, has the following properties.
\begin{enumerate}
\item $\widebar {\sigma}$ is well-defined, i.e. $\langle h_1,G^{\widetilde{P}}h_2\rangle_{\widebar {\V}} $ is independent of the representatives $h_i\in[h_i]$.
\item $\widebar {\sigma}$ is antisymmetric.
\item Let $T^\prime:\Gamma(\W)\to\Gamma(\V)$ be any differential operator satisfying property 4 in Theorem \ref{prop_propLEKGS} and define $\widebar {\sigma}^\prime$ in analogy to \eqref{eq_def_tau} but with the causal propagator $G^{{\widetilde{P}}^\prime}$ of $\widetilde{P}^\prime:={\widebar {P}}+T^\prime\circ \widebar {K}^\dagger$ instead of $G^{{\widetilde{P}}}$. Then $\widebar {\sigma}^\prime=\widebar {\sigma}$.
\end{enumerate}
\end{cor} 
The last property indicates that $\widebar {\sigma}$ is independent of the gauge-fixing operator $T$ and in this sense, gauge-invariant. Indeed, we shall demonstrate later that $\widebar {\sigma}$ can be rewritten in a manifestly gauge-invariant form. The form of $\widebar {\sigma}$ given here can be derived directly from the action $\widebar {S}^{(2)}(\Theta)=\frac12 \langle \Theta,{\widebar {P}}\Theta\rangle_{\widebar {\V}} $ by Peierls' method  in analogy to the derivation for electromagnetism in \cite{Sanders:2012sf}, see also \cite{Khavkine:2012jf, Khavkine:2014kya}) for a broader context.  

\subsubsection{Quantization of the presymplectic space}
\label{sec_quant}
The pair $(\widebar {\E},\widebar {\sigma})$ forms a (pre)symplectic space, which can be quantized in several, essentially equivalent, ways, that can all be called ``canonical''. One rather technical possibility is to quantize this presymplectic space in terms of a $C^\ast$-algebra \cite{degenerateCCR} (see also \cite{Araki:1970zza} for earlier work in this direction), which is essentially generated by exponentiated smeared quantum fields. A different option is to consider the polynomial algebra $\A$ of smeared quantum fields $\Theta(h)$ which correspond to the quantization of the classical linear observables $\langle \Theta,h\rangle_{\widebar {\V}} $. I.e. $\A$ is generated by a unit element $\1$ and sums of products of $\Theta(h_i)$ with $[h_i]\in\widebar {\E}_\bbC:=\widebar {\E}\otimes_\bbR\bbC$. There is an antilinear involution $^\dagger$ on $\A$ defined by
$$(A_1 A_2)^\dagger=A_2^\dagger A_1^\dagger\quad\forall A_i\in\A\,,\qquad \Theta(h)^\dagger=\Theta(h^\ast)$$
where $h^\ast$ is the complex conjugate of $h$. Finally the smeared quantum fields $\Theta(h)$ satisfy the (covariant) canonical commutation relations (CCR)\footnote{We tacitly extend $G_{\widetilde{P}}$, $\langle\cdot,\cdot\rangle_{\widebar {\V}} $ and $\widebar {\sigma}$ to complexified domains by $\bbR$-linearity.}
\beq\label{eq_def_CCR}\left[\Theta(h_1),\Theta(h_2)\right]:=\Theta(h_1)\Theta(h_2)-\Theta(h_2)\Theta(h_1)=i\widebar {\sigma}([h_1],[h_2])\1=i\langle h_1, G^{\widetilde{P}}h_2\rangle_{\widebar {\V}} \1\,.\eeq
We shall soon demonstrate that $\widebar {\sigma}$ is equivalent to a bi-linear form which can be computed on an arbitrary, but fixed Cauchy surface of $(M,g)$, i.e. at ``equal time''. This implies that the covariant CCR \eqref{eq_def_CCR} can be equivalently expressed in terms of ``equal-time'' CCR of field and canonical momentum. One may be tempted to read \eqref{eq_def_CCR} as ``$[\Theta(x),\Theta(y)]=i G^{\widetilde{P}}(x,y)\1$'', where $G^{\widetilde{P}}(x,y)$ is the integral kernel of the operator $G^{\widetilde{P}}$, but, while this makes sense for theories without gauge invariance, it does not make sense here because \eqref{eq_def_CCR} only holds for $h_i\in[h_i]\in\widebar {\E}$, i.e. for $h_i$ s.t. $\widebar {K}^\dagger h_i=0$. 

Note that the all elements of the quantum field algebra $\A$ constructed as above are gauge-invariant because both the generators $\Theta(h)$ and the commutation relation are gauge-invariant in the sense discussed before. In particular, all elements of $\A$ are physically meaningful observables. Given a state on the algebra $\A$, i.e. a positive and normalised linear functional $$\langle\;\;\rangle_\omega:\A\to\bbC,\qquad \langle \1\rangle_\omega=1,\qquad \langle A^\dagger A\rangle_\omega\ge 0\quad\forall A\in\A\,,$$
one can construct a Hilbert space representation $\pi$ of $\A$ by means of the GNS-theorem \cite{Haag}. This gives a Hilbert space $\H$ with a vector $\Omega$ s.t. the smeared quantum fields $\Theta(h)$ become operators $\pi(\Theta(h))$ on $\H$ and $\langle A\rangle_\omega=\langle\Omega|\pi(A)|\Omega\rangle$, $\pi(A)^\dagger=\pi(A^\dagger)$ for all $A\in\A$. Discussing the existence of states -- and thus of Hilbert space representations of $\A$ -- is beyond the scope of this paper, but we presume that one can show existence by methods similar to the ones used in \cite{Brunetti:2013maa} on general backgrounds or by generalising the construction of so-called states of low energy on cosmological backgrounds \cite{Olbermann:2007gn} from the scalar case to the one at hand. Note that the abstract properties 1-4 of $(\M,\V,\W,{\widebar {P}},\widebar {K})$ imply that the algebra $\A$ satisfies the so-called time-slice axiom, see \cite[Proposition 4.11]{HS}, i.e. for each generator $\Theta(h)$ there exists a $h\in[h]\in\widebar {\E}$ which has support in an arbitrarily small neighbourhood of an arbitrary but fixed Cauchy surface of $(M,g)$ -- a ``time slice'' -- and thus $\A$ is fully determined by information contained in this time slice.

An interesting question is whether the presymplectic form $\widebar {\sigma}$ is non-degenerate and thus symplectic, i.e. whether
$$\widebar {\sigma}([h_1],[h_2])=0\quad\forall\, [h_1]\in \widebar {\E}\qquad\Rightarrow\qquad [h_2]=[0]\,.$$
If this were not the case, then the quantum field algebra $\A$ would contain a non-trivial centre, i.e. elements not proportional to the identity which commute with all elements in $\A$. The question whether or not this happens is loosely related to the question whether the classical linear observables labelled by $\widebar {\E}$ are separating on $\widebar {\Sol}/\widebar {\G}$, and similar methods can be used to tackle this issue. Thus, for the same reasons outlined in the brief discussion of the separability of $\widebar {\Sol}/\widebar {\G}$, we believe that a full classification of the non-degeneracy of $\widebar {\sigma}$ is presumably possible with additional assumptions on $(M,g)$ using the methods of \cite{Fewster:2012bj} and \cite{HS}, but that a general statement for arbitrary $(M,g)$ is difficult to obtain. In the following Theorem \ref{propo_proptau2} we shall state a necessary condition for non-degeneracy of $\sigma$, in other words, a sufficient condition for degeneracy of $\sigma$.

\subsubsection{Equivalent formulation of the presymplectic space}
\label{sec_solform}
For explicit computations such as the ones we intend to perform for the special case of perturbations in inflation it is advantageous to have an equivalent form of the Poisson bracket $\widebar {\sigma}$ at our disposal. To this avail we split any on-shell configuration $\Theta\in\widebar {\Sol}$ into a ``future part'' $\Theta^+$ and a ``past part'' $\Theta^-$, where
we call $\Theta^\pm\in\Gamma(\V)$ a future/past part of $\Theta\in\Gamma(\V)$ if there exist two Cauchy surfaces $\Sigma^\pm$ such that $\Theta^\pm$ and $\Theta$ coincide on $J^\pm(\Sigma^\pm)$ and $\Theta^\pm$ vanishes on $J^\mp(\Sigma^\mp)$. A possibility to construct a future/past part of $\Theta$ is to pick two Cauchy surfaces $\Sigma^\pm$ of $(M,g)$ which satisfy $\Sigma^-\subset J^-(\Sigma^+)$ and $\Sigma^+\cap\Sigma^-=\emptyset$ and to pick a smooth partition of unity $\chi^++\chi^-=1$ s.t. $\chi^-$ vanishes on $J^+(\Sigma^+)$ and equals 1 on $J^-(\Sigma^-)$. Using this, we define $\Theta^\pm:=\chi^\pm \Theta$. Of course the such constructed future/past part is non-unique, but the difference $\Theta^+-\Theta^{+\prime}$ of any two future parts $\Theta^+$, $\Theta^{+\prime}$ of $\Theta$ has time-like compact support. Using any definition of $\Theta^+$, we construct a bilinear form on $\widebar {\Sol}$ by
\beq\label{eq_def_tau2}\langle \cdot, \cdot \rangle_{\widebar {\Sol}}:\widebar {\Sol}\times\widebar {\Sol} \to \bbR,\qquad (\Theta_1,\Theta_2)\mapsto\langle \Theta_1, \Theta_2 \rangle_{\widebar {\Sol}}:=\langle {\widebar {P}} \Theta^+_1, \Theta_2 \rangle_{\widebar {\V}} \,.\eeq
Employing once more the abstract properties 1-4 of $(\M,\V,\W,{\widebar {P}},\widebar {K})$ demonstrated in Theorem \ref{prop_propLEKGS}, one can prove the following properties of this bilinear form, cf. \cite[Proposition 5.1+Theorem 5.2]{HS}\footnote{In fact, some of the statements in Theorem \ref{propo_proptau2} are slight generalisations of \cite[Proposition 5.1]{HS} or are not spelt out explicitly there, but they can be proved using exactly the same steps used in the proof of said proposition.}.
\begin{theo}\label{propo_proptau2}
The bilinear form $\langle \cdot, \cdot \rangle_{\widebar {\Sol}}$ on $\widebar {\Sol}$ defined in \eqref{eq_def_tau2} has the following properties. 
\begin{enumerate}
\item $\langle \Theta_1, \Theta_2 \rangle_{\widebar {\Sol}}$ is well-defined for all $\Theta_1$, $\Theta_2\in\widebar {\Sol}$ with space-like compact overlapping support. In particular, it is independent of the choice of future part entering the definition \eqref{eq_def_tau2}.
\item $\langle \cdot, \cdot \rangle_{\widebar {\Sol}}$ is antisymmetric.
\item $\langle \cdot, \cdot \rangle_{\widebar {\Sol}}$ is gauge-invariant, i.e. $\langle \Theta_1, \Theta_2+\widebar {K}\varsigma \rangle_{\widebar {\Sol}}=\langle \Theta_1, \Theta_2\rangle_{\widebar {\Sol}}$ for all $\Theta_1,\Theta_2\in\widebar {\Sol}$, $\varsigma\in \Gamma(\W)$ s.t. $\Theta_1$ and $\varsigma$ have space-like compact overlapping support.
\item For all $\Theta\in\widebar {\Sol}$ and all $h\in\Ker_0(\widebar {K}^\dagger)$, $\langle \Theta, G^{\widetilde{P}}h\rangle_{\widebar {\Sol}}=\langle \Theta, h\rangle_{\widebar {\V}} $.
\item In particular, for all $h_1,h_2\in\Ker_0(\widebar {K}^\dagger)$, $\langle G^{\widetilde{P}}h_1, G^{\widetilde{P}}h_2\rangle_{\widebar {\Sol}}=\widebar {\sigma}([h_1],[h_2])$.
\item $G^{\widetilde{P}}$ descends to an isomorphism of presymplectic spaces $G^{\widetilde{P}}:(\widebar {\E},\widebar {\sigma})\to (\widebar {\Sol}_\sc/\widebar {\G}_{\sc,0},\langle\cdot,\cdot\rangle_{\widebar {\Sol}})$.
\item If $\widebar {\G}_{\sc,0}\subsetneq{\widebar {\G}}_\sc$, then $\widebar {\sigma}$ is degenerate, i.e. there exists $[0]\neq[h]\in\widebar {\E}$ s.t. $\widebar {\sigma}([h],\widebar {\E})=0$.
\end{enumerate}
\end{theo}

The fourth property in the above theorem shows that one can view the observable $\widebar {\Sol}/{\widebar {\G}}\ni[\Theta]\mapsto \langle \Theta, h\rangle_{\widebar {\V}} $, i.e. the smeared classical field, equivalently as ``(pre)symplectically smeared classical field'' $\widebar {\Sol}/{\widebar {\G}}\ni[\Theta]\mapsto \langle \Theta, H\rangle_{\widebar {\Sol}}$ with $H\in[H]\in \widebar {\Sol}_\sc/\widebar {\G}_{\sc,0}$. Moreover, because of the fifth (and the ensuing sixth) property, the quantum field algebra can be viewed as generated by the symplectically smeared quantum field $\Theta(H)$, corresponding to the quantization of $\langle \Theta, H\rangle_{\widebar {\Sol}}$, and satisfying the CCR
$$\left[\Theta(H_1),\Theta(H_2)\right]=i\langle H_1,H_2\rangle_{\widebar {\Sol}}\1\,,$$
which may be interpreted as ``equal-time'' CCR on account of the following result.

As anticipated, we will now demonstrate that the bilinear form $\langle\cdot,\cdot\rangle_{\widebar {\Sol}}$, whenever it is defined, can be computed as an integral on an arbitrary but fixed Cauchy surface of the spacetime $(M,g)$. In fact, the integrand of this integral is the ``charge'' of a ``conserved current''. To see this, for any two sections of $\V$, we define a ``current'' $\widebar {j}$ by

\begin{equation}
\label{eq_current}\widebar {j}:\Gamma(\V)\times \Gamma(\V)\to  T^*M
\end{equation}
\begin{align*}(\Theta_1, \Theta_2) \mapsto \widebar {j}_a(\Theta_1,\Theta_2):=&-\frac{4}{\alpha}{\theta_1}^{bc}\nabla_a{\theta_2}_{bc}+\frac{8}{\alpha}{\theta_1}_{a}^{\phantom{c}c}\nabla^b{\theta_2}_{bc}-\frac{4\xi^2\phi^2-2\beta}{\beta\alpha}{\theta_1}_{c}^{\phantom{c}c}\nabla_a{\theta_2}_{d}^{\phantom{c}d}-\\&-\frac{2\xi\phi}{\beta}{\theta_1}_{c}^{\phantom{c}c}\nabla_a \zeta_2-\frac{2\xi\phi}{\beta}\zeta_1\nabla_a{\theta_2}_{c}^{\phantom{c}c}-\frac{\alpha}{\beta}\zeta_1\nabla_a\zeta_2+\\&+\frac{2\xi\left(2-\alpha\right)}{\beta\alpha}{\theta_1}_{c}^{\phantom{c}c}\left(\nabla_a\phi\right)\zeta_2+\frac{4}{\alpha}\zeta_1 \left(\nabla^b\phi\right){\theta_2}_{ab}-``1\leftrightarrow 2"\end{align*}
The covariant divergence / codifferential of this satisfies
\begin{equation}\label{eq_currentconservation}
-\delta \widebar {j}(\Theta_1,\Theta_2) = \nabla^a\widebar {j}_a(\Theta_1,\Theta_2)=\left\langle\left\langle \Theta_1,{\widebar {P}}\Theta_2\right\rangle\right\rangle_{\widebar {\V}} -\left\langle\left\langle {\widebar {P}}\Theta_1,\Theta_2\right\rangle\right\rangle_{\widebar {\V}} 
\end{equation}
where $\langle\langle\cdot, \cdot\rangle\rangle_{\widebar {\V}} $ is the integrand of $\langle\cdot, \cdot\rangle_{\widebar {\V}} $. The existence of a $\widebar {j}$ with this property is related to the formal selfadjointness of ${\widebar {P}}$. Moreover, this property implies in particular $\delta \widebar {j} (\Theta_1,\Theta_2)=0$ if ${\widebar {P}}\Theta_1={\widebar {P}}\Theta_2=0$, which motivates the nomenclature for $\widebar {j}$. This current can be written equivalently as follows. If we define
\beq\label{eq_def_redefinedq2} {\widebar {P}}=:M^{ab}\nabla_a\nabla_b + N^a\nabla_a + L\qquad \widebar {\Gamma}=:F\Gamma \qquad \widebar {\Gamma}^{-1}=F^{-1}\Gamma \eeq
\beq\label{eq_fibrewisered}\langle\cdot, \cdot\rangle_{\V} =:\int\limits_M\vol \langle\langle\cdot, \cdot\rangle\rangle_{\V} \qquad \langle\cdot, \cdot\rangle_{\widebar {\V}} =:\int\limits_M\vol \langle\langle\cdot, \cdot\rangle\rangle_{\widebar {\V}} =\int\limits_M\vol \langle\langle\cdot, F^{-1}\cdot\rangle\rangle_{\V} \eeq
then 
$$\widebar {j}_a(\Theta_1,\Theta_2)= \left\langle\left\langle \Theta_1,M_{a}^{\phantom{a}b}\nabla_b \Theta_2\right\rangle\right\rangle_{\widebar {\V}} -\left\langle\left\langle \Theta_2,M_{a}^{\phantom{a}b}\nabla_b \Theta_1\right\rangle\right\rangle_{\widebar {\V}} +$$$$+\left\langle\left\langle \Theta_1,\left(N_a-F\left(\nabla^b F^{-1}M_{ba}\right)\right) \Theta_2\right\rangle\right\rangle_{\widebar {\V}} \,.$$
Note that the definitions of $M^{ab}$ and $L$ and, consequently, the definition of $\widebar {j}_a$ in this way is not unique. However, two possible $M^{ab}$, ${M^\prime}^{ab}$ lead to $\widebar {j}_a$, $\widebar {j}^\prime_a$ which differ by a co-exact oneform $\delta \omega = \widebar {j}-\widebar {j}^\prime$. This is irrelevant as we are interested only in $\delta \widebar {j}$. A unique $\widebar {j}$ satisfying \eqref{eq_currentconservation} can be specified by defining it via the covariant conjugate momentum of $\Theta$ obtained from the quadratic action $\widebar {S}^{(2)}(\Theta)$. We choose a different current here because it has a slightly shorter form than this ``canonical'' one.

We now pick two arbitrary but fixed $\Theta_1$, $\Theta_2\in\widebar {\Sol}$ with space-like compact overlapping support, an arbitrary but fixed Cauchy surface $\Sigma$ of $(M,g)$ with forward pointing unit normal vector field $n$, and then compute
$$\langle \Theta_1, \Theta_2\rangle_{\widebar {\Sol}}=\langle {\widebar {P}} \Theta^+_1 ,\Theta_2\rangle_{\widebar {\V}} =\int\limits_M \vol\left\langle\left\langle {\widebar {P}} \Theta^+_1,\Theta_2 \right\rangle\right\rangle_{\widebar {\V}} =$$$$=-\!\!\!\!\int\limits_{J^+(\Sigma)}\!\!\!\! \vol\left\langle\left\langle {\widebar {P}} \Theta^-_1,\Theta_2 \right\rangle\right\rangle_{\widebar {\V}} +\!\!\!\!\int\limits_{J^-(\Sigma)} \!\!\!\!\vol\left\langle\left\langle {\widebar {P}} \Theta^+_1,\Theta_2 \right\rangle\right\rangle_{\widebar {\V}} 
 \,,$$
where we have used that ${\widebar {P}}\Theta^+={\widebar {P}}(\Theta-\Theta^-)=-{\widebar {P}}\Theta^-$ for all $\Theta\in\widebar {\Sol}$. Using \eqref{eq_currentconservation} and Stokes' theorem, the two summands in the last expression can be rewritten as
$$\mp\!\!\!\!\int\limits_{J^\pm(\Sigma)} \!\!\!\!\vol\left\langle\left\langle {\widebar {P}} \Theta^\mp_1,\Theta_2 \right\rangle\right\rangle_{\widebar {\V}} =\mp\!\!\!\!\int\limits_{J^\pm(\Sigma)} \!\!\!\!\vol\left(\left\langle\left\langle {\widebar {P}} \Theta^\mp_1,\Theta_2 \right\rangle\right\rangle_{\widebar {\V}} -\left\langle\left\langle  \Theta^\mp_1,{\widebar {P}} \Theta_2 \right\rangle\right\rangle_{\widebar {\V}} \right)=$$
$$=\mp\!\!\int\limits_{J^\pm(\Sigma)} \!\!\!\!\vol\left( -\delta \widebar {j}(\Theta^\mp_1,\Theta_2)\right)=\mp\!\!\!\!\int\limits_{J^\pm(\Sigma)} \!\!\!\!d\ast \widebar {j}(\Theta^\mp_1,\Theta_2)=\int\limits_\Sigma \vols\, n^a\widebar {j}_a(\Theta^\mp_1,\Theta_2)\,.$$
Here $\ast$ denotes the Hodge star operator and while applying Stokes' theorem, we have used the fact that, by our assumptions, the sets $\supp\, \Theta^\pm_1\cap \supp\, \Theta_2 \cap J^{\mp}(\Sigma)$ are compact. Finally, note that a partial sign flip occurred in last step by matching the chosen orientations of $\Sigma$. Summing up, we have proved the following.

\begin{lem}\label{prop_solcurrent}For every $\Theta_1$, $\Theta_2\in\widebar {\Sol}$ with space-like compact overlapping support, and every Cauchy surface $\Sigma$ of $(M,g)$ with forward pointing unit normal vector field $n$
\beq\label{eq_tau_current}\langle \Theta_1, \Theta_2\rangle_{\widebar {\Sol}}=\int\limits_\Sigma \vols\, n^a \widebar {j}_a(\Theta_1,\Theta_2)\,,\eeq
where $\widebar {j}_a(\Theta_1,\Theta_2)$ is defined in \eqref{eq_current}.
\end{lem}

\subsubsection{Removing the field redefinition}
\label{sec_removingredefinition}
We close our general treatment of the linearised Einstein-Klein-Gordon system by demonstrating that the field redefinition $\Theta=\widebar {\Gamma}=F\Gamma$ is irrelevant for the results of our constructions. To this avail, we recall the relations \eqref{eq_def_redefinedq} between the original and redefined quantities and  define

$$\Sol:= \{\Gamma\in \Gamma(\V)\,|\, P \Gamma=0\}\qquad \Sol_\text{sc}:= \Sol\cap \Gamma_\sc(\V)$$
$${\G}:=K\left[\Gamma(\W)\right]\qquad {\G}_\text{sc}:= {\G}\cap \Gamma_\sc(\V)\qquad \G_{\sc,0}:=K[\Gamma_\sc(\W)] $$
\beq\label{eq_def_origquants}\Ker_0(K^{\dagger}):=\{h\in\Gamma_0(\V)\,|\,K^{\dagger} h=0\}\qquad \E:= \Ker_0(K^{\dagger})/P[\Gamma_0(\V)]\eeq
$$ \sigma:\E\times \E\to\bbR,\quad ([h_1],[h_2])\mapsto \sigma([h_1],[h_2]):=\langle h_1, F^{-1}G^{\widetilde {P}} h_2 \rangle_\V$$
$$\langle \cdot,\cdot\rangle_\Sol: \Sol\times\Sol\to\bbR,\quad (\Gamma_1,\Gamma_2)\mapsto \langle \Gamma_1,\Gamma_2\rangle_\Sol:=\langle P \Gamma^+_1,\Gamma_2\rangle_\V $$
$$
j:\Gamma(\V)\times \Gamma(\V)\to  T^*M\qquad
(\Gamma_1,\Gamma_2)\mapsto j(\Gamma_1,\Gamma_2):=\widebar {j}(\widebar {\Gamma}_1,\widebar {\Gamma}_2)\,.
$$

A straightforward computation shows that $K^{\dagger}=\widebar {K}^\dagger$. Moreover, as the redefinition operator $F$ is invertible by our assumptions that $\alpha=1-\xi\phi^2$ and $\beta=1+(6\xi-1)\xi\phi^2$ have definite sign on all $M$, we have $P[\Gamma_0(\V)]={\widebar {P}}[\Gamma_0(\V)]$. Thus $\Ker_0(K^{\dagger})=\Ker_0(\widebar {K}^{\dagger})$, $\E=\widebar {\E}$. Moreover, since for $\Theta=F\Gamma$, $\langle \Theta,h\rangle_{\widebar {\V}} =\langle \Gamma,h\rangle_\V$, the classical observables $\widebar {\Sol}/{\widebar {\G}}\ni[\Theta]\mapsto \langle \Theta,h\rangle_{\widebar {\V}} $ and $\Sol/{\G}\ni[\Gamma]\mapsto \langle \Gamma,h\rangle_\V$ are the same for all $h\in[h]\in\widebar {\E}$. A direct computation also shows that $\widebar {\sigma}=\sigma$ and $\langle \Gamma_1,\Gamma_2\rangle_\Sol=\langle \widebar {\Gamma}_1,\widebar {\Gamma}_2\rangle_{\widebar {\Sol}}$, and thus all presymplectic spaces under consideration are isomorphic
$$(\widebar {\E},\widebar {\sigma})\simeq (\widebar {\Sol}_\sc/\widebar {\G}_{\sc,0},\langle\cdot,\cdot\rangle_{\widebar {\Sol}})\simeq (\Sol_\sc/\G_{\sc,0},\langle\cdot,\cdot\rangle_\Sol)\simeq (\E,\sigma)\,.$$
One can show that the operator $F^{-1}\circ G^{\widetilde {P}}$ in the definition of $\sigma$ is the causal propagator of $\widetilde P \circ F=P+T\circ K^\dagger \circ F$, which can be interpreted as a gauge-fixed version of $P$ and has a well-defined Cauchy problem although it is not normally hyperbolic. Thus we could have in principle quantized directly the system $(\M,\V,\W,P,K)$ without quantizing $(\M,\V,\W,{\widebar {P}},\widebar {K})$ as an in-between step. The reason why we chose this indirect route is that we wanted to use the results of \cite{HS}, which are not directly applicable to  $(\M,\V,\W,P,K)$, because e.g. $K^\dagger\circ K$ does not have a well-defined Cauchy problem.

To close this section, we consider the original current $j$ as a preparation for the next section. Note that the previous observations  and Lemma \ref{prop_solcurrent} imply that this Lemma holds also for $\langle\cdot,\cdot\rangle_\Sol$ and $j$. Finally, a direct computation shows
\begin{gather}j_a(\Gamma_1,\Gamma_2)=-\frac{\alpha}{4}{\gamma_1}^{bc}\nabla_a{\gamma_2}_{bc}+\frac{\alpha}{2}{\gamma_1}_{a}^{\phantom{c}c}\nabla^b{\gamma_2}_{bc}+\frac{\alpha}{4}{\gamma_1}_{c}^{\phantom{c}c}\nabla_a{\gamma_2}_{d}^{\phantom{c}d}-\frac{\alpha}{4}{\gamma_1}_{c}^{\phantom{c}c}\nabla_d{\gamma_2}_{a}^{\phantom{c}d}-\frac{\alpha}{4}{\gamma_1}_{ab}\nabla^b{\gamma_2}_{c}^{\phantom{c}c}+\notag\\\label{eq_origcurrent}+\xi\phi\varphi_1\nabla_c{\gamma_2}_{a}^{\phantom{c}c}
+\xi\phi{\gamma_1}_{a}^{\phantom{c}c}\nabla_c \varphi_2
-\xi\phi{\gamma_1}_{c}^{\phantom{c}c}\nabla_a \varphi_2-\xi\phi\varphi_1\nabla_a{\gamma_2}_{c}^{\phantom{c}c}-\varphi_1\nabla_a\varphi_2+\\+\frac{1-2\xi}{2}{\gamma_1}_{c}^{\phantom{c}c}\left(\nabla_a\phi\right)\varphi_2+\left((1-\xi)\varphi_1+\xi\phi{\gamma_1}_{c}^{\phantom{c}c} \right) \left(\nabla^b\phi\right){\gamma_2}_{ab}-\xi\phi{\gamma_1}_{ab} \left(\nabla^c\phi\right){\gamma_2}_{c}^{\phantom{c}b}-``1\leftrightarrow 2"\,.\notag\end{gather}


\section{Quantization of perturbations in Inflation}
\label{sec_inflation}

The basic idea of Inflation, see e.g. the monographs \cite{Ellis, Mukhanov:2005sc, Straumann:2005mz}, is the assumption that the early universe underwent a phase of exponential expansion driven by one (or several) scalar fields. This expansion is thought to have homogenised the universe respectively all of its matter-energy content. Afterwards, the quantized perturbations of the metric and scalar field are believed to have been the seeds for the small scale inhomogeneities of the universe we see today. Mathematically this amounts to consider solutions of the Einstein-Klein-Gordon equation \eqref{eq_fullcoupled} such that the background spacetime $(M,g)$ is a Friedmann-Lema\^itre-Robertson-Walker (FLRW) spacetime
$$M=I\times \bbR^3\subset \bbR^4\,, \qquad g=a(\tau)^2(-d\tau^2 + d\vec{x}^2)$$
where we here consider only the case where the spatial slices are diffeomorphic to $\bbR^3$, i.e. flat. $a(\tau)$ denotes the scale factor, $\tau$ is conformal time and we denote derivatives w.r.t. $\tau$ by $^\prime$ throughout this section. Moreover, Latin indices in the middle of the alphabet $i,j,k,...$ denote spatial indices in the FLRW coordinates, and these indices are always raised by means of the (inverse of the) Euclidean metric $\delta_{ij}$ and not by the induced metric $a^2 \delta_{ij}$. An important quantity is 
$$\H:=\frac{a^\prime}{a}=\frac{\partial_\tau a}{a}\qquad\Rightarrow\qquad R=\frac{6(\H^\prime+\H^2)}{a^2}\,.$$
FLRW-solutions of the Einstein-Klein-Gordon equations \eqref{eq_fullcoupled} can only exist if the scalar field $\phi$ does not depend on the spatial coordinates, i.e. $\phi=\phi(\tau)$. Thus, for solutions of FLRW-type, the Einstein-Klein-Gordon equations simplify to (recall $\alpha=1-\xi\phi^2$)
$$6\alpha\H^2-(\phi^\prime)^2-12\xi \H\phi\phi^\prime -2a^2V=0$$
$$2\alpha(\H^2+2\H^\prime)+(1-4\xi)(\phi^\prime)^2-4\xi\H\phi\phi^\prime-4\xi\phi\phi^{\prime\prime}-2a^2V=0$$
$$\phi^{\prime\prime}+2\H\phi^\prime +6\xi(\H^\prime+\H^2)\phi+a^2 \partial_\phi V=0\,.$$
In the previous section we have indicated how to quantize perturbations of the Einstein-Klein-Gordon system in a gauge-invariant way on any background $(M,g,\phi)$ and thus in particular on backgrounds of FLRW type. Our aim in this section is to compare this general construction with the usual approach to the quantization of perturbations in Inflation, see e.g. \cite{Bardeen:1980kt, Ellis, Mukhanov:1990me, Mukhanov:2005sc, Straumann:2005mz}, which has been recently investigated from the point of view of algebraic quantum field theory in \cite{Eltzner:2013soa}. While these works consider only the case of minimal coupling $\xi=0$, there are quite a lot of works on inflationary perturbations in the non-minimally coupled case. It is barely possible to cite them all at this point, so we would like to mention only the early works \cite{Futamase:1987ua, Salopek:1988qh,Makino:1991sg}. 

\subsection{Classification of perturbations}
\label{sec_classpert}

The starting point of the usual treatment of perturbations in inflation is to use the high symmetry of the spatial slices of FLRW spacetimes, given by the Euclidean group, in order to split the metric perturbations into scalar, vector and tensor components. The existence and uniqueness of this split is directly related to the existence and uniqueness of solutions to the Poisson equation on $\bbR^3$
$$\Delta u=\vec{\nabla}^2 u=f\,.$$
While it is probably widely known that a unique solution $u$ which vanishes at infinity exists if $f$ is vanishing at infinity (and that $u$ is smooth if $f$ is smooth), it is maybe less widely known that a solution $u$ (which is unique up to harmonic functions $\Delta u=0$) exists also for $f$ which do not vanish at infinity, see e.g. \cite[Corollary 10.6.8]{Hormander}. Hence we explicitly state this splitting result\footnote{We state the result only for smooth functions, but it holds also for distributions, cf. \cite[Section 10.6]{Hormander}.}.
\begin{propo}\label{prop_splitbasic}The following results hold for smooth vector and tensor fields on $\bbR^3$.
\begin{enumerate}
\item Every vector field $f\in C^\infty(\bbR^3,\bbR^3)$ can be split as
$$f_i=\partial_i B + V_i\,,\qquad B\in C^\infty(\bbR^3,\bbR)\,,\quad V\in  C^\infty(\bbR^3,\bbR^3)\,,\quad \partial^i V_i=0\,,$$
where $B$ and $V$ are defined as
$$\Delta B:= \partial^if_i\,,\qquad V_i:=f_i-\partial_i B\,.$$
The splitting is unique if $f$ vanishes at infinity and one requires that $B$ and $V$ vanish at infinity. $B$ and $V$ have compact support if and only if $\partial^if_i=\Delta s$ with $s$ compactly supported.
\item Every symmetric tensor field $f\in C^\infty(\bbR^3,\bigvee^2\bbR^3)$ can be split as
$$f_{ij}=\partial_i \partial_j E + \delta_{ij}D+2 \partial_{(i}W_{j)}+T_{ij}\,,\quad E,D\in C^\infty(\bbR^3,\bbR)$$
$$W\in  C^\infty(\bbR^3,\bbR^3)\,,\quad \partial^i W_i=0\,,\quad T\in C^\infty(\bbR^3,{\textstyle\bigvee^2}\bbR^3)\,,\quad T_i^i=0\,,\quad \partial^i T_{ij}=0\,,$$
where $E,D,W,T$ are defined as
$$2\Delta D := \Delta f_i^i-\partial^i\partial^j f_{ij}\,,\quad \Delta E := f^i_i-3 D\,,$$
$$\Delta W_i := \partial^j f_{ij}-\partial_i \Delta E - \partial_i D\,,\quad T_{ij}:=f_{ij}-\partial_i \partial_j E + \delta_{ij}D+2 \partial_{(i}W_{j)}\,.$$
The splitting is unique if $f$ vanishes at infinity and one requires that $D,E,W,T$ vanish at infinity. 
\end{enumerate}
\end{propo}
The non-uniqueness of the splitting in case of vector and tensor fields $f$ not vanishing at infinity implies that e.g. a vector field $f$ can be simultaneously of ``scalar'' and ``vector'' type if $f_i=\partial_i B$ with $\Delta B=0$, $B\neq 0$. Similarly, a symmetric tensor field $f$ can be simultaneously of  ``scalar'', ``vector'' and ``tensor'' type if $f_{ij}=\partial_i\partial_j E$ with $\Delta E=0$, $E\neq 0$.

Because $\Delta$ commutes with conformal time derivatives and is invertible on functions which vanish at spatial infinity (and because $\Delta$ and its inverse are continuous), the above splitting can be directly applied to the $\gamma_{0i}$ and $\gamma_{ij}$ components of all field configurations $\Gamma = (\gamma_{ab},\varphi)^T \in  \Gamma(\V)$ and to the $\varsigma_i$ components of all gauge transformation parameters $\varsigma\in\Gamma(\W)$ which vanish at spatial infinity. The resulting scalar, vector and tensor parts of such $\Gamma$ and $\varsigma$ are then smooth functions on $M$ themselves. For $\Gamma$ and $\varsigma$ not vanishing at spatial infinity it is less clear whether a splitting in smooth scalar, vector and tensor parts exist because the non-uniqueness of the splitting for each $\tau$ impedes a proof that the splitting is continuous in $\tau$\footnote{I would like to thank Marco Benini for pointing out this issue to me.}. Presumably it is possible to prove the existence of a splitting smooth in $\tau$ by repeating the proof of solvability of the Poisson equation via compact exhaustions of $\bbR^3$ in \cite[Section 10.6.]{Hormander} in a manner which is uniform in $\tau$. We shall however refrain from doing this here, as the existence of a smooth splitting for $\Gamma$ and $\varsigma$ which vanish at spatial infinity will be sufficient for our results. Notwithstanding, we shall provide a few definitions and statements covering the split form of general sections which are valid even if this split form is not proven to exist for all smooth sections.

With this in mind, we define the splitting here directly for the field perturbations themselves rather than for the redefined perturbations $\Theta=\widebar {\Gamma}$, in order to match the conventions with other works on perturbations in Inflation. In general we will work exclusively with original rather than redefined quantities in this section as we have seen in Section \ref{sec_removingredefinition} that the field redefinition is invertible and irrelevant for the classical and quantum field theory of the linearised Einstein-Klein-Gordon system.
That said, we split the metric perturbations $\gamma_{ab}$ as
\beq\label{eq_def_split}\gamma_{ab}=a(\tau)^2\begin{pmatrix}
-2 A & \left(-\partial_i B+V_i\right)^T\\
-\partial_i B+V_i & 2\left(\partial_i\partial_j E
+\delta_{ij}D+\partial_{(i} W_{j)}+T_{ij}\right)
\end{pmatrix}\,,\eeq
where
$$A,B,D,E\in C^\infty(M,\bbR)\,,\quad V, W\in  C^\infty(M,\bbR^3)\,,\quad \partial^i V_i=\partial^i W_i=0$$
$$T\in C^\infty(M,{\textstyle\bigvee^2}\bbR^3)\,,\quad T_i^i=0\,,\quad \partial^i T_{ij}=0$$
and the signs and factors are conventional and convenient. Similarly, we split the gauge transformation parameters $\varsigma\in\Gamma(\W)$ as
\beq\label{eq_def_splitgauge}\varsigma_a=a(\tau)^2\begin{pmatrix}-r\\\partial_i s+z_i\end{pmatrix},\quad r,s\in C^\infty(M,\bbR)\,,\quad z\in C^\infty(M,\bbR^3)\,,\quad\partial^i z_i=0\,.\eeq
We say that $\Gamma = (\gamma_{ab},\varphi)^T\in \Gamma(\V)$ ...
\begin{itemize}
\item[...] is of scalar type if $\gamma_{ab}$ can be split as \eqref{eq_def_split} with $V_i=W_i=T_{ij}=0$.
\item[...] is of vector type if $\varphi=0$ and $\gamma_{ab}$ can be split as \eqref{eq_def_split} with $A=B=D=E=T_{ij}=0$.
\item[...] is of tensor type if $\varphi=0$ and $\gamma_{ab}$ can be split as \eqref{eq_def_split} with $A=B=D=E=V_i=W_i=0$.
\end{itemize}
And define $\varsigma\in\Gamma(\W)$ of scalar/vector type analogously. This motives the following definitions of section spaces, where ${\cal X}$ stands for an arbitrary trivial vector bundle over $M$.
$$\Gamma_\infty({\cal X}):=\{\Gamma\in\Gamma({\cal X})\,|\,\partial_{i_1}\cdots\partial_{i_n}\Gamma(\tau,\vec{x}) \text{ vanishes for $|\vec{x}|\to\infty$ for all $n\ge 0$}\}$$
\beq\label{eq_def_svtspaces}\Gamma^{S/V/T}(\V):=\{\Gamma\in\Gamma(\V)\,|\,\Gamma \text{ is of scalar/vector/tensor type}\}$$
$$\Gamma^{S/V}(\W):=\{\varsigma\in\Gamma(\W)\,|\,\varsigma \text{ is of scalar/vector type}\}\eeq
$$\Gamma^{S/V/T}_{\infty/\sc/0}(\V):= \Gamma_{\infty/\sc/0}(\V)\cap\Gamma^{S/V/T}(\V)\qquad \Gamma^{S/V}_{\infty/\sc}(\W):= \Gamma_{\infty/\sc}(\W)\cap\Gamma^{S/V}(\W)$$
The strong condition that elements in $\Gamma_\infty(\V)$ and $\Gamma_\infty(\W)$ vanish at spatial infinity with all spatial derivatives is not necessary as in the following we shall only need the vanishing of at most two derivatives. However imposing the stronger condition simplifies the discussion as one does not have to track the number of derivatives vanishing at infinity. Similarly whenever we speak of ``vanishing at spatial infinity'' in the following we shall mean ``vanishing at spatial infinity with all derivatives vanishing at spatial infinity''.

Naturally, we have e.g. $\Gamma_0(\V)\subset \Gamma_\sc(\V)\subset\Gamma_\infty(\V)$. Proposition \ref{prop_splitbasic} and the subsequent discussion imply
$$\Gamma_\infty(\V)=\Gamma^S_\infty(\V)\oplus\Gamma_\infty^V(\V)\oplus \Gamma_\infty^T(\V)\,,\qquad \Gamma_\infty(\W)=\Gamma^S_\infty(\W)\oplus\Gamma_\infty^V(\W)\,,$$
$$\Gamma^S(\V)\cap\Gamma^V(\V)\cap \Gamma^T(\V)\neq \{0\}\,,$$
\beq\label{eq_splitintersect}\Gamma_\infty^S(\V)\cap\Gamma_\infty^V(\V)=\Gamma_\infty^S(\V)\cap\Gamma_\infty^T(\V)=\Gamma_\infty^V(\V)\cap\Gamma_\infty^T(\V) =\{0\}\,,\eeq
$$\Gamma_{0/\sc}^S(\V)\oplus\Gamma_{0/\sc}^V(\V)\oplus \Gamma_{0/\sc}^T(\V)\subsetneq  \Gamma_{0/\sc}(\V)\,.$$
In particular, the scalar/vector/tensor part of a section with (space-like) compact support does in general not have a (space-like) compact support, whereas is is presumably true but not proven that $\Gamma^S(\V)\oplus\Gamma^V(\V)\oplus \Gamma^T(\V)=\Gamma(\V)$ and $\Gamma^S(\W)\oplus\Gamma^V(\W)=\Gamma(\W)$. The uniqueness of the splitting for $\Gamma_\infty(\V)$ and $\Gamma_\infty(\W)$ implies that there exist surjective projectors
\beq\label{eq_def_proj}\PP^{S/V/T}_\V:\Gamma_\infty(\V)\to\Gamma^{S/V/T}_\infty(\V)\,,\qquad \PP^{S/V}_\W:\Gamma_\infty(\W)\to\Gamma^{S/V}_\infty(\W)\,,\eeq
which can be written explicitly in terms of spatial derivatives and the (chosen) inverse of $\Delta$.

In order to distinguish configurations of the perturbation variables, i.e. classical states, from test sections, meaning labels of observables, we use the following notation for the splitting of test sections $h=(k_{ab},f)^T\in\Gamma_0(\V)$.
\beq\label{eq_def_split0}k_{ab}=a(\tau)^2\begin{pmatrix}
-2 c & \left(-\partial_i b+v_i\right)^T\\
-\partial_i b+v_i & 2\left(\partial_i\partial_j e
+\delta_{ij}d+\partial_{(i} w_{j)}+t_{ij}\right)
\end{pmatrix}\,,\eeq
$$c,b,d,e\in C_\infty^\infty(M,\bbR)\,,\quad v, w\in  C_\infty^\infty(M,\bbR^3)\,,\quad \partial^i v_i=\partial^i w_i=0$$
$$t\in C_\infty^\infty(M,{\textstyle\bigvee^2}\bbR^3)\,,\quad t_i^i=0\,,\quad \partial^i t_{ij}=0\,.$$ 
Recall that, as discussed in Section \ref{sec_removingredefinition}, a test section $h\in\Gamma_0(\V)$ labels an observable independent of whether or not one considers original or redefined perturbation variables.

Note that the vanishing of the splitting components at spatial infinity and the fact that the kernel of $\partial_i$ consists of locally constant functions imply
\beq\label{eq_bdecompsupp}\Gamma\in\Gamma^S_\sc(\V)\;\;\Rightarrow\;\; B,D,E\in C^\infty_\sc(M,\bbR)\,,\qquad h\in\Gamma^S_0(\V)\;\;\Rightarrow\;\; b,d,e\in C^\infty_0(M,\bbR)\,,\eeq
$$\Gamma\in\Gamma^V_\sc(\V)\;\;\Rightarrow\;\; V,W\in C^\infty_\sc(M,\bbR^3)\,,\qquad h\in\Gamma^V_0(\V)\;\;\Rightarrow\;\; v,w\in C^\infty_0(M,\bbR^3)\,.$$ By a similar argument, using that $\Gamma\in\Gamma_{\sc/0}(\V)$ implies that both $\gamma^i_i$ and $\partial^i \gamma_{ij}$ have (space-like) compact support, one can obtain
\beq\label{eq_bdecompsupp2}\Gamma\in\Gamma_{\sc/0}(\V)\;\;\Rightarrow\;\; A,D,\Delta E,\varphi\in C^\infty_{\sc/0}(M,\bbR)\,.\eeq

Finally, we observe the following important result.

\begin{lem}\label{prop_nondegtensor}On FLRW backgrounds, $\langle\cdot,\cdot\rangle_\V$ and $\langle\cdot,\cdot\rangle_\W$ satisfy the following relations.
\begin{enumerate}
\item The splitting $\Gamma_\infty(\V)=\Gamma^S_\infty(\V)\oplus\Gamma_\infty^V(\V)\oplus \Gamma_\infty^T(\V)$ is orthogonal w.r.t. to $\langle\cdot,\cdot\rangle_\V$ and the splitting $\Gamma_\infty(\W)=\Gamma^S_\infty(\W)\oplus\Gamma_\infty^V(\W)$ is orthogonal w.r.t. to $\langle\cdot,\cdot\rangle_\W$. In particular, $\PP^{S/V/T}_\V$ ($\PP^{S/V}_\W$) is formally selfadjoint w.r.t. $\langle\cdot,\cdot\rangle_\V$ ($\langle\cdot,\cdot\rangle_\W$), that is, for any $\varsigma_1$, $\varsigma_2\in\Gamma_\infty(\W)$, $\Gamma_1$, $\Gamma_2\in\Gamma_\infty(\V)$,
$$\left\langle \Gamma_1, \PP^{S/V/T}_\V\Gamma_2\right\rangle_\V<\infty\quad\Rightarrow \quad  \left\langle \Gamma_1, \PP^{S/V/T}_\V\Gamma_2\right\rangle_\V= \left\langle \PP^{S/V/T}_\V \Gamma_1, \Gamma_2\right\rangle_\V\,,$$
$$\left\langle \varsigma_1, \PP^{S/V}_\W\varsigma_2\right\rangle_\W<\infty\quad\Rightarrow \quad \left\langle \varsigma_1, \PP^{S/V}_\W\varsigma_2\right\rangle_\W= \left\langle \PP^{S/V}_\W\varsigma_1, \varsigma_2\right\rangle_\W\,.$$
Moreover, $\langle \Gamma^S(\V),\Gamma^{V/T}_0(\V)\rangle_\V=\langle \Gamma^V(\V),\Gamma^{S/T}_0(\V)\rangle_\V=\langle \Gamma^T(\V),\Gamma^{S/V}_0(\V)\rangle_\V=\{0\}$ and $\langle \Gamma^S(\W),\Gamma^{V}_0(\W)\rangle_\W=\langle \Gamma^V(\W),\Gamma^{S}_0(\W)\rangle_\W=\{0\}$.
\item $\langle\cdot,\cdot\rangle_\V$ is non-degenerate on $\Gamma^S_\infty(\V)\times\Gamma^S_\infty(\V)$, $\Gamma^V_\infty(\V)\times\Gamma^V_\infty(\V)$ and $\Gamma^T_\infty(\V)\times\Gamma^T_\infty(\V)$ and $\langle\cdot,\cdot\rangle_\W$ is non-degenerate on $\Gamma^S_\infty(\W)\times\Gamma^S_\infty(\W)$ and $\Gamma^V_\infty(\W)\times\Gamma^V_\infty(\W)$, i.e.
$$\Gamma_1\in\Gamma^{S/V/T}_\infty(\V)\quad\text{and}\quad\langle \Gamma_1, \Gamma_2\rangle_\V=0\quad\forall\; \Gamma_2\in\Gamma^{S/V/T}_0(\V)\qquad\Rightarrow\qquad \Gamma_1=0\,,$$
$$\varsigma_1\in\Gamma^{S/V}_\infty(\W)\quad\text{and}\quad\langle \varsigma_1, \varsigma_2\rangle_\W=0\quad\forall\; \varsigma_2\in\Gamma^{S/V}_0(\W)\qquad\Rightarrow\qquad \varsigma_1=0\,.$$
\end{enumerate}
\end{lem}
\begin{proof}
{Proof of 1}: The statement follows by partial integration from
$$\phantom{\partial^{(i}W^{j)}_1}\langle\Gamma_1,\Gamma_2\rangle_\V=\int\limits_M\vol\left[ 4 A_1 A_2-2\left(\partial^i B_1 - V^i_{1}\right)\left(\partial_i B_2 - V_{2i}\right)+\phantom{\partial^{(i}W^{j)}_1}\right.$$$$\left.+4\left(\partial^i\partial^jE_1+\delta^{ij}D_1+\partial^{(i}W^{j)}_1 + T_1^{ij}\right)\left(\partial_i\partial_jE_2+\delta_{ij}D_2+\partial_{(i}{W_{2j)}}+ T_{2ij}\right)\right]\,,$$
$$\langle\varsigma_1,\varsigma_2\rangle_\W=\int\limits_M\vol\left[a^2r_1 r_2+a^2\left(\partial^i s_1+z_1^i\right)\left(\partial_i s_2+z_{2i}\right)\right]\,.$$
{Proof of 2}: We consider only the tensor case on $\V$, the other cases can be proven in the same fashion. If $(\Gamma,h)\in\Gamma^T_\infty(\V)\times\Gamma^T_0(\V)$, then $$\langle\Gamma,h\rangle_\V=\int\limits_M\vol\; T^{ij}t_{ij}\,.$$
We can replace $t_{ij}$ by $\Delta^2 f_{ij}$ for $f\in C^\infty_0(M,\bbR^3\otimes\bbR^3)$ arbitrary, because the antisymmetric part of $f$ does not contribute to the integral, and because the splitting of $\Delta^2 f_{(ij)}$ gives a compactly supported tensor part by Proposition \ref{prop_splitbasic} and the scalar and vector parts do not matter by orthogonality. By non-degeneracy of $(f_1,f_2)\mapsto \int_M\vol\,{f_1}^{ij}{f_2}_{ij}$ for arbitrary $f_1$, $f_2\in C^\infty(M,\bbR^3\otimes\bbR^3)$ with compact overlapping support and the fact that the only solution of the Laplace equation vanishing at spatial infinity is zero, $\langle\Gamma, h\rangle_\V=0$ for $\Gamma\in\Gamma^T_\infty(\V)$ and all $h\in\Gamma^T_0(\V)$ thus implies $\Gamma=0$.
\end{proof}

\subsection{Gauge invariant variables and splitting of the equations}
\label{sec_spliteom}

Using \eqref{eq_def_split} and \eqref{eq_def_splitgauge} one finds that, under a gauge transformation $\Gamma\mapsto \Gamma+K \varsigma$, the components of $\Gamma$ transform as
   \beq \label{eq_gaugetrafoFLRW}A\mapsto A+(\partial_\tau+\H)r\,,\qquad B\mapsto B+r-s^\prime\,,\qquad D\mapsto D+\H r\,,\qquad E\mapsto E+s\,,\eeq
$$\varphi\mapsto\varphi +  \phi^\prime r\,,\qquad V_i\mapsto V_i+z_i^\prime \,,\qquad W_i\mapsto W_i+ z_i\,,\qquad T_{ij}\mapsto T_{ij}\,,$$
We see that gauge transformations preserve the type of a section, i.e. $$K\left[\Gamma^{S/V}(\W)\right]\subset\Gamma^{S/V}(\V)\,,$$
in fact
\beq\label{eq_invK}\PP_\V^{S/V}\circ K|_{\Gamma_\infty(\W)}=K\circ \PP_\W^{S/V}\qquad \PP_\V^{T}\circ K|_{\Gamma_\infty(\W)}=0\,.\eeq
 This implies in particular that the tensor components of a configuration are already gauge-invariant.

We can also directly see that the conformal gauge\footnote{The name conformal gauge is motivated by the fact that this gauge is invariant under gauge transformations $\Gamma\mapsto \Gamma+K\varsigma$ with $\varsigma$ a constant multiple of the conformal Killing vector $\partial_\tau=(1,0,0,0)^T$.} $B=0$, $E=0$ is always possible by choosing $\varsigma$ such that $s=-E$, $r=-B-E^\prime$. Moreover, we can clearly set either $W_i$ or $V_i$ to zero by choosing $z_i=-W_i$ or $z_i(\tau,\vec{x})=-\int^\tau_{\tau_0}d\tau_1 \,V_i(\tau_1,\vec{x})$. These gauge conditions can be satisfied by means of a gauge transformation preserving the decay/support properties in spatial directions. A further possibility is the synchronous gauge $A=B=V_i=0$, i.e. $\gamma_{0a}=0$. This can be achieved e.g. by performing first a gauge transformation $\Gamma\mapsto\Gamma+K\varsigma$ with $a^2\partial_\tau \varsigma_i /a^2=-\gamma_{0i}$ and $\varsigma_0=0$ and then a gauge transformation with $r(\tau,\vec{x}) = -1/a(\tau)\int^\tau_{\tau_0}d\tau_1\, a(\tau_1) A(t,\vec{x})$, $s(\tau,\vec{x}) =\int^\tau_{\tau_0}d\tau_1\, r(\tau_1,\vec{x})$, $z_i=0$. Proceeding in this way, we see that, for $\Gamma\in\Gamma_{\sc/\infty}(\V)$, the synchronous gauge condition can be satisfied by means of a single gauge transformation with parameter $\varsigma\in\Gamma_{\sc/\infty}(\W)$. We combine these observations into the following lemma.
\begin{lem}\label{prop_confgauge}The following results hold for all $\Gamma\in\Gamma_\infty(\V)$.
\begin{enumerate}
\item There exists $\varsigma\in\Gamma_\infty(\W)$ such that $\Gamma+K\varsigma$ can be split as \eqref{eq_def_split} with $B=E=0$. If $\Gamma\in\Gamma^S_{\sc/0}(\V)$, $\varsigma$ can be chosen in $\Gamma^S_{\sc/0}(\W)$.
\item There exists $\varsigma\in\Gamma_\infty(\W)$ such that $\Gamma+K\varsigma$ can be split as \eqref{eq_def_split} with $V_i=0$ or $W_i=0$. If $\Gamma\in\Gamma^V_{\sc}(\V)$, $\varsigma$ can be chosen in $\Gamma^V_\sc(\W)$.
\item There exists $\varsigma\in\Gamma_\infty(\W)$ such that $\Gamma+K\varsigma$ can be split as \eqref{eq_def_split} with $A=B=V_i=0$. If $\Gamma\in\Gamma_{\sc}(\V)$, $\varsigma$ can be chosen in $\Gamma_{\sc}(\W)$.
\end{enumerate}
\end{lem}

The fact that $K$ commutes with the splitting \eqref{eq_invKdagger} and Lemma \ref{prop_nondegtensor} directly imply
\beq\label{eq_invKdagger}\PP_\W^{S/V}\circ K^\dagger|_{\Gamma_\infty(\V)}=K^\dagger\circ \PP_\V^{S/V}\,.\eeq
Thus the gauge-invariance condition $K^\dagger h=0$ for observables is satisfied if and only if the scalar and vector parts of $h$ satisfy this condition individually, viz.
\beq\label{eq_KdaggerFLRW}
K^\dagger h=0\qquad \Leftrightarrow \qquad\left\{\begin{array}{c}
-4(\partial_\tau+3\H)c + 2 \Delta b + 4\H(3d+\Delta e)+\phi^\prime f=0\\
\left(\partial_\tau + 4 \H\right)b+2\left(d+\Delta e\right)=0\\
\left(\partial_\tau + 4 \H\right)\vec{v}-\Delta\vec{w}^\prime=0
\end{array}\right..
\eeq

While the tensor components of the perturbation variables are already gauge-invariant, one can combine the scalar and vector components into the following well-known gauge-invariant quantities, where $\Psi$ and $\Phi$ are the so-called Bardeen potentials.
\beq\label{eq_def_Bardeen}\Psi:= A-(\partial_\tau + \H)(B+E^\prime)\qquad\Phi:= D-\H(B+E^\prime)\eeq$$\chi:= \varphi - \phi^\prime(B+E^\prime)\qquad X_i:=W_i^\prime-V_i$$

One can check that the equation of motion operator $P:\Gamma(\V)\to\Gamma(\V)$ is also compatible with the splitting\footnote{Essentially this is due to the fact that there are no non-trivial background vector fields $f\in C^\infty(M,\bbR^3)$, whereas the only non-trivial background tensor field $f\in C^\infty(M,\bbR^3\otimes \bbR^3)$ is the identity matrix.}, cf. Section \ref{sec_fullFLRW}, that is
\beq \label{eq_invP}P\left[\Gamma^{S/V/T}(\V)\right]\subset\Gamma^{S/V/T}(\V)\,,\qquad\PP^{S/V/T}_\V\circ P|_{\Gamma_\infty(\V)}=P\circ \PP^{S/V/T}_\V\,.\eeq Thus in particular the equations of motion $P\Gamma=0$ decouple for $\Gamma$ which vanish at spatial infinity and for $P\Gamma=0$ to hold the scalar, vector and tensor components of $\Gamma$ must satisfy individual equations of motion. The full scalar, vector and tensor parts of the equation $P\Gamma=0$, expressed in terms of the gauge invariant components $\Phi$, $\Psi$, $\chi$, $X_i$, $T_{ij}$ of $\Gamma$ are displayed in Section \ref{sec_fullFLRW}. For the vector and tensor degrees of freedom these imply
\beq\label{eq_EOMvec}\Delta X_i=0\qquad (\partial_\tau+2\H)\frac{\alpha}{4}X_i=0\eeq
\beq\label{eq_EOMten}P^TT_{ij}:=\frac{1}{a^2}\left((\partial_\tau+2\H)\alpha\partial_\tau- \alpha\Delta \right)T_{ij}=0\eeq
i.e. we see that there do not exist non-trivial solutions $X_i$ which do not vanish at spatial infinity and that $T_{ij}$ satisfies a (normally) hyperbolic equation. The scalar parts of $P\Gamma=0$ can be subsumed in a particularly simple form in terms of the so-called (generalised) Mukhanov-Sasaki variable $\mu$, which is proportional to ${\cal R}$, a quantity related to the perturbation of the spatial curvature and a key geometrical quantity in the physics of inflationary perturbations, see e.g. \cite{Makino:1991sg}.
\beq\label{eq_defmu}\mu:=\frac{\widebar {z}}{a}\left(\Psi-\frac{\xi\phi}{\alpha}\chi\right)+\sqrt{\frac{\beta}{\alpha}}\chi\qquad \widebar {z}:=
a\sqrt{\frac{\beta}{\alpha}}\frac{\phi^\prime}{\H+\sqrt{\alpha}^\prime}\qquad {\cal R} := -\sqrt{\alpha}a\frac{\mu}{\widebar {z}}\eeq
In terms of $\mu$, the scalar equations of motion  displayed in section \ref{sec_fullFLRW} can be re-expressed equivalently as
\beq\label{eq_EOMmu1}P^\mu\mu:=\left(-\nabla_c\nabla^c + \frac{R}{6}-\frac{\widebar {z}^{\prime\prime}}{\widebar {z}a^2}\right)\mu=0\,,\eeq
\beq\label{eq_EOMmu2}\Psi-\frac{\xi \phi}{\alpha}\chi=\frac{\H+\sqrt{\alpha}^\prime}{2a^2 \alpha}\left(\int\limits^\tau_{\tau_0}d\tau_1 a \widebar {z}\mu+\lambda_0\right)\,,\eeq
\beq\label{eq_EOMmu3}\Phi-\frac{\xi\phi}{\alpha}\chi=-\left(\Psi-\frac{\xi \phi}{\alpha}\chi\right)\qquad \chi=\frac{2 \alpha^2}{\beta \phi^\prime}(\partial_\tau+\H+\sqrt{\alpha}^\prime)\left(\Psi-\frac{\xi \phi}{\alpha}\chi\right)\,,\eeq
where $\tau_0$ is arbitrary and $\lambda_0$ is the unique solution of
\beq\label{eq_EOMmu4}\left.\Delta\lambda_0=a\sqrt{\alpha}\overline{z}\left(\mu^\prime+\left(\frac{\partial_\tau\left(\H+\sqrt{\alpha}^\prime\right)}{\H+\sqrt{\alpha}^\prime}-\frac{\partial_\tau\left(\frac{\sqrt{\beta}}{\alpha}\phi^\prime\right)}{\frac{\sqrt{\beta}}{\alpha}\phi^\prime}\right)\mu\right)\right\vert_{\tau=\tau_0}.\eeq
Thus one can view $\mu$ as the basic dynamical variable, which is a conformally coupled scalar field with time-dependent mass, whereas $\Phi$, $\Psi$ and $\chi$ can be inferred from $\mu$ on-shell. The ``correct'' definition of $\mu$ \eqref{eq_defmu} and the equations \eqref{eq_EOMmu1}--\eqref{eq_EOMmu3} can also be obtained from the simpler equations in the minimally coupled case $\xi=0$ by using the transformation outlined in Appendix \ref{sec_EinsteinJordan}. 

The previous discussions imply that the following definitions of section spaces are consistent (recall \eqref{eq_def_origquants}).
$$\Sol_\infty:=\Sol\cap \Gamma_\infty(\V)\qquad \G_\infty:=\G\cap \Gamma_\infty(\V)$$ 
\beq\label{eq_def_splitsoletc}\Sol^{S/V/T}_{(\infty/\sc)}:= \Sol\cap \Gamma_{(\infty/\sc)}^{S/V/T}(\V)\qquad \G^{S/V}_{(\infty/\sc)}:=\G\cap \Gamma_{(\infty/\sc)}^{S/V}(\V)\qquad \G_{(\infty/\sc)}^{T}:=\{0\}\eeq
$$\G^{S/V}_{\sc,0}:=K\left[\Gamma^{S/V}_\sc(\W)\right]$$
$$\Ker^{S/V/T}_0(K^\dagger):=\Ker_0(K^\dagger)\cap \Gamma_0^{S/V/T}(\V)\qquad \E^{S/V/T}:= \left.\Ker^{S/V/T}_0(K^\dagger)\right/P\left[\Gamma^{S/V/T}_0(\V)\right]$$
Some of these spaces have a particularly simple form, viz.
\beq\label{eq_simplespaces}\Ker^T_0(K^\dagger)=\Gamma^T_0(\V)\qquad \G^{S/V}_\sc=\G^{S/V}_{\sc,0}\qquad \Sol^V_\infty/\G^V_\infty=\Sol^V_\sc/\G^V_{\sc}=\{0\}\qquad \E^V=\{0\}\,.\eeq 
The first identity follows from \eqref{eq_KdaggerFLRW} while the second can be deduced from \eqref{eq_bdecompsupp},\eqref{eq_gaugetrafoFLRW} and the fact that the time-integral of a space-like compact function is space-like compact. The third identity follows from Lemma \ref{prop_confgauge} and \eqref{eq_EOMvec}. The last identity is, as we shall see, dual respectively equivalent to the third and follows from the fact that each vector solution $h\in\Gamma^V_0(\V)$ of $K^\dagger h=0$ is of the form $h=Pj$, $j\in\Gamma^V_0(\V)$, cf. \eqref{eq_KdaggerFLRW} and Section \ref{sec_fullFLRW}. Moreover, we can show the following.
\begin{theo}\label{prop_specialsplitrelations}The spaces \eqref{eq_def_splitsoletc} satisfy the following relations.
\begin{enumerate}
\item $F^{-1}\circ G^{\widetilde P}:\Gamma_0(\V)\to\Gamma(\V)$ induces bijective maps $F^{-1}\circ G^{\widetilde P}:\E^{S/V/T}\to \Sol^{S/V/T}_\sc/\G^{S/V/T}_\sc$.
\item $\E^{S/V/T}\subset\E$ and $\Sol^{S/V/T}_\sc/\G^{S/V/T}_\sc\subset \Sol_\sc/\G_{\sc,0}$.
\item $\Sol_{\infty}/\G_{\infty}$ can be split as
$$\Sol_{\infty}/\G_{\infty}=\Sol_{\infty}^{S}/\G_{\infty}^{S}\oplus\Sol_{\infty}^{V}/\G_{\infty}^{V}\oplus\Sol_{\infty}^{T}/\G_{\infty}^{T}\,.$$
\end{enumerate}
\end{theo}
\begin{proof}{Proof of 1}:
 This statement can be proven exactly as the fact that $F^{-1}\circ G^{\widetilde P}$ induces a bijective map $F^{-1}\circ G^{\widetilde P}:\E\to \Sol_\sc/\G_\sc$, whereby Proposition \ref{prop_splithyp} may be used to maintain the scalar/vector/tensor character of the objects at each step. As the necessary steps are rather lengthy, we transfer them to Theorem \ref{prop_surjinj} in the appendix.
 
{Proof of 2}: Assuming $\Sol^{S/V/T}_\sc/\G^{S/V/T}_\sc\subset \Sol_\sc/\G_{\sc,0}$, $\E^{S/V/T}\subset\E$ follows by $\E\simeq \Sol_\sc/\G_{\sc,0}$ and $\E^{S/V/T}\simeq \Sol^{S/V/T}_\sc/\G^{S/V/T}_\sc$, cf. Section \ref{sec_removingredefinition} and the first statement of this theorem. To prove $\Sol^{S/V/T}_\sc/\G^{S/V/T}_\sc\subset \Sol_\sc/\G_{\sc,0}$, we define $\iota:\Sol^{S/V/T}_\sc/\G^{S/V/T}_\sc\to \Sol_\sc/\G_{\sc,0}$ by $\iota([\Gamma]):=[\Gamma]$ and the wanted statement follows if we can show that $\iota$ is injective. To see this, we assume $\iota([\Gamma])=0$, thus there exists $\varsigma\in \Gamma_\sc(\W)$ such that $\Gamma=K\varsigma$. However, we know that $\Gamma\in\Sol^{S/V/T}_\sc$. In the scalar case, it follows that the vector components $V_i$ and $W_i$ of $\Gamma$ vanish, and thus the vector component $z_i$ of $\varsigma$ vanishes by \eqref{eq_gaugetrafoFLRW}. Hence, $\Gamma\in\G^S_\sc$ and $[h]=[0]\in\Sol^S_\sc/\G^S_\sc$. In the vector case, we have nothing to prove as $\Sol^V_\sc/\G^V_\sc=\{0\}$. Finally, in the tensor case, we note that $K[\Gamma_\sc(\W)]\cap \Sol^T_\sc=\{0\}$, thus $[\Gamma]=[0]\in\Sol_\sc/\G_{\sc,0}$ and $\Gamma\in\Sol^T_\sc$ implies $\Gamma=0$.

{Proof of 3}: The splitting of $\Sol_{\infty}/\G_{\infty}$ follows from \eqref{eq_splitintersect} and the splittings $\Sol_{\infty}=\Sol_{\infty}^{S}\oplus\Sol_{\infty}^{V}\oplus\Sol_{\infty}^{T}$ and $\G_{\infty}=\G_{\infty}^{S}\oplus\G_{\infty}^{V}\oplus\G_{\infty}^{T}$ which in turn follow from \eqref{eq_invP} and \eqref{eq_invK}, respectively.

\end{proof}

\subsection{Comparing the general quantization procedure with the standard approach}
\label{sec_inflationresults}

In the previous sections we have seen that, on FLRW backgrounds, on-shell configurations of the linearised Einstein-Klein-Gordon system which vanish at spatial infinity can be uniquely split into scalar, vector and tensor parts, where the vector parts are pure gauge and gauge-invariant linear combinations of the scalar and tensor parts satisfy hyperbolic equations of motion $P^\mu\mu=0$ \eqref{eq_EOMmu1} and $P^T T_{ij}=0$ \eqref{eq_EOMten}. Similarly, we have seen that, at least a subset of the observables on the linearised Einstein-Klein-Gordon system can be split analogously, whereby the resulting vector parts are trivial. It thus seems that the quantum theory of the linearised Einstein-Klein-Gordon system on FLRW backgrounds contains sub-theories corresponding to the scalar and tensor degrees of freedom, whereas all vector degrees of freedom are pure gauge. Consequently, the usual approach to quantizing perturbations in inflation is to take the equations $P^\mu\mu=0$ and $P^T T_{ij}=0$ (or the corresponding Lagrangians) as a starting point for canonical quantization. In view of our general approach to quantize the Einstein-Klein-Gordon system provided in Section \ref{sec_quantgen}, two questions arise:
\begin{enumerate}
\item Are the scalar and tensor observables in the full quantum theory of the linearised Einstein-Klein-Gordon system equivalent to the ones in the quantum theories of $\mu$ and $T_{ij}$ constructed based on $P^\mu\mu=0$ and $P^T T_{ij}=0$?
\item Can all local observables in the full quantum theory of the linearised Einstein-Klein-Gordon system be split into local scalar and local tensor observables, and thus into $\mu$- and $T_{ij}$-observables?
\end{enumerate}
The aim of this section is to answer these questions in a rigorous fashion, whereby we will see that the first question has a positive answer, whereas the answer to the second one is negative. While the latter implies that in principle the usual approach to directly quantize the scalar and tensor degrees of freedom does not give all observables of the quantum linearised Einstein-Klein-Gordon system, a subsequent analysis in Section \ref{sec_separability} will indicate that the local scalar and tensor observables are sufficient  if one considers only configurations of the linearised Einstein-Klein-Gordon system which vanish at spatial infinity.

As we have reviewed the canonical quantization of (pre)symplectic spaces in Section \ref{sec_quant}, we may answer the above questions on the level of (pre)symplectic spaces, i.e. on the level of classical field theories.

\subsubsection{The scalar sector}
\label{sec_scalsec}

If $\Gamma\in\Gamma_\infty(\V)$ then the splitting \eqref{eq_def_split} is unique and we can view $A,B,D,E,V_i,W_i,T_{ij}$, and, consequently, $\Psi,\Phi,\chi,\mu, X_i$ as functionals of $\Gamma$ and thus as maps $A:\Gamma_\infty(\V)\to C_\infty^\infty(M,\bbR)$, and so on. In order to investigate the scalar sector, we analyse the properties of $\mu:\Gamma_\infty(\V)\to C_\infty^\infty(M,\bbR)$. To this avail, we recall that $P^\mu$ is the normally hyperbolic differential operator \eqref{eq_EOMmu1} and define the following $\mu$-related quantities.
$$\Sol^\mu=\{\mu\in C^\infty(M,\bbR)\,|\, P^\mu \mu=0\}\,,\qquad \Sol^\mu_{\infty/\sc}:=\Sol^\mu\cap C_{\infty/\sc}^\infty(M,\bbR)$$
$$\langle \cdot,\cdot\rangle_\mu: C^\infty(M,\bbR)\times C^\infty(M,\bbR)\to\bbR,\quad (f_1,f_2)\mapsto \langle f_1,f_2\rangle_\mu:=\int\limits_M \vol\; f_1 f_2$$
$$\langle \cdot,\cdot\rangle_{\Sol^\mu}: \Sol^\mu\times\Sol^\mu\to\bbR,\quad (f_1,f_2)\mapsto \langle f_1,f_2\rangle_{\Sol^\mu}:=\langle P^\mu f^+_1,f_2\rangle_\mu$$
$$G_\pm^\mu:C^\infty_0(M,\bbR)\to C^\infty(M,\bbR),\;\; P^\mu G_\pm^\mu=\id_{C^\infty_0(M,\bbR)},\;\; \supp\, G_\pm^\mu f\subset J^\pm(\supp \,f)\;\forall \,f\in C^\infty_0(M,\bbR)$$ 
$$G^\mu:C^\infty_0(M,\bbR)\to C^\infty(M,\bbR)\,,\qquad G^\mu:=G_+^\mu-G_-^\mu$$
$$\E^\mu:=C^\infty_0(M,\bbR)/P^\mu\left[C^\infty_0(M,\bbR)\right]$$
$$\sigma^\mu:\E^\mu\times\E^\mu\to \bbR,\qquad([f_1],[f_2])\mapsto\sigma^\mu([f_1],[f_2]):=\langle f_1,G^\mu f_2\rangle_\mu$$
$$j^\mu:C^\infty(M,\bbR)\times C^\infty(M,\bbR)\to  T^*M\,,\qquad(f_1, f_2) \mapsto j^\mu_a(f_1,f_2):=-f_1\nabla_a f_2+f_2\nabla_a f_1$$
where $\langle\cdot,\cdot\rangle_\mu$ and $\langle \cdot,\cdot\rangle_{\Sol^\mu}$ are understood to be defined only on functions with compact, respectively, space-like compact overlapping support, and the future part $f^+$ of a smooth function $f$ is defined as in Section \ref{sec_solform}.

As already anticipated the equation of motion $P^\mu \mu=0$, which can be seen as coming from the action $S(\mu):=\frac12 \langle \mu, P^\mu \mu\rangle_\mu$, defines a field theory on FLRW backgrounds $\M=(M,g,\phi)$ on its own. In analogy to our discussion of the linearised Einstein-Klein-Gordon system, we can consider this field theory a specified by the data $(\M,M\times\bbR,P^\mu)=(\M,\V=M\times\bbR,\W=M\times\bbR,P=P^\mu,K=0)$. The following properties of this field theory can then be shown by the methods used in the analysis of the linearised Einstein-Klein-Gordon system, where the non-degeneracy of $\sigma^\mu$ can be proven as in \cite{Bar2}.
\begin{theo}\label{prop_mutheoryprop}The conformally coupled scalar field theory on FLRW backgrounds defined by $(\M,M\times\bbR,P^\mu)$ satisfies the following relations.
\begin{enumerate}
\item $P^\mu$ is formally selfadjoint w.r.t. $\langle\cdot,\cdot\rangle_\mu$.
\item $G^\mu$ descends to a bijective map $G^\mu:\E^\mu\to \Sol_\sc^\mu$.
\item $\sigma^\mu$ and $\langle \cdot,\cdot\rangle_{\Sol^\mu}$ are antisymmetric and well-defined. Moreover, $\sigma^\mu$ is non-degenerate.
\item For all $\mu\in\Sol^\mu$ and all $f\in C^\infty_0(M,\bbR)$, $\langle \mu, G^\mu f\rangle_{\Sol^\mu}=\langle \mu, f\rangle_\mu$. In particular, for all $f_1,f_2\in C^\infty_0(M,\bbR)$, $\langle G^\mu f_1,  G^\mu f_2\rangle_{\Sol^\mu}=\sigma^\mu([f_1],[f_2])$.
\item $G^\mu$ descends to an isomorphism of symplectic spaces $G^\mu:(\E^\mu,\sigma^\mu)\to (\Sol^\mu_\sc,\langle\cdot,\cdot\rangle_{\Sol^\mu})$.
\item For every $\mu_1$, $\mu_2\in \Sol^\mu$ with space-like compact overlapping support, and every Cauchy surface $\Sigma$ of $(M,g)$ with forward pointing unit normal vector field $n$
$$\langle \mu_1, \mu_2\rangle_{\Sol^\mu}=\int\limits_\Sigma \vols\, n^a {j}^\mu_a(\mu_1,\mu_2)\,.$$
\end{enumerate}
\end{theo}

By Theorem \ref{prop_specialsplitrelations}, the space $\E=\Ker_0(K^\dagger)/P[\Gamma_0(\V)]\ni h$ which labels all gauge-invariant linear observables of the classical linearised Einstein-Klein-Gordon system via $\Sol/\G\ni [\Gamma]\mapsto \langle \Gamma, h\rangle_\V$, contains the scalar subspace $\Ker^S_0(K^\dagger)/P[\Gamma^S_0(\V)]=\E^S\subset \E$. This space may be legitimately interpreted as indexing local scalar observables because Lemma \ref{prop_nondegtensor} implies that $\langle \E^S, \Sol^{V/T}\rangle_\V=\{0\}$\footnote{We presume that $\E^S$ is the maximal subspace of $\E$ with this property, but we have not been able to prove this as we are lacking a convenient way of parametrising $\Sol^V$. $\Sol^T$ can be parametrised by all time-like compact elements of $\Gamma^T(\V)$ because $P^T$ in \eqref{eq_EOMten} is normally hyperbolic, cf. \cite{Baernew}. The vector equations \eqref{eq_EOMvec} can also be combined into a normally hyperbolic equation, but $\Sol^V$ does not correspond to all solutions of this equation.}. The presymplectic form $\sigma$ on $\E$ can thus be restricted to $\E^S$, s.t. $(\E^S,\sigma)$ can be considered as a presymplectic subspace of $(\E,\sigma)$. These two presymplectic spaces, as well as $(\E^\mu,\sigma^\mu)$, can be quantized as described in Section \ref{sec_quant}, leading to quantum observable algebras $\A^S$, $\A$ and $\A^\mu$ respectively, whereby $\A^S$ is a subalgebra of $\A$. In order to demonstrate equivalence of the quantum theory of $\mu$, constituted by $\A^\mu$, to the scalar sector $\A^S$ of the full quantum theory of the linearised Einstein-Klein-Gordon system given by $\A$, we need to show both that $\A^S$ and $\A^\mu$ are isomorphic and that this isomorphism preserves the physical meaning of the quantum observables. This is achieved once we find bijective maps $\mu:\Sol_\infty/\G_\infty\to \Sol^\mu_\infty$ and $\mu_0:\E^S\to\E^\mu$ and show both that $\mu_0$ induces an isomorphism $\mu_0: (\E^S,\sigma)\to (\E^\mu,\sigma^\mu)$ and that $\langle \Gamma, h\rangle_\V=\langle \mu(\Gamma),\mu_0(h)\rangle_\mu$ for all $([\Gamma],[h])\in\Sol_\infty/\G_\infty\times \E$.

\begin{theo}
\label{prop_mumapprop}
The map $\mu:\Gamma_\infty(\V)\to C_\infty^\infty(M,\bbR)$ enjoys the following properties.
\begin{enumerate}
\item $\mu$ induces bijective maps $$\mu:\Sol^S_\infty/\G^S_\infty\to\Sol^\mu_\infty\quad\text{and}\quad\mu:\Sol^S_\sc/\G^S_{\sc}\to\Sol^\mu_\sc\,.$$
\item $\mu$, $F^{-1}\circ G^{\widetilde P}$ and $G^\mu$ induce a bijective map $$\mu_0:\E^S\to \E^\mu\,,\qquad\mu_0:= \left(G^\mu\right)^{-1}\circ \mu \circ  F^{-1}\circ G^{\widetilde P}\,.$$
\item For all $\Gamma\in\Sol^S_\infty$, $H\in\Sol^S_\sc$
$$\langle \Gamma,H\rangle_\Sol=\langle \mu(\Gamma),\mu(H)\rangle_{\Sol^\mu}\,.$$
\item For all $\Gamma\in\Sol^S_\infty$ and all $h\in\Ker^S_0(K^\dagger)$
$$\langle \Gamma,h\rangle_\V=\langle \mu(\Gamma),\mu_0(h)\rangle_\mu\,.$$
\item The presymplectic spaces $(\E^S,\sigma)$, $(\Sol^S_\sc/\G^S_\sc,\langle\cdot,\cdot\rangle_\Sol)$, $(\E^\mu,\sigma^\mu)$ and $(\Sol^\mu_\sc,\langle\cdot,\cdot\rangle_{\Sol^\mu})$ are all isomorphic and symplectic. In particular, $\sigma$ is non-degenerate on $\E^S$.
\end{enumerate}
\end{theo}
\begin{proof} We first note that $\mu$ is well-defined on the quotients because $\mu$ is gauge-invariant, i.e. $\mu\circ K|_{\Gamma^S_\infty(\W)}=0$, cf. \eqref{eq_defmu}, \eqref{eq_def_Bardeen} and \eqref{eq_gaugetrafoFLRW}. Surjectivity of the $\mu$ map follows from \eqref{eq_EOMmu2} as every smooth function on $M$ can be written as a conformal time derivative of a smooth function on $M$. To see injectivity, observe that, by \eqref{eq_EOMmu2} and \eqref{eq_EOMmu3}, $\mu(\Gamma)=0$ implies $\Psi(\Gamma)=\Phi(\Gamma)=\chi(\Gamma)=0$. By Lemma \ref{prop_confgauge} and the gauge invariance of $\mu$, $\Psi$, $\Phi$ and $\chi$ we may consider $\Gamma$ to be in conformal gauge, thus $\mu(\Gamma)=0\Rightarrow A=D=\varphi=0\Rightarrow [\Gamma]=[0]$.

The second statement follows from the fact that, by the first statement, Theorem \ref{prop_mutheoryprop} and Theorem \ref{prop_specialsplitrelations}, $\mu$ and $G^\mu$ and $F^{-1}\circ G^{\widetilde P}$ are bijective.  Thus $\mu_0$ is bijective as a composition of bijective maps.

The third statement can be seen by realising that for all $\Gamma\in\Sol^S_\infty$, all $H\in\Sol^S_\sc$ and every Cauchy surface $\Sigma$ of $(M,g)$ with forward pointing unit normal vector field $n$
$$\int\limits_\Sigma \vols\, n^a {j}_a(\Gamma,H)=\int\limits_\Sigma \vols\, n^a {j}^\mu_a(\mu(\Gamma),\mu(H))\,,$$ 
Where $j_a$ is defined in \eqref{eq_origcurrent}. This identity in turn is best checked on a Cauchy surface orthogonal to $\partial_\tau$ and by computing in conformal gauge, cf. Lemma \ref{prop_confgauge}. This computation is quite cumbersome, but straightforward. Using this identity, the statement follows from Lemma \ref{prop_solcurrent} and Theorem \ref{prop_mutheoryprop}.

The last two statements follow directly from the first three,  Theorem \ref{prop_mutheoryprop} and Section \ref{sec_removingredefinition}. In particular, it holds for all $\Gamma\in\Sol^S_\infty$ and all $h\in\Ker^S_0(K^\dagger)$
$$\langle \Gamma, h\rangle_V=\left\langle \Gamma, F^{-1}G^{\widetilde P}h\right\rangle_\Sol=\left\langle \mu(\Gamma),\mu\left(F^{-1}G^{\widetilde P}h\right)\right\rangle_{\Sol^\mu}=\langle \mu(\Gamma),\mu_0(h)\rangle_\mu\,.$$
\end{proof}

We have seen that the scalar sector of the linearised Einstein-Klein-Gordon system on FLRW backgrounds is equivalent to the field theory of the scalar field $\mu$. One might wonder whether there is a different linear combination $\widetilde\mu$ of the gauge invariant potentials $\Psi$, $\Phi$ and $\chi$ with the same properties. In \cite{Eltzner:2013soa} it has been proven that $\mu$ is indeed the only linear combination of these potentials which satisfies a normally hyperbolic differential equation and canonical commutation relations on all FLRW backgrounds. While \cite{Eltzner:2013soa} considers only the minimally coupled case $\xi=0$, the result can be directly extended to the non-minimally coupled case by means of the transformation outlined in Section \ref{sec_EinsteinJordan}.

\subsubsection{The tensor sector}

The tensor sector of the linearised Einstein-Klein-Gordon system can be discussed largely in analogy to the scalar sector, though the analysis is simplified considerably by the fact that the map $T_{ij}:\Gamma_\infty(\V)\to C^\infty(M,\bigvee^2\bbR^3)$ and the normally hyperbolic operator $P^T$ in \eqref{eq_EOMten} are defined in a rather direct fashion. Thus, one can straightforwardly prove a tensor version of the combined statements of Theorem \ref{prop_mutheoryprop} and Theorem \ref{prop_mumapprop}. For simplicity, we omit the indices in the notation of the $T_{ij}$ map in the following and define $$C_{(\infty/\sc/0)}^\infty(M,T):=\{T\in C_{(\infty/\sc/0)}^\infty(M,{\textstyle\bigvee^2}\bbR^3)\,|\,T^i_i=0\,,\partial^iT_{ij}=0\}\,.$$

\begin{theo}\label{prop_tensorprop}The map $T:\Gamma_\infty(\V)\to C^\infty(M,T)$ enjoys the following properties.
\begin{enumerate}
\item $\langle \cdot\,,\cdot\rangle_\V|_{\Gamma^T_\infty(\V)\times\Gamma^T_\infty(\V)}=\langle T(\,\cdot\,)\,,T(\,\cdot\,)\rangle_T$, where the bilinear form $\langle \cdot,\cdot\rangle_T: C^\infty(M,T)\times C^\infty(M,T)\to\bbR$, defined for functions of compact overlapping support as
$$(t_1,t_2)\mapsto \langle t_1,t_2\rangle_T:=\int\limits_M \vol\; t_1^{ij} {t_2}_{ij}\,,$$
is symmetric and non-degenerate.
\item $T\circ P|_{\Gamma_\infty(\V)}=P^T\circ T$ and $P^T$ is formally selfadjoint w.r.t. $\langle \cdot,\cdot\rangle_T$.
\item $T\left[\Sol_{\infty/\sc}^T\right]=\left\{T\in C_{\infty/\sc}^\infty(M,T)\,|\,P^T T=0\right\}$
\item $T\left[\E^T\right]=C_0^\infty(M,T)/P^T[C_0^\infty(M,T)]$
\item $T\circ F^{-1}\circ G^{\widetilde P}|_{\Gamma^T_0(\V)}=G^T\circ T|_{\Gamma^T_0(\V)}$, where 
$$G^T:C^\infty_0(M,T)\to C^\infty(M,T)\,,\qquad G^T:=G_+^T-G_-^T\,,$$
$$G_\pm^T:C^\infty_0(M,T)\to C^\infty(M,T),\quad P^T G_\pm^T=\id_{C^\infty_0(M,T)}\,,$$
$$\supp\, G_\pm^\mu t\subset J^\pm(\supp \,t)\;\;\forall \,t\in C^\infty_0(M,T)\,.$$ 
\item $\sigma|_{\E^T\times\E^T}=\sigma^T\circ(T\times T)$, where
$$\sigma^T:T[\E^T]\times T[\E^T]\to \bbR,\quad([t_1],[t_2])\mapsto\sigma^T([t_1],[t_2]):=\langle t_1,G^T t_2\rangle_T\,.$$
In particular, $\sigma$ is non-degenerate on $\E^T$.
\item $j|_{\Gamma^T_\infty(\V)\times\Gamma_\infty^T(\V)}=j^T\circ(T\times T)$, where
$$j^T:C^\infty(M,T)\times C^\infty(M,T)\to  T^*M\,,\quad(t_1, t_2) \mapsto j^T_a(t_1,t_2):=-t^{ij}_1\nabla_a {t^{\phantom{j}}_2}_{ij}+t^{ij}_2\nabla_a {t^{\phantom{j}}_1}_{ij}\,.$$
\item For every $\Gamma_1$, $\Gamma_2\in \Sol_\infty^T$ with space-like compact overlapping support, and every Cauchy surface $\Sigma$ of $(M,g)$ with forward pointing unit normal vector field $n$
$$\langle \Gamma_1,\Gamma_2\rangle_\Sol=\langle T(\Gamma_1),T(\Gamma_2)\rangle_{\Sol^T}=\int\limits_\Sigma \vols \;n^a j^T_a\left(T(\Gamma_1),T(\Gamma_2)\right)\,,$$
where 
$$\langle \cdot,\cdot\rangle_{\Sol^T}: T[\Sol_\infty^T]\times T[\Sol_\infty^T]\to\bbR,\quad (T_1,T_2)\mapsto \langle T_1,T_2\rangle_{\Sol^T}:=\langle P^T T^+_1,T^{\phantom{+}}_2\!\rangle_T\,.$$
\end{enumerate}
\end{theo}
\begin{proof} 1. follows from the form of $\langle \cdot\,,\cdot\rangle_\V$ on FLRW backgrounds, see e.g. the proof of Lemma \ref{prop_nondegtensor}, whereas 2. is manifest from the definition of $P^T$ in \eqref{eq_EOMten} and the formal selfadjointness of $P$. 3., 4. and 5. follow from 2., where we note that $G^T$ is well defined because $P^T$ is diagonal and commutes with spatial derivatives. 6. follows from 5. whereas 7. can be shown by a direct computation. Finally, 8. follows from 2., 7. and Lemma \ref{prop_solcurrent}.
\end{proof}

As in the scalar case, the above results imply that the subalgebra $\A^T$ of $\A$, corresponding to the local tensor observables of the linearised Einstein-Klein-Gordon system on FLRW backgrounds, can be understood either as the result of quantizing the symplectic space $(\E^T,\sigma)\subset (\E,\sigma)$ or as the result of quantizing the symplectic space $(T[\E^T],\sigma^T)$, which is the canonical symplectic space associated to the normally hyperbolic equation $P^T T=0$. Note that, owing to the orthogonality result Lemma \ref{prop_nondegtensor}, $\sigma(\E^S,\E^T)=\{0\}$ and thus the subalgebras $\A^S$ and $\A^T$ of $\A$ commute.

\subsubsection{Are all local observables of scalar and tensor type?}

As anticipated, we now prove that $\E\neq \E_0=\E^S\oplus \E^T$, which implies that $\A$ is not generated solely by $\A^S$ and $\A^T$.

\begin{theo}$\E^0\subsetneq\E$ and $\Sol^S/\G^S_\sc\oplus \Sol^T/\G^T_\sc\subsetneq \Sol_\sc/\G_{\sc,0}$, i.e. not all local observables of the linearised Einstein-Klein-Gordon field theory can be split into local scalar and local tensor observables.
\end{theo}
\begin{proof}The proof strategy is as follows. We give an example for $\Gamma\in\Sol^V$ and $h\in\Ker^0(K^\dagger)$ such that $\langle \Gamma, h\rangle_\V\neq 0$. Since by orthogonality, cf. Lemma \ref{prop_nondegtensor}, $\langle \Sol^V, h\rangle_\V=0$ for all $h\in[h]\in\E^0$, this implies that $[h]\in\E\setminus \E^0$. That said, we define $\Gamma\in\Gamma^V(\V)$ by 
$$W_i=0\,,\qquad V_i=\frac{1}{a^2 \alpha}\begin{pmatrix}
\exp (x_1) \sin (x_2) x_3\\\exp (x_1) \cos (x_2) x_3\\0
\end{pmatrix}\,,$$
and we define $h=(k_{ab},f)^T\in\Gamma_0(\V)$ by
$f=0$ and $$k_{ab}=\begin{pmatrix}0& \partial_3\partial_2 \lambda& 0 &-\partial_1\partial_2 \lambda\\
\partial_3\partial_2 \lambda & 0 & (\partial_\tau+2\H)\partial_3 \lambda & 0\\
0 & (\partial_\tau+2\H)\partial_3 \lambda & 0 & -(\partial_\tau+2\H)\partial_1 \lambda\\
-\partial_1\partial_2 \lambda & 0 & -(\partial_\tau+2\H)\partial_1 \lambda & 0
\end{pmatrix}\,,$$
where $\lambda\in C_0^\infty(M,\bbR)$ is arbitrary.
It is not difficult to check that $\Gamma\in\Sol^V$ and that $K^\dagger h=0$. Moreover, we can compute
$$\langle \Gamma, h\rangle_\V=-2\int \limits_M \vol\; \frac{\exp (x_1) \sin (x_2) x_3 \partial_3\partial_2 \lambda}{a^2 \alpha}=-2\int \limits_M \vol\; \frac{\exp (x_1) \cos (x_2)   \lambda}{a^2 \alpha}\,,$$
which is clearly non-vanishing for many $\lambda\in C_0^\infty(M,\bbR)$.
$\Sol^S/\G^S_\sc\oplus \Sol^T/\G^T_\sc\subsetneq \Sol_\sc/\G_{\sc,0}$ follows from $\E^0\subsetneq\E$ by Theorem \ref{prop_specialsplitrelations} and Section \ref{sec_removingredefinition}.
\end{proof}

This result implies in particular that $\E_0$ is not separating on $\Sol/\G$, i.e. there exist $[\Gamma]\in \Sol/\G$ such that $\langle [\Gamma],\E^0\rangle_\V=\{0\}$. Notwithstanding, we shall demonstrate in the next section that $\E^0$ is indeed separating on $\Sol_\infty/\G_\infty$.

\subsection{Separability of solutions and non-degeneracy of the presymplectic form}
\label{sec_separability}

We analyse two important structural properties of the classical and quantum theory of perturbations in inflation. 

\begin{theo}\label{prop_sepnondeg}
The linearised Einstein-Klein-Gordon system $(\M,\V,\W,P,K)$ on backgrounds $\M=(M,g,\phi)$ of FLRW type possesses the following properties.
\begin{enumerate}
\item $\E^0:=\E^S\oplus\E^V\oplus\E^T=\E^S\oplus\E^T$ is separating on $\Sol_\infty/\G_\infty$, i.e. $$\Gamma\in\Sol_\infty\quad\text{and}\quad\langle \Gamma, h\rangle_\V=0\quad\forall\; [h]\in\E^0\qquad\Rightarrow\qquad [\Gamma]=[0]\in\Sol_\infty/\G_\infty\,.$$
\item $\sigma$ is non-degenerate on $\E$, i.e. 
$$\sigma([h_1],[h_2])=0\quad\forall\, [h_1]\in \E\qquad\Rightarrow\qquad [h_2]=[0]\,.$$
\end{enumerate}
\end{theo}

\begin{proof}
{Proof of 1}: Since by Lemma \ref{prop_specialsplitrelations} $\Sol_\infty/\G_\infty=\Sol^S_\infty/\G^S_\infty\oplus \Sol^V_\infty/\G^V_\infty\oplus \Sol^T_\infty/\G^T_\infty$ and $\Sol^V_\infty/\G^V_\infty=\{0\}$, cf. \eqref{eq_simplespaces}, we can consider only solutions of the form $\Gamma = \Gamma_1+\Gamma_2$ with $\Gamma_1\in\Sol^S_\infty$ and $\Gamma_2\in\Sol^T_\infty$. Using the fourth statement of Theorem \ref{prop_mumapprop} and setting $h=h_1+h_2$ with $[h_1]\in\E^S$ and $[h_2]\in\E^T$, we can compute
$$\langle \Gamma, h\rangle_\V=\int\limits_M\vol\left( \mu(\Gamma_1)\mu_0(h_1)+T^{ij}t_{ij}\right).$$ Using now the injectivity of $\mu:\Sol^S_\infty/\G^S_\infty\to \Sol^\mu_\infty$ and the surjectivity of $\mu_0:\E^S\to\E^\mu$ (see Theorem \ref{prop_mumapprop}), $\Ker^T_0(K^\dagger)=\Gamma^T_0(\V)$, as well as Lemma \ref{prop_nondegtensor}, we find that $\Gamma\in\Sol_\infty$ and $\langle \Gamma, h\rangle_\V=0$ for all $[h]\in\E^0$ imply $\mu(\Gamma_1)=T_{ij}=0$ and thus $[\Gamma]=[0]$.\\\\
{Proof of 2}: We recall that $F^{-1}\circ G^{\widetilde{P}}$ defines a bijective map between $\E$ and $\Sol_\sc/\G_{\sc,0}$ and that $\sigma([h_1],[h_2])=\langle h_1,F^{-1}G^{\widetilde{P}}h_2\rangle_V$, cf. Section \ref{sec_removingredefinition}. Thus the statement is equivalent to 
$$\Gamma\in\Sol_\sc\quad\text{and}\quad\langle \Gamma, h\rangle_\V=0\quad\forall\; [h]\in\E
\qquad\Rightarrow\qquad [\Gamma]=[0]\in\Sol_\sc/\G_{\sc,0}\,,$$ i.e. to the separability of $\Sol_\sc/\G_{\sc,0}$ by $\E$. Since $\E^0\subset\E$, this follows already if we can prove that $\E^0$ separates $\Sol_\sc/\G_{\sc,0}$. To show this, we assume $\Gamma\in\Sol_\sc$ and $\langle \Gamma, h\rangle_\V=0$ for all $[h]\in\E^0$. Choosing $[h]\in\E^T$, we can argue as in the proof of the first part that $T_{ij}=0$, thus $\Gamma\in(\Sol^S_\infty\oplus \Sol^V_\infty)\cap\Sol_\sc$. Choosing instead $[h]\in\E^S$, we can again demonstrate as in the proof of part one that $\mu(\PP^S_\V \Gamma)=0$. We can not invoke the conformal gauge now, because we don't know whether  $\PP^S_\V\Gamma$ is space-like compact in the first place, whereas all gauge transformations we perform in this proof have to be elements of $\G_{\sc,0}$. Notwithstanding, we can compute in synchronous gauge $A=B=V_i$ as this is possible within $\G_{\sc,0}$ by Lemma \ref{prop_confgauge}. In this gauge, $\mu(\PP^S_\V \Gamma)=0$ implies by \eqref{eq_EOMmu2}, \eqref{eq_EOMmu3} and \eqref{eq_def_Bardeen} that  \beq\label{eq_temp}D=\H E^\prime\qquad \varphi=\phi^\prime E^\prime\qquad (\partial_\tau+\H)E^\prime=0\,.\eeq
\indent We now prove $\Gamma\in\G_{sc,0}$ under the assumption that either $\H\neq 0$ or $\phi^\prime\neq 0$ and consider the case $\H=\phi^\prime=0$ afterwards. That said, if $\H+\phi^\prime\neq 0$, we can infer from \eqref{eq_temp} via \eqref{eq_bdecompsupp2} that $E^\prime$ is space-like compact, and thus we can find a gauge transformation in $\G^S_\sc\subset \G_{\sc,0}$ (choosing $r=s^\prime=-E^\prime$ in \eqref{eq_def_splitgauge}) such that $D=\varphi=E^\prime = 0$ while maintaining $A=B=0$. Altogether, we find that, modulo $\G_{\sc,0}$, the only non-vanishing components of $\Gamma$ are $E$ and $W_i$, s.t. $\gamma_{ij}=2a^2(\partial_{i}\partial_j E+\partial_{(i} W_{j)})$. From $V_i=0$, the vector equation of motion \eqref{eq_EOMvec} and the vanishing of $W_i$ at spatial infinity we can infer that $W_i^\prime=0$. Moreover, since $\gamma_{ij}$ is space-like compact, we can deduce by setting $i=j$ that even $\partial_i E+W_i$ is space-like compact. Thus, by setting $r=0$, $s=E$ and $z_i=W_i$ in \eqref{eq_def_splitgauge}, we see that indeed $\Gamma\in\G_{sc,0}$.

\indent Finally, if $\H=\phi^\prime=0$, the synchronous gauge condition and \eqref{eq_temp} imply without further ado that the only non-vanishing components of $\Gamma$ are $E$ and $W_i$. As above, we can infer from this that   $\partial_i E+W_i$ is space-like compact and that $W_i^\prime=0$, whence $E^\prime$ is also space-like compact. Thus, as in the case $\H+\phi^\prime\neq 0$ we can achieve $E^\prime=0$ by a gauge transformation in $\G_{sc,0}$ which maintains $A=B=D=\varphi=V_i=0$ and leaves $W_i$ unchanged. Now $\Gamma\in\G_{sc,0}$ follows as above by setting $r=0$, $s=E$ and $z_i=W_i$ in \eqref{eq_def_splitgauge}. 
\end{proof}
Recall that the first statement of this theorem implies in physical terms that the set of classical observables labeled by $\E$ is large enough for distinguishing all classical pure states which vanish at spatial infinity. However, both statements of the above theorem are also relevant for the quantum theory of the linearised Einstein-Klein-Gordon system on FLRW backgrounds. Consider the polynomial quantum observable algebra $\A$ of the linearised Einstein-Klein-Gordon system on FLRW backgrounds, constructed out of $(\E,\sigma)$ as in Section \ref{sec_quant}, and recall that $\A$ is generated by smeared quantum fields $\Gamma(h)$. The second result of the above theorem implies that the the algebra $\A$ has a trivial centre, i.e. all elements which commute with the whole algebra are proportional to $\1$. Given any pure and Gaussian state $\langle\;\;\rangle_\omega$ on $\A$ and any classical solution $\Theta\in\Sol$, we may construct a coherent state $\langle\;\;\rangle_{\omega,\Theta}$ by defining an isomorphism $\iota:\A\to\A$ of $\A$ via $\A\ni\Gamma(h)\mapsto\iota(\Gamma(h)):=\Gamma(h)+\langle \Theta,h\rangle_\V\1$ and setting $\langle\;\;\rangle_{\omega,\Theta}:=\langle\;\;\rangle_{\omega}\circ\iota$. This coherent state has the one-point function $\langle\Gamma(h)\rangle_{\omega,\Theta}=\langle \Theta,h\rangle_\V$. Moreover, by extending the above splitting constructions to distributional sections, which is in principle straightforward as Proposition \ref{prop_splitbasic} is valid also for distributions, it is presumably possible to extend the above separability result in an inductive manner in order to prove that the subalgebra of $\A$ generated by the subalgebras $\A^S$ and $\A^T$ contains enough observables in order to distinguish quantum states whose correlation functions vanish at spatial infinity.

Unfortunately, we have not yet been able to prove the separability of the full solution space $\Sol/\G$ by $\E$, but the separability of $\Sol_\infty/\G_\infty$ by $\E^0$ is already sufficient in order to legitimate considering the smaller space $\E^0$ of observable labels rather than the full $\E$. The corresponding algebra of quantum observables is thus fully generated by local scalar and local tensor observables and corresponds to the standard quantization of perturbations in Inflation.

\subsection{Non-commutativity of the Bardeen potentials at space-like separations}
\label{sec_noncomm}

To close, we would like to comment on an interesting observation of \cite{Eltzner:2013soa}. The quantization of the scalar perturbations of inflation outlined in the previous sections implies that the quantum Mukhanov-Sasaki variable $\mu$ commutes at space-like separations. Indeed, by Section \ref{sec_inflationresults}, we have
$$[\mu(f_1),\mu(f_2)]=i \sigma^\mu([f_1],[f_2])\1=i\langle f_1,G^\mu f_2\rangle_\mu\1$$
where the smeared quantum field $\mu(f)$ may be interpreted as the quantization of the classical observable $\Sol^\mu\ni \mu \mapsto \langle \mu,f\rangle_\mu=\int_M\vol\, \mu f$. The above commutation relations may be re-written in ``unsmeared'' form as
$$[\mu(x),\mu(y)]=iG^\mu(x,y)\1$$
where $G^\mu(x,y)$ is the integral kernel of $G^\mu$. The construction of $G^\mu$ then implies that $G^\mu(x,y)$ vanishes for space-like separated $x$ and $y$. A further equivalent version of the commutation relations may be obtained from the fifth and sixth statement of Theorem \ref{prop_mutheoryprop}, which imply the equal-time relations\footnote{The $a^{-2}$ factor in the commutation relations may be removed by a conformal transformation $u:=a\mu$, whereby $u$ can be interpreted as a quantum field on Minkowski spacetime.}
\beq\label{eq_eqtccrmu}[\mu(\tau,\vec{x}),\mu(\tau,\vec{y})]=[\mu^\prime(\tau,\vec{x}),\mu^\prime(\tau,\vec{y})]=0\,,\qquad[\mu(\tau,\vec{x}),\mu^\prime(\tau,\vec{y})]=ia(\tau)^{-2}\delta(\vec{x},\vec{y})\1\,.\eeq

In \cite{Eltzner:2013soa} it has been found that these commutation relations for $\mu$ imply that the quantized gauge-invariant potentials $\Phi$, $\Psi$ and $\chi$ do {\it not} commute at space-like separations. This has been interpreted as a potential sign of ``non-commutativity of inflation'' as these potentials are gauge-invariant and thus have a clear physical meaning. In the following, we would like to argue that this non-commutativity is not surprising as a similar phenomenon occurs in any local quantum field theory. For simplicity we consider only the minimally coupled case $\xi=0$ (like \cite{Eltzner:2013soa}) and the Bardeen potential $\Psi$.

For convenience of the reader, we recall the relevant equations of motion.
\beq\label{eq_noncommeqns}
\left.\begin{array}{r}
\mu=\frac{1}{f_1 f_2}\Psi^\prime -\frac{f^\prime_1}{f^2_1 f_2}\Psi\quad\text{(a)}\\
\mu^\prime=f_3 \mu + \frac{1}{f_1 f_2}\Delta\Psi\quad\text{(b)}\\
\left(\partial^2_\tau-\Delta-\left(\frac{f_2}{a}\right)^{\prime\prime}\frac{a}{f_2}\right)a\mu=0\quad\text{(c)}
\end{array}\right\}\Leftrightarrow\left\{\begin{array}{r}
\Psi=f_1\left(\lambda_0+\int\limits^\tau_{\tau_0}d\tau_1 f_2 \mu \right)\quad\text{(d)}\\
\Delta\lambda_0 = \left.f_2\left(\mu^\prime-f_3 \mu\right)\right\vert_{\tau=\tau_0}\quad \text{(e)}\\
\quad\qquad\qquad\text{"}\qquad\quad\qquad\quad\text{(c)}
\end{array}\right.
\eeq
Here, $\tau_0$ is arbitrary, but fixed and the three purely time-dependent functions $f_i$ are defined as
$$f_1:=\frac{\H}{2a^2}\,,\qquad f_2:=\frac{a^2\phi^\prime}{\H}\,,\qquad f_3:=\frac{(f_1 f_2)^{\prime}}{f_1f_2}-\frac{(a^2f_1)^\prime}{a^2 f_1}\,.$$
The argument of \cite[Lemma 5.3.]{Eltzner:2013soa} now proceeds as follows. Assuming
\beq\label{eq_eltzner1}
\left[\Psi(\tau,\vec{x}),\Psi(\tau,\vec{y})\right]=0=\left[\Psi^\prime(\tau,\vec{x}),\Psi^\prime(\tau,\vec{y})\right],
\eeq
one obtains via \eqref{eq_eqtccrmu} and \eqref{eq_noncommeqns} (a) and (b)
$$\frac{i\delta(\vec{x}-\vec{y})}{a^2}\1=\left[\mu(\tau,\vec{x}),\mu^\prime(\tau,\vec{y})\right]=-\frac{1}{f^2_1 f^2_2}\Delta_{\vec{y}}\left[\Psi(\tau,\vec{y}),\Psi^\prime(\tau,\vec{x})\right]\,.$$
which implies
\beq\label{eq_eltzner2}\left[\Psi(\tau,\vec{y}),\Psi^\prime(\tau,\vec{x})\right]=-\frac{ia^2}{4\pi f_1^2 f_2^2}\frac{1}{|\vec{x}-\vec{y}|}\1\eeq
if one assumes that this commutator vanishes at spatial infinity. Clearly, this shows that $\Psi$ does not satisfy local commutation relations.

In order to proceed with our argument regarding the interpretation of the local non-commutativity of $\Psi$, we derive this result in a slightly different way. Namely, using the equations of motion in the form \eqref{eq_noncommeqns} (d) and (e), one finds that the covariant commutator of $\Psi$ is of the form

\beq\label{eq_noncommfull}\left[\Psi(x),\Psi(y)\right]=i(G_1^\Psi(x,y)+G_2^\Psi(x,y)+G_3^\Psi(x,y))\1\,,\eeq
$$G_1^\Psi(\tau_x,\vec{x},\tau_y,\vec{y}):=f_1(\tau_x)f_1(\tau_y)\int\limits^{\tau_x}_{\tau_0}\int\limits^{\tau_y}_{\tau_0}d\tau_1d\tau_2\,
f_2(\tau_1)f_2(\tau_2)G^\mu(\tau_1,\vec{x},\tau_2,\vec{y})\,,$$
$$G_2^\Psi(\tau_x,\vec{x},\tau_y,\vec{y}):=f_1(\tau_x)f_1(\tau_y)\left(\int\limits^{\tau_x}_{\tau_0}d\tau_1\,
f_2(\tau_1)\widetilde{G}_2^\Psi(\tau_1,\vec{x},\vec{y})-\int\limits^{\tau_y}_{\tau_0}d\tau_1\,
f_2(\tau_1)\widetilde{G}_2^\Psi(\tau_1,\vec{y},\vec{x})\right)\,,$$
$$\Delta_{\vec{y}}\,\widetilde{G}_2^\Psi(\tau_x,\vec{x},\vec{y})\1:=f_2(\tau_0)\left[\mu(\tau_x,\vec{x}),\mu^\prime(\tau_0,\vec{y})-f_3(\tau_0)\mu(\tau_0,\vec{y})\right]\,,$$
$$\Delta_{\vec{x}}\Delta_{\vec{y}}\frac{G_3^\Psi(\tau_x,\vec{x},\tau_y,\vec{y})}{f_1(\tau_x)f_1(\tau_y)}\1=f_2(\tau_0)^2\left[\mu^\prime(\tau_0,\vec{x})-f_3(\tau_0)\mu(\tau_0,\vec{x}),\mu^\prime(\tau_0,\vec{y})-f_3(\tau_0)\mu(\tau_0,\vec{y})\right]=0\,.$$
Assuming the vanishing of $G_3^\Psi(\tau_x,\vec{x},\tau_y,\vec{y})$ at spatial infinity, one finds that this contribution to the commutator vanishes. The remaining contributions $G_1^\Psi$ and $G_2^\Psi$ display two kinds of local non-commutativity behaviour. Whereas $G_2^\Psi$ is the result of integrating the initial conditions \eqref{eq_eltzner1} and \eqref{eq_eltzner2} at $\tau=\tau_0$ (by means of the normally hyperbolic equation satisfied by $\Psi$ which may be inferred from \eqref{eq_noncommeqns}), and is thus strictly non-local, $G_1^\Psi(\tau_1,\vec{x},\tau_2,\vec{y})$ still vanishes if $|\vec{x}-\vec{y}|$ is sufficiently large. Note that this implies that the assumptions \eqref{eq_eltzner1} are not met at arbitrary $\tau$, but only at $\tau=\tau_0$. Certainly this does not invalidate the qualitative local non-commutativity result of  \cite{Eltzner:2013soa}, it just shows that \eqref{eq_eltzner1} and \eqref{eq_eltzner2} do not capture the full local non-commutativity of $\Psi$.

By deriving this local non-commutativity in the form \eqref{eq_noncommfull}, we arrive at natural interpretation of this feature. Indeed \eqref{eq_noncommeqns} (d) and (e) clearly display that $\Psi$ is a  functional of $\mu$ which is non-local in both time and in space. Thus it is no surprise that $\Psi$ does not commute at space-like separations upon quantization. Equivalently, one may say that the smeared quantum field $\Psi(f)$ is not contained in the algebra of local observables $\A^S=\A^\mu$ for general test functions $f$, but only for $f$ which are suitable derivatives of other test functions. However, we believe that one should not attribute any physical meaning to this non-locality, despite of $\Psi$ being gauge-invariant and thus a valid observable, because one may easily construct non-local observables out of local ones in a similar fashion as above in any local quantum field theory.


\section*{Acknowledgements}
The work of T.-P. H. is supported by a research fellowship 
of the Deutsche Forschungsgemeinschaft (DFG). Many symbolic computations have been performed with ``Ricci'', a Mathematica package for symbolic tensor calculus by J.~M.~Lee, available at [www.math.washington.edu/$\sim$lee/Ricci].


\appendix

\section{A few details}

\subsection{Mapping \texorpdfstring{$\xi\neq0$ to $\xi=0$}{xi<>0 to xi=0}}
\label{sec_EinsteinJordan}

The coupled system \eqref{eq_fullcoupled} for $(g,\phi)$ and non-minimal coupling $\xi\neq 0$ can me mapped to the simplified case $\xi= 0$ by redefining $(g,\phi)$ in the following way, which is a generalised conformal transformation \cite{Futamase:1987ua, Salopek:1988qh, Makino:1991sg}.
$$\widebar {g}:= \alpha g\qquad \frac{d\widebar {\phi}}{d\phi}=\frac{\sqrt{\beta}}{\alpha}\qquad \widebar {V}(\widebar {\phi}):= \frac{V(\phi)}{\alpha^2}$$
Here, as before, $\alpha = 1-\xi\phi^2$ and $\beta = 1+\kappa\xi\phi^2$, $\kappa=\left(6\xi-1\right)\xi$. With these definitions, the fields $(\widebar {g},\widebar {\phi})$, defining the ``Einstein frame'' in contrast to the original ``Jordan frame'' specified by $(g,\phi)$, satisfy the coupled Einstein-Klein-Gordon equations with $\xi=0$ and the potential $\widebar {V}$. Under our standing assumptions that $\alpha$ and $\beta$ don't vanish anywhere on $M$ (and thus have definite sign by continuity), $\widebar {\phi}$ can be integrated as
$$\widebar {\phi}=\frac{-\sqrt{\kappa} \,\text{asinh}(\sqrt{\kappa}\phi)+\sqrt{\xi+\kappa}\,\text{atanh}\left(\frac{\phi \sqrt{\xi+\kappa}}{\sqrt{\beta}}\right)}{\xi}\,.$$

For the geometric quantities and perturbations in FLRW spacetimes the transformation amounts to 
$$\widebar {a}:=\sqrt{\alpha} a\quad\Rightarrow\quad \widebar {\H}:= \H+\sqrt{\alpha}^\prime=\H-\frac{\xi\phi^\prime \phi}{\alpha}$$
$$\widebar {X}:= X-\frac{\xi\phi}{\alpha}\varphi\qquad\widebar {Y}:=Y\qquad \widebar {\varphi}:=\frac{\sqrt{\beta}}{\alpha}\varphi\,,$$
where $X\in\{A,B,D,E\}$ and $Y\in\{V_i,W_i,T_{ij}\}$.

\subsection{The full expressions for \texorpdfstring{$P$}{P}, \texorpdfstring{$\overline{P}$}{Pbar}, and \texorpdfstring{$\widetilde{P}$}{Ptilde} }

The full original form of the linearised equation of motion operator reads

\begin{gather}P:\Gamma(\V)\to\Gamma(\V)\qquad P=\begin{pmatrix}P_0 & P_2\\P_3 & P_1\end{pmatrix}\label{eq_originaleom}\\
(P_0 \gamma)_{ab}=\frac{1-\xi\phi^2}{4}\left(-\nabla^c\nabla_c \gamma_{ab} + 2 \nabla^c\nabla^{\phantom{c}}_{(a}\gamma_{b)c}-g_{ab}\nabla^c\nabla^d\gamma_{cd}-\nabla_a\nabla_b{\gamma}_{c}^{\phantom{c}c}+g_{ab}\nabla_c\nabla^c{\gamma}_{d}^{\phantom{c}d}+\right.\notag\\
\left. +g_{ab}R^{cd}\gamma_{cd}-R\gamma_{ab}\right)+
\left(\frac{1-4\xi}{4}(\nabla_c\phi)(\nabla^c\phi)+\frac12 V-\xi\phi(\nabla_c\nabla^c\phi)\right) \gamma_{ab}+\xi g_{ab}\gamma_{cd}\phi \nabla^c\nabla^d\phi+\notag\\
+\xi\phi(\nabla^c\phi)\left(\frac{1}{2}\nabla_c\gamma_{ab}-\nabla^{\phantom{2}}_{(a}\gamma_{b)c}+ g_{ab}\nabla^d\gamma_{cd}-\frac{1}{2}g_{ab}\nabla_c{\gamma}_{d}^{\phantom{c}d}\right)+\frac{4\xi-1}{4}g_{ab}\gamma_{cd}(\nabla^c\phi)\nabla^d\phi \notag\\
(P_2\varphi)_{ab}=\left\{g_{ab}\left(-\xi \phi \nabla_c\nabla^c+\frac{1-4\xi}{2}(\nabla^c\phi)\nabla_c+\frac12 \partial_\phi V+\frac{\xi}{2}R\phi-\xi(\nabla_c\nabla^c\phi)\right)+\right.\notag\\
\left.+\xi \phi \nabla_a\nabla_b+\left(2\xi-1\right)\left(\nabla_{(a}\phi\right)\nabla_{b)}+\xi(\nabla_a\nabla_b\phi)-\xi R_{ab}\phi\phantom{\frac12}\!\!\!\!\right\}\varphi\notag\\
P_3 \gamma = \left(\xi\phi \nabla^c\nabla^d-\xi\phi\nabla_a\nabla^a g^{cd}-\xi\phi R^{cd}+(\nabla^c\phi)\nabla^d-\frac12 (\nabla^a\phi)\nabla_a g^{cd}+(\nabla^c\nabla^d\phi)\right) \gamma_{cd}\notag\\
P_1\varphi=\left(-\nabla_c\nabla^c+\xi R + \partial^2_\phi V\right)\varphi\notag\,.
\end{gather}
We recall
$$\alpha := 1-\xi\phi^2\qquad\beta := 1+\kappa\phi^2\qquad \kappa := \left(6\xi-1\right)\xi\,.$$
With these abbreviations, the full form of the linearised equation of motion operator after the field redefinition is

\begin{equation}\label{eq_redefinedP}
\widebar {P}:\Gamma(\V)\to\Gamma(\V)\qquad \widebar {P}=\begin{pmatrix}\widebar {P}_0 & \widebar {P}_2\\\widebar {P}_3 & \widebar {P}_1\end{pmatrix}\end{equation}

$$(\widebar {P}_0 \theta)_{ab}=-\nabla_c\nabla^c \theta_{ab}+2\nabla^c \nabla_{(a} \theta_{b)c} - g_{ab} \nabla^c\nabla^d \theta_{cd}+\frac{4\xi\phi}{\alpha}(\nabla_{(a}\phi)\nabla^c \theta_{b)c}-\frac{2\xi\phi}{\alpha}(\nabla^{c}\phi)\nabla_c \theta_{ab}-$$$$-\frac{2}{\alpha}V\theta_{ab}-\frac{2\xi}{\beta\alpha}(\nabla_{a}\phi)(\nabla_{b}\phi)\theta_c^{\phantom{c}c}+\frac{4\xi(2-\alpha)}{\alpha^2}(\nabla^{c}\phi)(\nabla_{(a}\phi)\theta_{b)c}-\frac{4\xi^2\phi^2}{\alpha^2}(\nabla\phi)^2\theta_{ab}+\frac{4\xi\phi}{\alpha}(\nabla^{c}\nabla_{(a}\phi)\theta_{b)c}$$
 
$$\left(\widebar {P}_2 \zeta\right)_{ab} = \left(-(\nabla_{(a}\phi)\nabla_{b)}+\frac{\kappa\phi}{\beta}(\nabla_{a}\phi)(\nabla_{b}\phi)+\frac{1}{2}g_{ab}(\nabla^{c}\phi)\nabla_{c}+\right.$$$$\left.+\frac{4\xi\phi V + \alpha \partial_\phi V}{2\beta}g_{ab}-\frac{\kappa\phi}{2\beta}(\nabla\phi)^2g_{ab}\right)\zeta$$ 
 
$$\widebar {P}_3 \theta = 4\left(\frac{\xi\phi}{\alpha}\nabla^a\nabla^b+\frac{\xi\phi\left(2\xi\alpha-\alpha-2\xi\right)}{\beta\alpha^2}V g^{ab}+\frac{\xi\alpha-\alpha-2\xi}{2\beta\alpha}\partial_\phi V g^{ab}+\frac{\xi\phi}{2\beta}\partial^2_\phi V g^{ab}+\right.$$$$\left.+\frac{3\xi^2\phi\left(4\xi-\beta\right)}{\beta^2\alpha^2}\left(\nabla\phi\right)^2g^{ab}+\frac{4\xi+\alpha-4\xi\alpha}{\alpha^2}\left(\nabla^a\phi\right)\nabla^b-\frac{\xi}{\beta\alpha}\left(\nabla^c\phi\right)\nabla_c g^{ab}+\right.$$$$\left.+
\frac{\xi\phi\left(8\xi+\alpha-4\xi\alpha\right)}{\alpha^3}(\nabla^a\phi)(\nabla^b\phi)+\frac{4\xi+\alpha-4\xi\alpha}{\alpha^2}(\nabla^a\nabla^b\phi)\right)\theta_{ab}$$
 
$$\widebar {P}_1 \zeta = \left(-\nabla_a\nabla^a+\xi R + \frac{2\xi\left(1+3\xi\right)}{\beta}\phi\partial_\phi V+\frac{\alpha}{\beta}\partial^2_\phi V+\right.$$$$\left.+\frac{2\kappa\phi}{\beta}\left(\nabla^a\phi\right)\nabla_a-\frac{2\kappa\left(\kappa\phi^2-1\right)\left(\nabla\phi\right)^2}{\beta^2}\right)\zeta\,.$$ 

Finally, the full form of the gauge-fixed equation of motion operator $\widetilde{P}$ reads

\begin{equation}\label{eq_PTilde}
\widetilde{P}:\Gamma(\V)\to\Gamma(\V)\qquad \widetilde{P}=\begin{pmatrix}\widetilde{P}_0 & \widetilde{P}_2\\\widetilde{P}_3 & \widetilde{P}_1\end{pmatrix}
\end{equation}
$$(\widetilde{P}_0 \theta)_{ab}=-\nabla_c\nabla^c \theta_{ab}-2R_{a\phantom{b}b}^{\phantom{a}c\phantom{a}d}\theta_{cd}-\frac{8\xi^2\phi^2}{\beta\alpha}V\theta_{ab}-\frac{2\xi\phi}{\beta}\partial_\phi V \theta_{ab}+\frac{4\xi\phi}{\alpha}(\nabla_{(a}\phi)\nabla^c\theta_{b)c}-$$$$-\frac{2\xi}{\beta\alpha}(\nabla_a \phi)(\nabla_b \phi){\theta}_{c}^{\phantom{c}c}-\frac{2\xi\phi}{\alpha}(\nabla^c \phi)\nabla_c \theta_{ab}+\frac{2\left(4\xi+\alpha-4\xi\alpha\right)}{\alpha^2}(\nabla^c \phi)(\nabla_{(a} \phi)\theta_{b)c}-$$$$-\frac{2\xi\left(\alpha+2\beta \xi\phi^2\right)}{\beta\alpha^2}\left(\nabla\phi\right)^2 \theta_{ab}-\frac{2\xi\phi}{\alpha}g_{ab}(\nabla^d\phi)\nabla^c\theta_{cd}$$
$$\left(\widetilde{P}_2 \zeta\right)_{ab} =\left(\frac{\kappa\phi}{\beta}(\nabla_{a}\phi)(\nabla_{b}\phi)+\frac{\xi\phi}{\alpha}(\nabla\phi)^2g_{ab}+ (\nabla_a\nabla_b\phi)\right)\zeta$$
$$\widetilde{P}_3 \theta = 4\left(\frac{\xi\phi\left(2\xi\alpha-\alpha-2\xi\right)}{\beta\alpha^2}V g^{ab}+\frac{\xi\alpha-\alpha-2\xi}{2\beta\alpha}\partial_\phi V g^{ab}+\frac{\xi\phi}{2\beta}\partial^2_\phi V g^{ab}-\right.$$$$\left.-\frac{3\xi^2\phi\left(\beta-4 \xi\right)}{\beta^2\alpha^2}\left(\nabla\phi\right)^2g^{ab}+\frac{2\xi^2\phi^2}{\alpha^2}\left(\nabla^a\phi\right)\nabla^b-\frac{\xi}{\beta\alpha}\left(\nabla^c\phi\right)\nabla_c g^{ab}+\right.$$$$\left.+
\frac{\xi\phi\left(8\xi+\alpha-4\xi\alpha\right)}{\alpha^3}(\nabla^a\phi)(\nabla^b\phi)+\frac{4\xi+\alpha-4\xi\alpha}{\alpha^2}(\nabla^a\nabla^b\phi)\right)\theta_{ab}$$
$$\widetilde{P}_1 \zeta = \left(-\nabla_a\nabla^a+\frac{\xi\left(2-\alpha\right)}{\alpha} R + \frac{2\xi\left(\beta+\alpha+3\xi\alpha\right)}{\beta\alpha}\phi\partial_\phi V+\frac{\alpha}{\beta}\partial^2_\phi V+\right.$$$$\left.+\frac{12\xi^2\phi}{\beta\alpha}\left(\nabla^a\phi\right)\nabla_a+\frac{2\left(12 \beta \xi^2+\beta^2\alpha-\beta^2\xi\alpha+12\xi^2\alpha - 18\beta\xi^2 \alpha\right)\left(\nabla\phi\right)^2}{\beta^2\alpha^2}\right)\zeta\,.$$ 

\subsection{The linearised equations on FLRW backgrounds}
\label{sec_fullFLRW}

We display the scalar, vector and tensor parts of $P\Gamma$ for $\Gamma\in\Gamma_\infty(\V)$, expressed in terms of the gauge invariant components $\Phi$, $\Psi$, $\chi$, $X_i$, $T_{ij}$.

$$-2a^2 A(P\Gamma)=-\alpha \Delta\Phi+\xi\phi\Delta\chi+\left(3\alpha\H-3\xi\phi\phi^\prime\right)\Phi^\prime-\left(3\xi\H\phi+\frac{\phi^\prime}{2}\right)\chi^\prime-$$$$-a^2V\Psi+\left(3\xi\H^\prime\phi+\left(1-3\xi\right)\H\phi^\prime+\frac{\phi^{\prime\prime}}{2}\right)\chi$$
$$-a^2 B(P\Gamma)=-\alpha \Phi^\prime+\left(\partial_\tau - \H\right)\xi\phi\chi+\left(\alpha\H-\xi\phi\phi^\prime\right)\Psi-\frac{\phi^\prime}{2}\chi$$
$$2a^2D(P\Gamma)=-\alpha \Phi^{\prime\prime}+\xi\phi \chi^{\prime\prime}+\frac{\alpha}{2} \Delta\left(\Psi+ \Phi\right)-\xi\phi \Delta \chi+\left(\alpha\H-\xi\phi\phi^\prime\right)\Psi^\prime+\left(2\xi\phi\phi^\prime-2\alpha\H\right)\Phi^\prime+$$
$$+\left(\xi\phi\H+\frac{4\xi-1}{2}\phi^\prime\right)\chi^\prime+\frac{1}{2}\left(2\alpha\H^2+4\alpha\H^\prime-4\H\xi\phi\phi^\prime+(1-4\xi)(\phi^\prime)^2-4\xi\phi\phi^{\prime\prime}\right)\left(\Psi-\Phi\right)+$$
$$+a^2V\Phi-\left(2\xi\H^2\phi+\xi\H^\prime\phi+(1-\xi)\H\phi^\prime+\frac{1-2\xi}{2}\phi^{\prime\prime}\right)\chi=0$$
$$2a^2 E(P\Gamma)=\frac{1}{2}\left(2\xi\phi\chi-\alpha\left(\Psi+\Phi\right)\right)$$
$$a^2\varphi(P\Gamma)=6\xi\phi \Phi^{\prime\prime}+\chi^{\prime\prime}-2\xi\phi\Delta(\Psi+2\Phi)-\Delta\chi-\left(6\xi\H\phi+\phi^\prime\right)(\Psi^\prime-3\Phi^\prime)+2\H\chi^\prime-$$$$-2\left(6\xi\H^2\phi+6\xi\H^\prime\phi+2\H\phi^\prime+\phi^{\prime\prime}\right)\Psi+\left(6\xi\H^2+6\xi\H^\prime+a^2V^{\prime\prime}\right)\chi$$
$$a^2 V_i(P\Gamma)=\frac{\alpha}{4}\Delta X_i\qquad a^2 W_i(P\Gamma)=(\partial_\tau+2\H)\frac{\alpha}{4}X_i$$
$$2a^2 T_{ij}(P\Gamma)=\frac{1}{2}\left((\partial_\tau+2\H)\alpha\partial_\tau-\alpha \Delta \right)T_{ij}$$

Due to gauge invariance, i.e. $K^\dagger P=0$, in case of $P\Gamma=0$ the two scalar equations with second time derivatives $\varphi(P\Gamma)=0$ and $D(P\Gamma)=0$ follow from the remaining three scalar equations.

\subsection{Hyperbolic operators and splittings}
\label{sec_splithyp}

We prove a few important results regarding the interplay of hyperbolic operators and the splitting into scalar, vector and tensor pieces on FLRW backgrounds. 

\begin{propo}
\label{prop_splithyp} Let $(M,g)$ be a globally hyperbolic FLRW spacetime, let $\V:= \bigvee^2 T^*M \oplus \left(M\times \bbR\right)$ and $\W := TM$ and let $P^\V:\Gamma(\V)\to\Gamma(\V)$ and $P^\W:\Gamma(\W)\to\Gamma(\W)$ be Cauchy-hyperbolic operators, i.e. $P^\V$ and $P^\W$ have a well-posed Cauchy problem and thus in particular unique advanced and retarded Green's operators $G^{P^\V}_\pm$ and $G^{P^\W}_\pm$ for $P^\V$ and $P^\W$ respectively exist, cf. \cite{Bar2}. Moreover, let $G^{P^\V}:=G_+^{P^\V}-G_-^{P^\V}$ and $G^{P^\W}:=G_+^{P^\W}-G_-^{P^\W}$ and assume that $$\PP^{S/V/T}_\V\circ P^\V|_{\Gamma_\infty(\V)}=P^\V\circ \PP^{S/V/T}_\V\,,$$ $$\PP^{S/V}_\W\circ P^\W|_{\Gamma_\infty(\W)}=P^\W\circ \PP^{S/V}_\W\,,$$ where $\PP^{S/V/T}_\V$ and $\PP^{S/V}_\W$ are defined in  \eqref{eq_def_proj}. Then, the following statements hold.
\begin{enumerate}
\item If $\Gamma\in\Gamma^{S/V/T}_\sc(\V)$ satisfies $P^\V\Gamma=0$ then $\Gamma$ can be written as $\Gamma=G^{P^\V}h$ with $h\in\Gamma^{S/V/T}_0$. If $\varsigma\in\Gamma^{S/V}_\sc(\W)$ satisfies $P^\W\varsigma=0$ then $\varsigma$ can be written as $\varsigma=G^{P^\W}f$ with $f\in\Gamma^{S/V}_0(\W)$.
\item If $\Gamma_1\in\Gamma^{S/V/T}_\sc(\V)$, $P^\V \Gamma_2 = \Gamma_1$ can be solved by a $\Gamma_2\in\Gamma^{S/V/T}_\sc(\V)$. If $\varsigma_1\in\Gamma^{S/V}_\sc(\W)$, $P^\V \varsigma_2 = \varsigma_1$ can be solved by a $\varsigma_2\in\Gamma^{S/V}_\sc(\W)$.
\item If $P^\V h\in\Gamma^{S/V/T}_0(\V)$ and $h\in\Gamma_0(\V)$, then $h\in\Gamma^{S/V/T}_0(\V)$. If $P^\W f\in\Gamma^{S/V}_0(\W)$ and $f\in\Gamma_0(\W)$, then $f\in\Gamma^{S/V}_0(\W)$.
\end{enumerate}
\end{propo}
\begin{proof}
We consider the case of $P^\V$, the other one can be shown analogously. 
 
Proof of 1: The proof follows the proof of e.g. \cite[Theorem 3.5.]{Bar2} with a few modifications. To this avail, we split $\Gamma\in\Gamma^{S/V/T}_\sc(\V)$ with $P^\V\Gamma=0$ into a future and past part $\Gamma=\Gamma^++\Gamma^-$ as described at the beginning of Section \ref{sec_solform}. By choosing a partition of unity $\chi=\chi^++\chi^-$ which is constant in the spatial coordinates, we can achieve $\Gamma^\pm\in\Gamma^{S/V/T}_\sc(\V)$. We now set $h:=P^\V \Gamma^+$. Since $h=-P^\V \Gamma^-$ and $P^\V$ commutes with $\PP^{S/V/T}_\V$, $h\in\Gamma^{S/V/T}_0(\V)$. Moreover $G^{P^\V}h=G_+^{P^\V}h-G_-^{P^\V}h=G_+^{P^\V}P^\V\Gamma^++G_-^{P^\V}P^\V\Gamma^-=\Gamma^++\Gamma^-=\Gamma$, where $G_\pm^{P^\V}P^\V\Gamma^\pm=\Gamma^\pm$ is demonstrated in the proof of \cite[Theorem 3.5.]{Bar2}.

Proof of 2: As in the proof of the first part, we split $\Gamma_1=\Gamma^+_1+\Gamma^-_1$ into a future and past part such that $\Gamma^\pm_1\in \Gamma^{S/V/T}_\sc(\V)$. Then a possible solution of $P^\V\Gamma_2=\Gamma_1$ is $\Gamma_2:=G^{P^\V}_+\Gamma^+_1+G^{P^\V}_-\Gamma^-_1$ which is well-defined by e.g. \cite[Theorem 3.8]{Baernew}. Since $P^\V$ commutes with $\PP^{S/V/T}_\V$, the same holds for its inverses $G^{P^\V}_\pm$ and thus $\Gamma_2\in \Gamma^{S/V/T}_\sc(\V)$.

Proof of 3: If $P^\V h\in\Gamma^{S/V/T}_0(\V)$ and $h\in\Gamma_0(\V)$, then $(1-\PP^{S/V/T}_\V)P^\V h=P^\V(1-\PP^{S/V/T}_\V)h=0$, where $(1-\PP^{S/V/T}_\V)h$ has compact support in time. However, such solutions to $P^\V\Gamma=0$ do not exist (see e.g. \cite[Corollary 3.9]{Baernew}) due to the Cauchy-hyperbolicity of $P^\V$. Thus $(1-\PP^{S/V/T}_\V)h=0$ which implies $h\in\Gamma^{S/V/T}_0(\V)$.
\end{proof}

Although we shall not need this more general case in the present work, the above proof can be repeated omitting the assumption that all occurring sections have compact support in spatial directions.

\subsection{The two faces of scalar and tensor observables}

Here we prove part 1 of Theorem \ref{prop_specialsplitrelations}.

\begin{theo}
\label{prop_surjinj} On FLRW backgrounds, $F^{-1}\circ G^{\widetilde P}:\Gamma_0(\V)\to\Gamma(\V)$ induces bijective maps $F^{-1}\circ G^{\widetilde P}:\E^{S/V/T}\to \Sol^{S/V/T}_\sc/\G^{S/V/T}_\sc$. In particular
\begin{enumerate}
\item (Surjectivity) Every $\Gamma\in \Sol^{S/V/T}_\sc$ 
can be split as $\Gamma=\Gamma_1+\Gamma_2$ with $\Gamma_2\in \G^{S/V/T}_\sc$ and $\Gamma_1\in \left(F^{-1}\circ G^{\widetilde{P}}\right)\left[\Ker^{S/V/T}_0(K^\dagger)\right]$.
\item (Injectivity) $h\in\Ker^{S/V/T}_0(K^\dagger)$ and $F^{-1}\circ G^{\widetilde{P}}h\in\G^{S/V/T}_{\sc}$ if and only if $h\in{P}\left[\Gamma^{S/V/T}_0(\V)\right]$.
\end{enumerate}
\end{theo}
\begin{proof}
 Since $\E^V=\{0\}$ and $\Sol^V/\G^V_\sc=\{0\}$, we only need to prove the scalar/tensor case.
To this avail, we recall a few identities from Section \ref{sec_quantgen}.
$$\widebar {P}=P\circ F^{-1}\,,\quad \widebar {K}=F\circ K\,,\quad \widebar {K}^\dagger=K^\dagger\,,\quad T=\frac{2}{\alpha}\widebar  K\,,\quad R=\widebar {K}^\dagger\circ \widebar {K}=K^\dagger\circ F\circ K$$$$Q=\widebar {K}^\dagger\circ T=K^\dagger \circ T\,,\quad\widetilde{P}=\widebar {P}+T\circ\widebar {K}^\dagger=P\circ F^{-1}+T\circ K^\dagger\,,$$
where $Q$, $R$ and $\widetilde P$ are (multiples of) normally hyperbolic operators and thus have a well-defined Cauchy problem. In particular, unique advanced and retarded propagators $G^Q_\pm$, $G^R_\pm$ and $G_\pm^{\widetilde P}$ exist, and we define the causal propagators $G^Q:=G^Q_+-G^Q_-$, $G^R:=G^R_+-G^R_-$ and $G^{\widetilde P}:=G^{\widetilde P}_+-G^{\widetilde P}_-$. Since $F$ is invertible, also $\widetilde P\circ F$ has a well-posed Cauchy problem, and the causal propagator of $\widetilde P\circ F$ is $F^{-1}\circ G^{\widetilde P}$. Moreover it holds
$$K^\dagger\circ \widetilde P=Q\circ K^\dagger\,,\quad \widetilde P\circ \widebar {K}=T\circ R\qquad\Rightarrow\qquad K^\dagger \circ G^{\widetilde P}=G^{Q}\circ K^\dagger\,,\quad \widebar  K\circ G^R=G^{\widetilde P}\circ T\,,$$ where the implications can be shown as in \cite[Theorem 3.12]{HS}. Finally, since $P$, $K$ and $K^\dagger$ intertwine, respectively commute with, the scalar/tensor projectors on the appropriate spaces, cf. \eqref{eq_invK}, \eqref{eq_invKdagger}, \eqref{eq_invP}, and one can show that $F$ commutes with $\PP^{S/V/T}_\V$ as well, analogous relations hold for $T$, $R$, $Q$, $\widetilde P$, and the inverses $G^Q_\pm$, $G^R_\pm$, $G_\pm^{\widetilde P}$.

Proof of 1: To see surjectivity in the tensor case, we recall $F[\Gamma^{S/T}(\V)]=\Gamma^{S/T}(\V)$ and $K^\dagger[\Gamma^T_\infty(\V)]=\{0\}$, in particular $\Ker^{T}_0(K^\dagger)=\Gamma^T_0(\V)$. Thus $\widetilde PF\Gamma=P\Gamma+TK^\dagger F\Gamma=0$, and by Proposition \ref{prop_splithyp} there exists $h_1\in\Gamma^T_0(\V)$ s.t. $\Gamma=F^{-1}G^{\widetilde P}h_1$. In the scalar case, we solve $R\varsigma = K^\dagger F K\varsigma =K^\dagger F \Gamma$ for $\varsigma$. By Proposition \ref{prop_splithyp}, this is possible with $\varsigma \in\Gamma^S_\sc(\W)$. We set $\Gamma_3=K\varsigma\in\G^S_\sc$ and obtain by construction $\widetilde PF(\Gamma-\Gamma_3)=0$. Thus, by Proposition \ref{prop_splithyp} there exists $h_2\in\Gamma^S_0(\V)$ s.t. $\Gamma-\Gamma_3=F^{-1}G^{\widetilde P}h_2$. We compute $0=K^\dagger FF^{-1}G^{\widetilde P}h_2=K^\dagger G^{\widetilde P}h_2=G^Q K^\dagger h_2$, hence there exists $f\in\Gamma_0(\W)$ s.t. $K^\dagger h_2=Qf$, cf. \cite[Theorem 3.5.]{Bar2}. By Proposition \ref{prop_splithyp}, $f\in\Gamma^S_0(\W)$. We set $h_3:=h_2-Tf$ and obtain by construction $h_3\in\Ker^S_0(K^\dagger)$. We define $\Gamma_1:=F^{-1}G^{\widetilde P}h_3$ and it remains to show that $\Gamma_4:=\Gamma-\Gamma_3-\Gamma_1\in\G^S_\sc$. To see this, we compute $\Gamma_4=F^{-1}G^{\widetilde P}Tf=F^{-1}\widebar {K}G^Rf=K G^Rf$. This is in $\Gamma_\sc(\V)$ due to \cite[Theorem 3.5.]{Bar2} and in $\G^S_\sc$ because $K$ and $G^R$ preserve the scalar type.

Proof of 2: We start again with the tensor case. If $h_1=Ph_2$ for $h_2\in\Gamma^T(\V)$, then $K^\dagger h_1=K^\dagger Ph_2=0$ because $K^\dagger\circ P=0$ and $h_1=PF^{-1}Fh_2=(P\circ F^{-1}+T\circ K^\dagger)Fh_2=\widetilde P Fh_2$ due to $K^\dagger Fh_2=0$. Thus $F^{-1}G^{\widetilde P}h_1=0$. If instead $h_1\in\Ker^T_0(K^\dagger)=\Gamma^T_0(\V)$ and $F^{-1}G^{\widetilde P}h_1=0$, then there exists $h_3\in\Gamma_0(\V)$ s.t. $h_1=\widetilde P h_3$ by \cite[Theorem 3.5.]{Bar2}, and owing to Proposition \ref{prop_splithyp}, $h_3\in\Gamma^T_0(\V)$. Reversing the computation above, we find that $h_1=\widetilde P h_3=P F^{-1}h_3$, where $F^{-1}h_3\in\Gamma_0(\V)$ since $F^{-1}$ preserves the tensor type. To show the scalar case, we assume $h_1=Ph_2$ for $h_2\in\Gamma^S(\V)$. We note that $h_1\in\Ker^S_0(K^\dagger)$ and compute $F^{-1}G^{\widetilde P}h_1=F^{-1}G^{\widetilde P}(\widetilde P\circ F-T\circ K^\dagger\circ F)h_2=-F^{-1}G^{\widetilde P}TK^\dagger Fh_2$ $=F^{-1}\widebar {K}G^R K^\dagger Fh_2=K G^R K^\dagger Fh_2$. By the same arguments as in the proof of 1, this is an element of $\G^S_\sc$. Assuming instead that $h_1\in\Ker^S_0(K^\dagger)$ and $F^{-1}G^{\widetilde P}h_1=K \varsigma$ for some $\varsigma\in\Gamma^S_\sc(\W)$ (recall $\G^S_\sc=K[\Gamma^S_\sc(\W)]$), we may compute $R\varsigma=K^\dagger F K\varsigma=K^\dagger FF^{-1}G^{\widetilde P}h_1=K^\dagger G^{\widetilde P}h_1=G^Q K^\dagger h_1=0$. Thus by \cite[Theorem 3.5.]{Bar2} and Proposition \ref{prop_splithyp}, there exists $f\in\Gamma^S_0(\W)$ s.t. $\varsigma=G^R f$. We have $F^{-1}G^{\widetilde P}h_1=K \varsigma=K G^R f=F^{-1}\widebar {K}G^R f=F^{-1}G^{\widetilde P}Tf$, thus, by \cite[Theorem 3.5.]{Bar2} and Proposition \ref{prop_splithyp}, there exists $h_2\in\Gamma^S_0(\V)$ such that $h_1=\widetilde P h_2+Tf$. We may further compute $0=K^\dagger h_1=K^\dagger \widetilde P h_2+K^\dagger Tf=QK^\dagger h_2+Qf$, from which $f=-K^\dagger h_2$ follows by \cite[Theorem 3.5.]{Bar2}. Thus $h_1=\widetilde P h_2+Tf=\widetilde P h_2-TK^\dagger h_2=PF^{-1}h_2$ with $F^{-1}h_2\in \Gamma^S_0(\V)$.

\end{proof}

\section{List of symbols}
\label{sec_listsymbols}

We compile a list of the symbols used across various sections of the text along with their definition or a reference to their first appearance.\\\\

\begin{tabular}{p{5cm}p{8cm}}
$\langle \cdot, \cdot \rangle_{\V}$ & bilinear form on $\Gamma(\V)$ prior to field redefinition, \eqref{eq_bilinearV}\\
$\langle \cdot, \cdot \rangle_{\widebar\V}$ & bilinear form on $\Gamma(\V)$ after field redefinition, \eqref{eq_redefinedForm}\\
$\langle \cdot, \cdot \rangle_{\W}$ & bilinear form on $\Gamma(\W)$, \eqref{eq_formW}\\
$\langle \cdot, \cdot \rangle_{ {\Sol}}$ &  bilinear form on $ {\Sol}$, \eqref{eq_def_origquants}\\
$\langle \cdot, \cdot \rangle_{\widebar {\Sol}}$ &  bilinear form on $\widebar {\Sol}$, \eqref{eq_def_tau2}\\
$a$ & scale factor, Section \ref{sec_inflation}\\
$\alpha$ & $\alpha = 1-\xi\phi^2$\\
$A$ & scalar component of $\gamma_{ab}$, \eqref{eq_def_split}\\
\end{tabular}

\begin{tabular}{p{5cm}p{8cm}}
$\A$ & polynomial quantum field algebra corresponding to the presymplectic space $(\widebar\E,\widebar\sigma)\simeq(\E,\sigma)$, Section \ref{sec_quant}\\
$\A^{\mu}$ & polynomial quantum field algebra of $\mu$, $\A^{\mu}\simeq\A^S$, Section \ref{sec_scalsec} \\
$\A^{S/T}$ & polynomial quantum field algebra corresponding to the symplectic space $(\E^{S/T},\sigma)$\\
$b$ & scalar component of $k_{ab}$, \eqref{eq_def_split0}\\
$B$ & scalar component of $\gamma_{ab}$, \eqref{eq_def_split}\\
$\beta$ & $\beta = 1+\left(6\xi-1\right)\xi\phi^2$ \\
$c$ & scalar component of $k_{ab}$, \eqref{eq_def_split0}\\
\end{tabular}

\begin{tabular}{p{5cm}p{8cm}}
$C^\infty_\infty(M,{\cal T})$ & space of smooth functions with values in ${\cal T}$ which vanish at spatial infinity with all derivatives\\
$\gamma_{ab}$ &  metric perturbation prior to field redefinition, Section \ref{sec_eometc}\\
$\Gamma$ & tuple of field perturbations prior to redefinition $\Gamma=(\gamma_{ab},\varphi)^T$, Section \ref{sec_eometc}\\
$\Gamma(\V/\W)$ &  space of smooth sections of $\V$ resp. $\W$\\
$\Gamma_{0/\sc/\tc}(\V/\W)$ &  spaces of smooth sections of $\V$ resp. $\W$ of compact, spacelike compact, timelike compact support, Section \ref{sec_notations}\\
\end{tabular}

\begin{tabular}{p{5cm}p{8cm}}
$\Gamma_\infty(\V/\W)$ & space of smooth sections which vanish at spatial infinity with all derivatives\\
$\Gamma^{S/V/T}(\V/\W)$ & space of smooth sections of scalar/vector/tensor type\\
$\Gamma^{S/V/T}_{\infty/\sc/0}(\V/\W)$ & $\Gamma^{S/V/T}_{\infty/\sc/0}(\V/\W)= \Gamma_{\infty/\sc/0}(\V/\W)\cap\Gamma^{S/V/T}(\V/\W)$\\

$d$ & scalar component of $k_{ab}$, \eqref{eq_def_split0}\\
$D$ & scalar component of $\gamma_{ab}$, \eqref{eq_def_split}\\
$\Delta$ & Laplace operator on $\bbR^3$\\
$e$ & scalar component of $k_{ab}$, \eqref{eq_def_split0}\\
\end{tabular}

\begin{tabular}{p{5cm}p{8cm}}
$E$ & scalar component of $\gamma_{ab}$, \eqref{eq_def_split}\\
${\E}$ & ${\E}= \Ker_0( {K}^\dagger)/{ {P}}[\Gamma_0(\V)]$ \\
$\widebar{\E}$ & $\widebar{\E}= \Ker_0(\widebar {K}^\dagger)/{\widebar {P}}[\Gamma_0(\V)]$ \\
$\E^{S/V/T}$ & $\E^{S/V/T}= \left.\Ker^{S/V/T}_0(K^\dagger)\right/P\left[\Gamma^{S/V/T}_0(\V)\right]$\\
\end{tabular}

\begin{tabular}{p{5cm}p{8cm}}
$f$ & test function\\
$F$ & field redefinition operator, \eqref{eq_def_redefinedq2}\\
$\phi$ & (background) scalar field, Section \ref{sec_eometc}\\
$\varphi$ &  scalar field perturbation prior to field redefinition, Section \ref{sec_eometc}\\
$\Phi$ & Bardeen potential $\Phi= D-\H(B+E^\prime)$\\
\end{tabular}

\begin{tabular}{p{5cm}p{8cm}}
$g_{ab}$ & smooth Lorentzian (background) metric on $M$, Section \ref{sec_notations}\\
$\GG$  & tuple of background fields $\GG=(g,\phi)^T$, Section \ref{sec_eometc}\\
${\G}$ & $ {\G}= {K}\left[\Gamma(\W)\right]$\\
$\widebar {\G}$ & $\widebar {\G}=\widebar {K}\left[\Gamma(\W)\right]$\\
\end{tabular}

\begin{tabular}{p{5cm}p{8cm}}
${\G}_\text{sc}$ & ${\G}_\text{sc}= { {\G}}\cap \Gamma_\sc(\V)$\\
$\widebar{\G}_\text{sc}$ & $\widebar{\G}_\text{sc}= {\widebar {\G}}\cap \Gamma_\sc(\V)$\\
${\G}_{\sc,0}$ & $ {\G}_{\sc,0}= {K}[\Gamma_\sc(\W)] $\\
$\widebar {\G}_{\sc,0}$ & $\widebar {\G}_{\sc,0}=\widebar {K}[\Gamma_\sc(\W)] $\\
$\G_\infty$ & $\G_\infty=\G\cap \Gamma_\infty(\V)$\\
$\G^{S/V}_{(\infty/\sc)}$ & $\G^{S/V}_{(\infty/\sc)}=\G\cap \Gamma_{(\infty/\sc)}^{S/V}(\V)$\\
$\G_{(\infty/\sc)}^{T}$ & $\G_{(\infty/\sc)}^{T}=\{0\}$\\
\end{tabular}

\begin{tabular}{p{5cm}p{8cm}}
$\G^{S/V}_{\sc,0}$ & $\G^{S/V}_{\sc,0}=K\left[\Gamma^{S/V}_\sc(\W)\right]$\\
$G^{\widetilde{P}}$ & causal propagator of $\widetilde{P}$, \eqref{eq_def_G}\\
$h$ & $h=(k_{ab},f)^T$ compactly supported section of $\V$, Section \ref{sec_quantization}\\
$\H$ & $\H=(\partial_\tau a) / a$\\
$\chi$ & gauge-invariant scalar field perturbation $\chi= \varphi - \phi^\prime(B+E^\prime)$\\
$j$ & conserved current prior to field redefinition, \eqref{eq_origcurrent}\\
\end{tabular}

\begin{tabular}{p{5cm}p{8cm}}
$\widebar j$ & conserved current after field redefinition, \eqref{eq_current}\\
$k_{ab}$ & metric component of compactly supported section of $\V$, Section \ref{sec_quantization}\\
$K$ & gauge transformation operator prior to field redefinition, \eqref{eq_gaugetrafosorig}\\
$\widebar K$ & gauge transformation operator after field redefinition, \eqref{eq_redefinedK}\\
$\Ker_0( {K}^\dagger)$ & $\Ker_0( {K}^\dagger)=\{h\in\Gamma_0(\V)\,|\, {K}^\dagger h=0\}$\\
$\Ker_0(\widebar {K}^\dagger)$ & $\Ker_0(\widebar {K}^\dagger)=\{h\in\Gamma_0(\V)\,|\,\widebar {K}^\dagger h=0\}$\\
\end{tabular}

\begin{tabular}{p{5cm}p{8cm}}
$\Ker^{S/V/T}_0(K^\dagger)$ & $\Ker^{S/V/T}_0(K^\dagger)=\Ker_0(K^\dagger)\cap \Gamma_0^{S/V/T}(\V)$\\
$M$  & four-dimensional smooth manifold, Section \ref{sec_notations}\\
$(M,g)$ & four-dimensional globally hyperbolic spacetime, Section \ref{sec_notations}\\
$\mu$ & Mukhanov-Sasaki variable, \eqref{eq_defmu}\\
$n$ & forward-pointing unit normal vector field on $\Sigma$\\
$P$ & linearised Einstein-Klein-Gordon operator prior to field redefinition, \eqref{eq_originaleom}\\
$\widebar P$ & linearised Einstein-Klein-Gordon operator after field redefinition, \eqref{eq_redefinedP}\\
\end{tabular}

\begin{tabular}{p{5cm}p{8cm}}
$\widetilde P$ & gauge-fixed linearised Einstein-Klein-Gordon operator, \eqref{eq_PTilde}\\
$P^T$ & equation of motion operator for $T_{ij}$, \eqref{eq_EOMten}\\
$P^\mu$ & equation of motion operator for $\mu$, \eqref{eq_EOMmu1}\\
$\PP^{S/V/T}_\V$ & projectors $\PP^{S/V/T}_\V:\Gamma_\infty(\V)\to\Gamma^{S/V/T}_\infty(\V)$\\
$\PP^{S/V}_\W$ & projectors $\PP^{S/V}_\W:\Gamma_\infty(\W)\to\Gamma^{S/V}_\infty(\W)$\\
$\Psi$ & Bardeen potential $\Psi= A-(\partial_\tau + \H)(B+E^\prime)$\\
$R$ & Ricci curvature scalar, Section \ref{sec_notations}\\
\end{tabular}

\begin{tabular}{p{5cm}p{8cm}}
$r$ & scalar component of $\varsigma$, \eqref{eq_def_splitgauge}\\
$s$ & scalar component of $\varsigma$, \eqref{eq_def_splitgauge}\\
\end{tabular}

\begin{tabular}{p{5cm}p{8cm}}
${\Sol}$ & $ {\Sol}= \{\Gamma\in \Gamma(\V)\,|\, { {P}}\Gamma=0\}$\\
$\widebar {\Sol}$ & $\widebar {\Sol}= \{\Theta\in \Gamma(\V)\,|\, {\widebar {P}}\Theta=0\}$\\
${\Sol}_\text{sc}$ & ${\Sol}_\text{sc}=  {\Sol}\cap \Gamma_\sc(\V)$\\
$\widebar {\Sol}_\text{sc}$ &$\widebar {\Sol}_\text{sc}= \widebar {\Sol}\cap \Gamma_\sc(\V)$\\
$\Sol_\infty$ & $\Sol_\infty=\Sol\cap \Gamma_\infty(\V)$\\
$\Sol^{S/V/T}_{(\infty/\sc)}$  & $\Sol^{S/V/T}_{(\infty/\sc)}= \Sol\cap \Gamma_{(\infty/\sc)}^{S/V/T}(\V)$\\
\end{tabular}

\begin{tabular}{p{5cm}p{8cm}}
$\varsigma$ & vector field, $\varsigma\in\Gamma(\W)=\Gamma(T M)$\\
${\sigma}$ & presymplectic form on $ \E$, \eqref{eq_def_origquants}\\
$\widebar {\sigma}$ & presymplectic form on $\widebar \E$, \eqref{eq_def_tau}\\
$\Sigma$ & Cauchy surface of $(M,g)$\\

$t_{ij}$ & tensor component of $k_{ab}$, \eqref{eq_def_split0}\\
$T$ & gauge fixing operator $T=2/\alpha \widebar K$ \\
$T_{ij}$ & tensor component of $\gamma_{ab}$, \eqref{eq_def_split}\\
$\tau$ & conformal time, Section \ref{sec_inflation}\\
$\theta_{ab}$ & metric perturbation after field redefinition, Section \ref{sec_eometc}\\
$\Theta$ & tuple of field perturbations after field  redefinition $\Theta=(\theta_{ab},\zeta)^T$, Section \ref{sec_eometc}\\
\end{tabular}

\begin{tabular}{p{5cm}p{8cm}}
$v_i$ & vector component of $k_{ab}$, \eqref{eq_def_split0}\\
$V_i$ & vector component of $\gamma_{ab}$, \eqref{eq_def_split}\\
$V$  & potential for $\phi$\\
$\V$  & $\V= \bigvee^2 T^*M \oplus \left(M\times \bbR\right)$, Section \ref{sec_notations}\\
$w_i$ & vector component of $k_{ab}$, \eqref{eq_def_split0}\\
$W_i$ & vector component of $\gamma_{ab}$, \eqref{eq_def_split}\\
$\W$  & $\W = TM$, Section \ref{sec_notations}\\
\end{tabular}

\begin{tabular}{p{5cm}p{8cm}}
$X_i$ & gauge invariant vector perturbation $X_i=W_i^\prime-V_i$\\
$\xi$ & coupling to scalar curvature, Section \ref{sec_eometc}\\
$z_i$ & vector component of $\varsigma$, \eqref{eq_def_splitgauge}\\
$\zeta$ &  scalar field perturbation after field redefinition, Section \ref{sec_eometc}\\
\end{tabular}



\begin{thebibliography}{10}


\bibitem[AS70]{Araki:1970zza}
  H.~Araki and M.~Shiraishi,
  ``On Quasifree States Of The Canonical Commutation Relations (I),''
  Publ. Res. Inst. Math. Sci. {\bf 7}(1), 105–120 (1971).
  
\bibitem[BGP07]{Bar:2007zz} 
  C.~B\"ar, N.~Ginoux and F.~Pf\"affle,
  ``Wave equations on Lorenzian manifolds and quantization,''
  Zuerich, Switzerland: Eur. Math. Soc. (2007) 194 p
  [arXiv:0806.1036 [math.DG]].
    
\bibitem[BG11]{Bar2}
  C.~Bar and N.~Ginoux,
  ``Classical and Quantum Fields on Lorentzian Manifolds,''
  Springer Proc.\ Math.\  {\bf 17} (2011) 359
  [arXiv:1104.1158 [math-ph]].
  
\bibitem[Ba13]{Baernew}
    B\"ar C.,
    ``Green-hyperbolic operators on globally hyperbolic spacetimes,''
    [arXiv:1310.0738 [math-ph]].

\bibitem[Ba80]{Bardeen:1980kt}
  J.~M.~Bardeen,
  ``Gauge Invariant Cosmological Perturbations,''
  Phys.\ Rev.\ D {\bf 22} (1980) 1882.
  
\bibitem[BDS13]{Benini:2013tra}
  M.~Benini, C.~Dappiaggi and A.~Schenkel,
  ``Quantized Abelian principal connections on Lorentzian manifolds,''
  [arXiv:1303.2515 [math-ph]].
  
\bibitem[BDH13]{Benini:2013fia}
  M.~Benini, C.~Dappiaggi and T.~-P.~Hack,
  ``Quantum Field Theory on Curved Backgrounds -- A Primer,''
  Int.\ J.\ Mod.\ Phys.\ A {\bf 28} (2013) 1330023
  [arXiv:1306.0527 [gr-qc]].  
  
\bibitem[BDHS13]{Benini:2013ita}
  M.~Benini, C.~Dappiaggi, T.~-P.~Hack and A.~Schenkel,
  ``A C*-algebra for quantized principal U(1)-connections on globally hyperbolic Lorentzian manifolds,''
  [arXiv:1307.3052 [math-ph]].

\bibitem[BS04]{Bernal:2004gm} 
  A.~N.~Bernal and M.~S{\'a}nchez,
  ``Smoothness of time functions and the metric splitting of globally hyperbolic space-times,''
  Commun.\ Math.\ Phys.\  {\bf 257}, 43 (2005)
  [arXiv:gr-qc/0401112].
  
\bibitem[BS05]{Bernal:2005qf} 
  A.~N.~Bernal and M.~S{\'a}nchez,
  ``Further results on the smoothability of Cauchy hypersurfaces and Cauchy time functions,''
  Lett.\ Math.\ Phys.\  {\bf 77}, 183 (2006)
  [arXiv:gr-qc/0512095].
  
\bibitem[BHR04]{degenerateCCR}
E.~Binz, R.~Honegger and A.~Rieckers,
``Construction and uniqueness of the $C^\ast$-Weyl algebra over a general pre-symplectic space,''
J.~Math.~Phys.~{\bf 45}, 2885 (2004).

\bibitem[BD82]{Birrell:1982ix}
  N.~D.~Birrell and P.~C.~W.~Davies,
  ``Quantum Fields In Curved Space,''
  Cambridge, Uk: Univ. Pr. (1982) 340p  

\bibitem[BFR13]{Brunetti:2013maa}
  R.~Brunetti, K.~Fredenhagen and K.~Rejzner,
  ``Quantum gravity from the point of view of locally covariant quantum field theory,''
  [arXiv:1306.1058 [math-ph]].
  
\bibitem[Di92]{DimockVector}
J.~Dimock, 
``Quantized Electromagnetic Field on a Manifold,''
 Rev.~Math.~Phys.~{\bf 4}, 223-233 (1992). 

\bibitem[EMM12]{Ellis}
  G.~F.~R.~Ellis, R.~Maartens and M.~A.~H.~MacCallum,
  ``Relativistic Cosmology,''
  Cambridge, Uk: Univ. Pr. (2012)

\bibitem[El13]{Eltzner:2013soa}
  B.~Eltzner,
  ``Quantization of Perturbations in Inflation,''
  [arXiv:1302.5358 [gr-qc]].

\bibitem[FH12]{Fewster:2012bj} 
  C.~J.~Fewster and D.~S.~Hunt,
  ``Quantization of linearised gravity in cosmological vacuum spacetimes,''
   Rev.\ Math.\ Phys.\  {\bf 25} (2013) 1330003
  [arXiv:1203.0261 [math-ph]].
    
\bibitem[FR11]{Fredenhagen:2011mq} 
  K.~Fredenhagen and K.~Rejzner,
  ``Batalin-Vilkovisky formalism in perturbative algebraic quantum field theory,''
  Commun.\ Math.\ Phys.\  {\bf 317} (2013) 697
  [arXiv:1110.5232 [math-ph]].
  
\bibitem[FM89]{Futamase:1987ua}
  T.~Futamase and K.~-i.~Maeda,
  ``Chaotic Inflationary Scenario in Models Having Nonminimal Coupling With Curvature,''
  Phys.\ Rev.\ D {\bf 39} (1989) 399.  
  
\bibitem[Gi09]{Ginoux}
N.\ Ginoux,``Linear wave equations'', in: C.\ B\"ar and K.\ Fredenhagen (eds.)
``Quantum field theory on curved spacetimes'', Springer, Berlin Heidelberg (2009)  

\bibitem[Ha96]{Haag}
R.\ Haag, \emph{Local quantum physics},
Springer, Berlin Heidelberg (1996)

\bibitem[HS13]{HS}
  T.~-P.~Hack and A.~Schenkel,
  ``Linear bosonic and fermionic quantum gauge theories on curved spacetimes,''
  Gen.\ Rel.\ Grav.\  {\bf 45} (2013) 877
  [arXiv:1205.3484 [math-ph]].
  
\bibitem[Ho07]{Hollands:2007zg}
  S.~Hollands,
  ``Renormalized Quantum Yang-Mills Fields in Curved Spacetime,''
  Rev.\ Math.\ Phys.\  {\bf 20} (2008) 1033
  [arXiv:0705.3340 [gr-qc]].  
  
\bibitem[HW14]{Hollands:2014eia}
  S.~Hollands and R.~M.~Wald,
  ``Quantum fields in curved spacetime,''
  [arXiv:1401.2026 [gr-qc]].  
  
\bibitem[H\"o90]{Hormander}
L.~H\"ormander,
{\it The Analysis of Linear Partial Differential Operators II}
Spinger (1990).  
  
\bibitem[Kh12]{Khavkine:2012jf}
  I.~Khavkine,
  ``Characteristics, Conal Geometry and Causality in Locally Covariant Field Theory,''
  [arXiv:1211.1914 [gr-qc]]. 
  
\bibitem[Kh14]{Khavkine:2014kya}  
  I.~Khavkine, ``Covariant phase space, constraints, gauge and the Peierls formula,''
  Int.\ J.\ Mod.\ Phys.\ A {\bf 29} (2014) 5,  1430009
  [arXiv:1402.1282 [math-ph]].
  
\bibitem[MS91]{Makino:1991sg}
  N.~Makino and M.~Sasaki,
  ``The Density perturbation in the chaotic inflation with nonminimal coupling,''
  Prog.\ Theor.\ Phys.\  {\bf 86} (1991) 103.  
  
\bibitem[Mu05]{Mukhanov:2005sc}
  V.~Mukhanov,
  ``Physical foundations of cosmology,''
  Cambridge, UK: Univ. Pr. (2005) 421 p

\bibitem[MFB92]{Mukhanov:1990me}
  V.~F.~Mukhanov, H.~A.~Feldman and R.~H.~Brandenberger,
  ``Theory of cosmological perturbations. Part 1. Classical perturbations. Part 2. Quantum theory of perturbations. Part 3. Extensions,''
  Phys.\ Rept.\  {\bf 215} (1992) 203.
  
\bibitem[Mu07]{Mukhanov:2007zz}
  V.~Mukhanov and S.~Winitzki,
  ``Introduction to quantum effects in gravity,''
  Cambridge, UK: Cambridge Univ. Pr. (2007) 273 p
   
\bibitem[Ol07]{Olbermann:2007gn}
  H.~Olbermann,
  ``States of low energy on Robertson-Walker spacetimes,''
  Class.\ Quant.\ Grav.\  {\bf 24} (2007) 5011
  [arXiv:0704.2986 [gr-qc]].  
   
\bibitem[PT09]{Parker:2009uva}
  L.~Parker and D.~Toms,
  ``Quantum field theory in curved spacetime : Quantized field and gravity,''
  Cambridge, UK: Cambridge Univ. Pr. (2009) 455 p  
  
\bibitem[PS13]{Pinamonti:2013zba}
  N.~Pinamonti and D.~Siemssen,
  ``Scale-Invariant Curvature Fluctuations from an Extended Semiclassical Gravity,''
  [arXiv:1303.3241 [gr-qc]].  
  
\bibitem[SBB89]{Salopek:1988qh}
  D.~S.~Salopek, J.~R.~Bond and J.~M.~Bardeen,
  ``Designing Density Fluctuation Spectra in Inflation,''
  Phys.\ Rev.\ D {\bf 40} (1989) 1753.  
  
\bibitem[Sa12]{Sanders:2012ac}
  K.~Sanders,
  ``A note on spacelike and timelike compactness,''
  Class.\ Quant.\ Grav.\  {\bf 30} (2013) 115014
  [arXiv:1211.2469 [math-ph]].  
  
\bibitem[SDH12]{Sanders:2012sf}
  K.~Sanders, C.~Dappiaggi and T.~-P.~Hack,
  ``Electromagnetism, local covariance, the Aharonov-Bohm effect and Gauss' law,''
  [arXiv:1211.6420 [math-ph]].  
  
\bibitem[SW74]{Stewart:1974uz}
  J.~M.~Stewart and M.~Walker,
  ``Perturbations of spacetimes in general relativity,''
  Proc.\ Roy.\ Soc.\ Lond.\ A {\bf 341} (1974) 49.
  
\bibitem[St05]{Straumann:2005mz}
  N.~Straumann,
  ``From primordial quantum fluctuations to the anisotropies of the cosmic microwave background radiation,''
  Annalen Phys.\  {\bf 15} (2006) 701
  [arXiv:hep-ph/0505249].  

\bibitem[Wa84]{Wald:1984}
R.~M.~Wald,
``General Relativity,'' 
Chicago University Press (1984) .

\bibitem[Wa95]{Wald:1995yp}
  R.~M.~Wald,
  ``Quantum field theory in curved space-time and black hole thermodynamics,''
  Chicago, USA: Univ. Pr. (1994) 205 p

  
\end{thebibliography}
\end{document}